\documentclass[a4paper,UKenglish,cleveref, autoref, thm-restate]{lipics-v2021}

\pdfoutput=1
\hideLIPIcs

\usepackage{booktabs}
\usepackage{amsthm}
\usepackage{textcomp}
\usepackage{xcolor}
\usepackage{footmisc}
\usepackage{ifthen}

\newif\iffullversion
\fullversiontrue

\newcommand{\fv}[2]{\iffullversion#1\else#2\fi}
\newcommand{\fvonly}[1]{\iffullversion#1\fi}
\newcommand{\cronly}[1]{\iffullversion\else#1\fi}

\newcommand{\fullver}{the full version~\cite{chen2024samm}}

\usepackage{xspace}
\usepackage{url}
\usepackage{subcaption}
\usepackage{mathtools}
\usepackage{thmtools}
\usepackage{graphicx}
\usepackage{siunitx}
\newcommand{\SingleAction}{\mathcal{A}_{\mathbf{1}}}
\newcommand{\ActionSpaceBA}{\mathcal{A}^{\textit{BA}}}
\newcommand{\SmallestAction}{\mathcal{A}_{\mathbf{1}, \min}}
\newcommand{\SmallestNum}{n_{\min}}
\newcommand{\poolcontent}[3]{%
    \ifthenelse{\equal{#2}{}}
    {\ifthenelse{\equal{#3}{}}
        {\ifthenelse{\equal{#1}{}}
            {\ensuremath{ R }\xspace}
            {\ensuremath{ R_{#1} }\xspace}
        }
        {\ifthenelse{\equal{#1}{}}
            {\ensuremath{ R^{#3} }\xspace}
            {\ensuremath{ R^{#3}_{#1} }\xspace}
        }
    }
    {\ifthenelse{\equal{#3}{}}
        {\ifthenelse{\equal{#1}{}}
            {\ensuremath{ R({#2}) }\xspace}
            {\ensuremath{ R_{#1}({#2}) }\xspace}
        }
        {\ifthenelse{\equal{#1}{}}
            {\ensuremath{ R^{#3}({#2}) }\xspace}
            {\ensuremath{ R_{#1}^{#3}({#2}) }\xspace}
        }
    }
}
\newcommand{\poolcontentvector}[3]{%
    \ifthenelse{\equal{#2}{}}
    {\ifthenelse{\equal{#3}{}}
        {\ifthenelse{\equal{#1}{}}
            {\ensuremath{ \textbf{R} }\xspace}
            {\ensuremath{ \textbf{R}_{#1} }\xspace}
        }
        {\ifthenelse{\equal{#1}{}}
            {\ensuremath{ \textbf{R}^{#3} }\xspace}
            {\ensuremath{ \textbf{R}^{#3}_{#1} }\xspace}
        }
    }
    {\ifthenelse{\equal{#3}{}}
        {\ifthenelse{\equal{#1}{}}
            {\ensuremath{ \textbf{R}({#2}) }\xspace}
            {\ensuremath{ \textbf{R}_{#1}({#2}) }\xspace}
        }
        {\ifthenelse{\equal{#1}{}}
            {\ensuremath{ \textbf{R}^{#3}({#2}) }\xspace}
            {\ensuremath{ \textbf{R}_{#1}^{#3}({#2}) }\xspace}
        }
    }
}
\newcommand{\fillupac}[3]{%
    \ifthenelse{\equal{#2}{\empty}}{%
        \ifthenelse{\equal{#3}{\empty}}{%
            \ifthenelse{\equal{#1}{\empty}}{%
                \ensuremath{\hat{l}}\xspace
            }{%
                \ensuremath{\hat{l}_{#1}}\xspace
            }%
        }{%
            \ifthenelse{\equal{#1}{\empty}}{%
                \ensuremath{\hat{l}^{#3}}\xspace
            }{%
                \ensuremath{\hat{l}^{#3}_{#1}}\xspace
            }%
        }%
    }{%
        \ifthenelse{\equal{#3}{\empty}}{%
            \ifthenelse{\equal{#1}{\empty}}{%
                \ensuremath{\hat{l}({#2})}\xspace
            }{%
                \ensuremath{\hat{l}_{#1}({#2})}\xspace
            }%
        }{%
            \ifthenelse{\equal{#1}{\empty}}{%
                \ensuremath{\hat{l}^{#3}({#2})}\xspace
            }{%
                \ensuremath{\hat{l}_{#1}^{#3}({#2})}\xspace
            }%
        }%
    }%
}
\newcommand{\lpac}[3]{%
    \ifthenelse{\equal{#2}{\empty}}{%
        \ifthenelse{\equal{#3}{\empty}}{%
            \ifthenelse{\equal{#1}{\empty}}{%
                \ensuremath{l}\xspace
            }{%
                \ensuremath{l_{#1}}\xspace
            }%
        }{%
            \ifthenelse{\equal{#1}{\empty}}{%
                \ensuremath{l^{#3}}\xspace
            }{%
                \ensuremath{l^{#3}_{#1}}\xspace
            }%
        }%
    }{%
        \ifthenelse{\equal{#3}{\empty}}{%
            \ifthenelse{\equal{#1}{\empty}}{%
                \ensuremath{l({#2})}\xspace
            }{%
                \ensuremath{l_{#1}({#2})}\xspace
            }%
        }{%
            \ifthenelse{\equal{#1}{\empty}}{%
                \ensuremath{l^{#3}({#2})}\xspace
            }{%
                \ensuremath{l_{#1}^{#3}({#2})}\xspace
            }%
        }%
    }%
}
\newcommand{\uac}[3]{%
    \ifthenelse{\equal{#2}{\empty}}{%
        \ifthenelse{\equal{#3}{\empty}}{%
            \ifthenelse{\equal{#1}{\empty}}{%
                \ensuremath{b}\xspace
            }{%
                \ensuremath{b_{#1}}\xspace
            }%
        }{%
            \ifthenelse{\equal{#1}{\empty}}{%
                \ensuremath{b^{#3}}\xspace
            }{%
                \ensuremath{b^{#3}_{#1}}\xspace
            }%
        }%
    }{%
        \ifthenelse{\equal{#3}{\empty}}{%
            \ifthenelse{\equal{#1}{\empty}}{%
                \ensuremath{b({#2})}\xspace
            }{%
                \ensuremath{b_{#1}({#2})}\xspace
            }%
        }{%
            \ifthenelse{\equal{#1}{\empty}}{%
                \ensuremath{b^{#3}({#2})}\xspace
            }{%
                \ensuremath{b_{#1}^{#3}({#2})}\xspace
            }%
        }%
    }%
}
\newcommand{\poolinput}[3]{%
    \ifthenelse{\equal{#2}{\empty}}{%
        \ifthenelse{\equal{#3}{\empty}}{%
            \ifthenelse{\equal{#1}{\empty}}{%
                \ensuremath{I}\xspace
            }{%
                \ensuremath{I_{#1}}\xspace
            }%
        }{%
            \ifthenelse{\equal{#1}{\empty}}{%
                \ensuremath{I^{#3}}\xspace
            }{%
                \ensuremath{I^{#3}_{#1}}\xspace
            }%
        }%
    }{%
        \ifthenelse{\equal{#3}{\empty}}{%
            \ifthenelse{\equal{#1}{\empty}}{%
                \ensuremath{I({#2})}\xspace
            }{%
                \ensuremath{I_{#1}({#2})}\xspace
            }%
        }{%
            \ifthenelse{\equal{#1}{\empty}}{%
                \ensuremath{I^{#3}({#2})}\xspace
            }{%
                \ensuremath{I_{#1}^{#3}({#2})}\xspace
            }%
        }%
    }%
}
\newcommand{\pooloutput}[3]{%
    \ifthenelse{\equal{#2}{\empty}}{%
        \ifthenelse{\equal{#3}{\empty}}{%
            \ifthenelse{\equal{#1}{\empty}}{%
                \ensuremath{O}\xspace
            }{%
                \ensuremath{O_{#1}}\xspace
            }%
        }{%
            \ifthenelse{\equal{#1}{\empty}}{%
                \ensuremath{O^{#3}}\xspace
            }{%
                \ensuremath{O^{#3}_{#1}}\xspace
            }%
        }%
    }{%
        \ifthenelse{\equal{#3}{\empty}}{%
            \ifthenelse{\equal{#1}{\empty}}{%
                \ensuremath{O({#2})}\xspace
            }{%
                \ensuremath{O_{#1}({#2})}\xspace
            }%
        }{%
            \ifthenelse{\equal{#1}{\empty}}{%
                \ensuremath{O^{#3}({#2})}\xspace
            }{%
                \ensuremath{O_{#1}^{#3}({#2})}\xspace
            }%
        }%
    }%
}

\newenvironment{restate}[1]{\begin{trivlist} \item {\bf #1 (restated)}. 
    \em} {\end{trivlist}}

\iffullversion
\title{SAMM: Sharded Automated Market Maker\thanks{Published in the 8th Conference on Advances in Financial Technologies (AFT~2026). Please cite as~\protect\cite{chen2026samm}.}}
\else
\title{SAMM: Sharded Automated Market Maker}
\fi
\titlerunning{SAMM: Sharded Automated Market Maker}
\author{Hongyin Chen}{Technion, Haifa, Israel}{hongyin.chen.contact@gmail.com}{https://orcid.org/0000-0003-3443-095X}{}

\author{Amit Vaisman}{Technion, Haifa, Israel}{amit.vaisman@campus.technion.ac.il}{https://orcid.org/0009-0008-1649-278X}{}

\author{Ittay Eyal}{Technion, Haifa, Israel}{ittay@technion.ac.il}{https://orcid.org/0000-0001-7595-2258}{}

\authorrunning{H. Chen, A. Vaisman, and I. Eyal}



\begin{document}
\nolinenumbers
\maketitle

\begin{abstract}
The central challenge in addressing the security-performance trade-off in blockchain systems has moved from consensus protocols to enabling parallel execution.
Various \emph{smart contracts} (stateful blockchain programs) already benefit from parallelization.
But we show that non-parallelizable operations within smart contracts remain a bottleneck.
A natural approach is to apply \emph{sharding}, that is, partitioning a \emph{single logical service} into multiple identical parallelizable shards.
However, we observe that this introduces game-theoretic issues in various applications: Participants may be incentivized to overload the system.

We focus on \emph{Automated Market Makers} (\emph{AMMs}), which enable token exchange: Traders swap tokens with the AMM for a fee, while \emph{liquidity providers} pre-fund the AMM and receive these fees. 
We show experimentally that their demand will soon reach the single-thread performance limit.

We present \emph{SAMM}, a sharded AMM system where multiple independent AMM shards operate on the same blockchain.
The key challenge is ensuring SAMM's security against overloading and destabilization attacks.
SAMM addresses this with careful mechanism design; we provide both game-theoretic analysis and simulation-based validation.
We implement SAMM on Sui and Solana testnets, achieving~5x and~16x throughput improvements, respectively, with potential for more as blockchain parallelism advances. 
SAMM is deployable and has already attracted industry interest.

\end{abstract}


	\section{Introduction}\label{sec:intro}

The trade-off between security and performance is a key challenge in blockchain systems.
Recent advances in high-throughput consensus (e.g.~\cite{spiegelman2022bullshark,danezis2022narwhal,bagaria2019prism}) have largely shifted the throughput bottleneck from consensus to smart-contract execution~\cite{xu2025vegeta}.
Modern blockchains like Sui~\cite{blackshear2023sui}, Aptos~\cite{aptos2022}, and Solana~\cite{yakovenko2018solana} support parallel execution to address this challenge (though not Bitcoin~\cite{nakamoto2008bitcoin} or Ethereum~\cite{wood2014ethereum}).
They automatically distribute independent transactions across multiple execution threads.

Practical smart-contract deployments adopt \emph{sharding}~\cite{zamani2018rapidchain, kokoris2018omniledger, wang2019monoxide}, a classical throughput-scaling technique, to improve performance~(\S\ref{sec:related}).
They trivially separate independent components of a contract into multiple independent instances called \emph{shards}, and the underlying blockchain runs them in parallel threads.
However, non-parallelizable operations within high-demand smart contracts remain a bottleneck.
To the best of our knowledge, no prior work studies improving the throughput of non-parallelizable smart contracts via sharding.

In this work, we study the sharding of a \emph{single non-parallelizable logical service} into multiple \emph{identical} shards.
For sharding to be effective, participants' demand should be evenly distributed across shards, without incurring additional overhead from splitting a single activity across multiple shards.
But in decentralized systems we cannot assume participants behave benignly; security issues arise if participants deviate from the desired behavior.

We find that in several applications\fv{ (\S\ref{sec:contracts})}{, such as \emph{lending protocols}~\cite{morphoMarketDocs,kaminoU,aaveV3TechPaper} and \emph{centralized limit order books} (CLOBs)~\cite{deepbookV3SuiDocs,phoenixCli}}, rational participants may be incentivized to launch overloading attacks: splitting their activity across multiple shards can increase their revenue\cronly{ (see the full version of our report~\cite{chen2024samm})}.
This is not merely theoretical: In practice, users already split orders across multiple similar (but non-sharded) contracts to optimize revenue via so-called \emph{routers}~\cite{angeris2022optimal,zhou2021a2mm}.

The root cause of this security challenge is common in blockchain protocols~\cite{zhang2019lay,hou2019squirrl,mirkin2020bdos,carlsten2016instability,yaish2023uncle} and applications~\cite{tsabary2021mad,mccorry2019smart,daian2020flash,chen2026prrr,chen2026time,mouallem2026resilient}: 
They are open to anyone and participants attack the system for more revenue.
The resilience of the sharded contract stems from game-theoretic security~\cite{brugger2023checkmate}: Mechanism design prevents such attacks by eliminating additional revenue due to misbehavior.

Since different protocols have distinct incentive structures, mechanism design must be tailored to each protocol.
The non-parallelizable part of many popular smart contracts is less than $1$ \emph{transaction per second} (\textit{tps}; \fv{\S\ref{app:demand}}{see \fullver}), well below the single-thread performance limit.
Therefore, we focus on addressing the only urgent context we are aware of: \emph{Automated Market Makers}~(AMMs)~\cite{ghosh2023automated, goyal2023finding}.
AMMs allow users to immediately swap token pairs by maintaining \emph{liquidity pools}: tokens of both types supplied by other users serving as \emph{liquidity providers}. 
It requires sequential handling of transactions since the outcome of each transaction depends on the current state of the token-pair pool and, in turn, alters this state.

We show that AMM performance does not scale in two state-of-the-art blockchain systems, Sui~\cite{blackshear2023sui} and Solana~\cite{yakovenko2018solana}. 
The throughput is limited by a single thread at $214\textit{tps}$ in Sui (Figure~\ref{fig:multiple_latency}, ${n=1}$) and $129\textit{tps}$ in Solana (Figure~\ref{fig:samm_tps}).
However, AMMs' demand is growing exponentially and by~2029 neither system could satisfy the most popular token pairs' trading demand~\fv{(\S\ref{app:growingdemand})}{(see \fullver)}.
This encourages races and lengthens the queue for AMM trades, widening the time window in which adversaries can manipulate transaction ordering and thereby exacerbating front-running attacks such as sandwiching~\cite{zhou2021high,li2023demystifying}.
To overcome this, users defer workload to expansion systems called layer-2 protocols~\cite{uniswaplayer2, adams2024layer}; but this entails additional security assumptions.\footnote{\label{fnote:l2}Blockchains are secured by financially-invested nodes. 
A {layer-2} protocol relies on an additional set of nodes, with significantly less at stake~\cite{tvl2024}, 
making it vulnerable to cheaper attacks.} 
Moreover, the single-thread bottleneck still exists in these layer-2 systems (e.g., $138\textit{tps}$ in ammBoost~\cite{michel2025ammboost}).

We apply sharding to AMMs to improve their throughput.
We model the system~(\S\ref{sec:model}) as a set of AMM shards with rational users of two types: traders, who purchase tokens and seek to minimize expenses, and liquidity providers, who deposit tokens and earn~fees.

The shards are AMMs based on the standard \emph{Constant Product Market Maker}~(\emph{CPMM}) contract: 
Roughly, the contract maintains the product of the two tokens constant after each trade. 
Thus, purchasing a larger amount of a token increases its unit cost, an effect called \emph{slippage}. 
See \fv{Appendix~\ref{app:preliminaries}}{\fullver} for more details.

We present SAMM~(\S\ref{sec:samm}), an AMM protocol that uses multiple \emph{shards} operating on the same blockchain, each with an independent liquidity pool.
Ideally, all shards should be \emph{balanced}, with the same liquidity; and traders should randomly select a shard to trade.
Traders should not split trades across multiple shards, multiplying their overhead. 
Liquidity providers should not cause imbalances in shard sizes, making some shards less attractive to traders and reducing parallelism.
Additionally, the system should re-stabilize if attackers unbalance the shards.
Our goal is to disincentivize attacks that deviate from this ideal behavior, such as overloading the system by splitting trades across shards or concentrating trades in a subset of shards by unbalancing liquidity.

We prove that traditional CPMM shards are not secure.
They charge a fixed-rate fee (e.g.,~$0.3\%$ in Uniswap v2~\cite{adams2020uniswap}), which simply scales a trade’s cost by a constant factor relative to the no-fee trade.
Traders are incentivized to route a single trade across multiple shards to reduce slippage.

At the core of SAMM is a novel \emph{trading fee function}.
SAMM generalizes the fixed-rate fee to a non-constant rate, incentivizing traders to use the smallest shard and avoid splitting trades.
It is a simple monomial family, which is easy to parameterize and to analyze.

The model gives rise to a game~(\S\ref{sec:game}) played by the users:
In each step, either a liquidity provider adds liquidity to a subset of the shards, or a trader executes a trade using a subset of the shards. 
We assume myopic liquidity-provider behavior, under which the liquidity provider treats the revenue from the next trade as an estimation for long-term revenue.
This reduces the analysis to a Stackelberg game.
Myopic utility is a common assumption in prior work~\cite{roughgarden2021transaction,ferreira2021dynamic}, and we further justify it in Section~\ref{GT:game_model}.
Our analysis shows that, in all best responses, traders use one of the smallest shards.
This implies that when shards are balanced, traders will randomly select a shard to trade, as intended.
We also prove that shards converge (e.g., following intentional destabilization) to a balanced state in all Subgame-Perfect Nash Equilibria (SPNE) and present a specific SPNE, showing that SAMM is robust against irrational destabilization attacks.
Malicious attackers suffer ongoing losses, and when their budget runs out, the system returns to a balanced state.

\begin{figure}[t]
	\centering
	\includegraphics[width=0.95\textwidth]{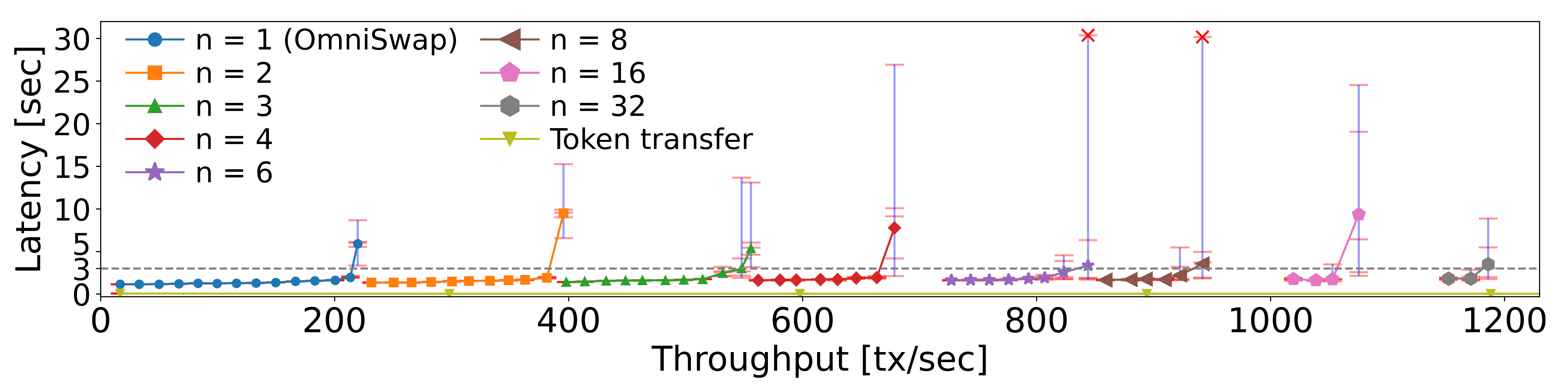}
	\caption{Sui trade latency as a function of demand with~$n$ SAMM shards.}
	\label{fig:multiple_latency}
\end{figure}

Before validating the theoretical analysis with a simulation, we turn our attention to the performance of SAMM, which will inform this simulation. 
We implement the protocol and deploy it in Sui and Solana local test networks (\S\ref{sec:evaluation}). 
This local validation is not an oversimplification: Our setting suffices to demonstrate the improvement in execution time, which is the bottleneck of these modern blockchains, as they process consensus separately from execution.
We test \textit{tps} and latencies of trades under different demand rates and numbers of shards.
In Sui, the throughput increase is limited by the serial elements of Sui's transaction processing, following Amdahl's Law.
In Solana, the throughput increases linearly with the number of shards up to a limit due to the blockchain's constraints. 
SAMM achieves over a 5x throughput increase in Sui and a~16x increase in Solana, compared to a single-contract AMM.
Figure~\ref{fig:multiple_latency} shows that with more shards in Sui, SAMM achieves higher throughput (X axis) with low average latency (Y axis).
Error bars show additional experiments, $\times$~marking failure due to overload.
We observe similar results in Solana~\fv{(\S\ref{app:solana})}{(see \fullver)}.
Improving the blockchains' parallelism would allow for even better performance. 

The Mysten Labs Team reproduced our results and intends to make SAMM part of their benchmark suite.\footnote{From personal communication.}\textsuperscript{,}\footnote{\label{footnote:interest}The authors are not affiliated with Mysten Labs, Solana, or RISE and have no financial interest in their tokens.}
The RISE Chain benchmark demonstrates a $10$x throughput improvement for multiple independent Uniswap instances. 
The team considers our design as necessary for scaling decentralized applications~\cite{risechainpevm}.\footref{footnote:interest}

Finally, we conduct a simulation to justify our theoretical assumptions (\S\ref{sec:simulation}).
We simulate traders that are free to split their trades using demand from real traces without synthetic arbitrage. 
Indeed, traders exhibit the behavior predicted by our theoretical analysis.
Furthermore, simulation results quantify the economic impact of SAMM's throughput gain.
Splitting liquidity across shards worsens per-trade slippage, but each shard now processes more trades.
As a result, a lower per-trade fee still yields substantially higher aggregate revenue for liquidity providers, and the lower fee offsets the added slippage, keeping traders' costs comparable to a standard CPMM (and lower under some parameters).
The same fee-slippage trade-off also lets SAMM absorb large trades (may split across shards), with smaller price impact than a single-pool CPMM.
Hence, SAMM is attractive to both liquidity providers and traders when competing with other standard AMMs.
Beyond the simulation, we also prove that, although split liquidity increases slippage, it neither worsens sandwich attacks nor harms liquidity providers when token prices fluctuate. 
Moreover, price fluctuations add only a negligible number of arbitrage transactions in the worst case.

In summary, our contributions are: 
(1)~Identification of game-theoretic security challenges in sharded contracts,
(2)~quantification of the AMM performance bottleneck, 
(3)~SAMM: a sharded AMM contract, 
(4)~generalization of the AMM trading fee function and a novel function for SAMM, 
(5)~game-theoretic analysis identifying a Subgame-Perfect Nash Equilibrium (SPNE), 
(6)~evaluation demonstrating an increase in throughput of~5x in Sui and of~16x in Solana (up to each blockchain's limit),
(7)~simulation with real trade data confirming the effects of SAMM's incentive design, and
(8)~validation that sharding does not incentivize sandwich attacks, and that it neither increases losses nor incurs substantial transaction overhead under price fluctuations.

These results provide an initial demonstration that high-demand smart contracts can be scaled securely by sharding a single logical service into multiple identical instances~(\S\ref{sec:conclusion}).
They also hint at an upcoming challenge in smart-contract platform design: minimizing the serial elements of transaction processing. 
But SAMM can already be employed to scale AMM performance.


\section{Related Work} \label{sec:related}

Addressing blockchain throughput limitations has been a central research focus since its inception.
The first generation of blockchains, starting with Bitcoin~\cite{nakamoto2008bitcoin}, suffered from throughput limitations due to their consensus protocols. 
A body of work overcame this with a variety of protocols~\cite{eyal2016bitcoin,zhang2022nc,yakovenko2018solana,spiegelman2022bullshark,danezis2022narwhal,bagaria2019prism, abraham2023colordag,gueta2019sbft,li2020decentralized, rocket2019scalable, zamani2018rapidchain, kokoris2018omniledger, wang2019monoxide}. 
With consensus constraints out of the way, the serial execution of transactions became the bottleneck.

Subsequent approaches identify read/write set conflicts and parallelize non-conflicting smart contract transaction execution~\cite{sergey2017concurrent,dickerson2017adding,liu2021parallel,pirlea2021,garamvolgyi2022utilizing}.
Numerous practical blockchain platforms adopt such parallel execution engines: some require applications to explicitly shard independent components to expose parallelism~\cite{blackshear2023sui,yakovenko2018solana}, while others automatically parallelize non-conflicting transactions~\cite{aptos2022}.
On the former type of platforms, existing smart contracts (stateful programs) often rely on naive sharding, which partitions an independent component into multiple instances~\cite{morphoMarketDocs,kaminoU,bucketProtocolDocs,hubbleProtocolDocs}, to benefit from parallel execution.
However, non-parallelizable operations within a smart contract remain a bottleneck.
To the best of our knowledge, no prior work has addressed improving the throughput of non-parallelizable smart contracts via sharding and studied the associated security challenges.

Uniswap~v1 and v2~\cite{zhang2018formal,adams2020uniswap} established the \emph{Constant Product Market Maker} (\emph{CPMM}) as a widely adopted, order-book-free design for token exchange~\cite{SushiSwap,OmniSwap}.
Uniswap~v3~\cite{adams2021uniswap} and Curve~\cite{Curve} further improved AMMs' ability to handle larger trades.
Uniswap~v4~\cite{asef2024x} and UniswapX~\cite{uniswapx} introduced practical updates to how pools are organized and trades are executed.
Our sharding design targets parallelism and incentive alignment and can be composed with~them.

Research on AMMs has studied generalized market-making mechanisms~\cite{angeris2020improved}, \emph{routers} which maximize trading utility by using multiple AMMs~\cite{angeris2022optimal,engel2021composing,zhou2021a2mm}, advanced arbitrage techniques~\cite{zhou2021high, kulkarni2023routing, bartoletti2022maximizing, wang2023n,wang2022cyclic}, improving liquidity providers' returns~\cite{goyal2023finding}, measuring loss due to arbitrage~\cite{milionis2022automated, fritsch2024measuring,milionis2023automated}, ensuring transaction privacy~\cite{chitra2022differential}, eliminating Miner Extractable Value (MEV) for fair trades~\cite{chan2024mechanism, xavier2023credible, canidio2023batching,heimbach2022eliminating}, and examining the synergy between blockchain-based AMMs and prediction market mechanisms~\cite{SchlegelKM23}.
They do not consider AMM throughput and its security implications.

AMM throughput is limited by a single thread in current blockchains.
Some AMM contracts~\cite{zkswap,QuickSwap,michel2025ammboost,unichainWhitepaper,uniswapOptimismAlpha} are deployed on so-called layer-2 solutions~\cite{adams2024layer,gudgeon2020sok, polygon, arbitrum} to defer some workload off-chain, but this merely creates a separate environment for AMMs, with scaling issues persisting (e.g., $138\textit{tps}$ in ammBoost~\cite{michel2025ammboost}) within this realm while introducing additional security assumptions.\footref{fnote:l2}
Some existing AMMs use multiple instances of a single token pair with different parameters for flexibility~\cite{adams2021uniswap}, but most of the demand concentrates on a single instance~\cite{lehar2023fragmentation}.
To the best of our knowledge, previous work did not address the single-thread bottleneck of smart contracts in general and specifically for AMMs.


	\section{Model} \label{sec:model}

We abstract away the blockchain details for our analysis and model the system as a set of participants interacting directly with AMMs, exchanging tokens.
The system progresses in discrete steps and there exists an external market used by arbitrageurs to arbitrage in our system.
The model uses a generic AMM, which we will later instantiate based on previous work and with SAMM.


		\subparagraph*{Participants.}\label{sec:model:participants}
There are two types of participants, \emph{liquidity providers} and \emph{traders}.
Each liquidity provider holds some \textit{token~A} and \textit{token~B}.
They aim to increase their holdings.
Traders are either \textit{AB} or \textit{BA}, based on their goals.
Each \textit{BA} trader occasionally wishes to get a certain amount of  \textit{token~A}, and vice versa for \textit{AB} traders.
They have sufficiently many tokens of the opposite type to complete their trade, but they aim to minimize the cost of obtaining the desired tokens. 
Both liquidity providers and traders can send and receive tokens to and from the smart contracts.


\subparagraph*{Automated Market Makers.}
The system also includes Automated Market Makers, stateful programs that facilitate \emph{depositing} and \emph{trading} of tokens. 
Each AMM maintains some deposited amounts $\poolcontent{}{}{A}$ of \textit{token~A} and $\poolcontent{}{}{B}$ of \textit{token~B}.
We call these tokens \emph{liquidity}.

Within an AMM, there are three primary operations: \emph{liquidity addition}, \emph{liquidity removal}, and \emph{trade}. Specifically, a liquidity provider deposits $\poolinput{}{}{A}$ \textit{token~A} and $\poolinput{}{}{B}$ \textit{token~B} to the contract (liquidity addition); withdraws $\pooloutput{}{}{A}$ \textit{token~A} and $\pooloutput{}{}{B}$ \textit{token~B} from the contract (liquidity removal); or a trader sends $\poolinput{}{}{A}$ \textit{token~A} (resp., $\poolinput{}{}{B}$ \textit{token~B}) to the contract and gets $\pooloutput{}{}{B}$ \textit{token~B} (resp., $\pooloutput{}{}{A}$ \textit{token~A}) from the contract (trade).

In a trade operation, the required input amount for a trader to receive a specific output amount of another token depends on both the output amount and the AMM's current state.
We define the gross amount of a \textit{BA} trade as the required input amount of \textit{token~B} to get $\pooloutput{}{}{A}$ \textit{token~A} in the AMM.
We denote it by $G(\poolcontent{}{}{A}, \poolcontent{}{}{B}, \pooloutput{}{}{A})$ (resp., $G(\poolcontent{}{}{B}, \poolcontent{}{}{A},  \pooloutput{}{}{B})$ for \textit{AB} traders).
Within the gross amount, traders pay a so-called \emph{trading fee} which contributes to the liquidity providers' revenue.

Liquidity providers supply liquidity by depositing their tokens in the contract. 
Once contributing to the AMM, a provider receives tokens from trading fees, hence earning revenue. 
Liquidity providers can later withdraw their tokens from the AMM.

Following prior work~\cite{milionis2022automated,goyal2023finding}, we assume that the trading fee is directly paid to the liquidity providers.
Although some practical AMMs (e.g., Uniswap v2~\cite{adams2020uniswap}) reinvest the trading fees into themselves, allowing liquidity providers to withdraw more tokens than they deposited as utilities, the trading fee's impact is negligible relative to the deposited amount: 
The average ratio of the output amount to the deposited amount is very small (less than $0.036\%$ as we find in \fv{\S\ref{app:uniswapdata}}{\fullver}), and automated market makers (AMMs) typically charge a low trading fee relative to the gross amount (e.g., $0.3\%$ in Uniswap), so the trading fee's impact is negligible relative to the amount of deposited tokens. 


\subparagraph*{System State and Progress.}
There are~$n$ independent AMM shards $\textit{shard}_1$, $ \textit{shard}_2, \cdots, \textit{shard}_n$ in the system.
The system progresses in discrete steps $k = 0, 1,  2,\cdots$ and is orchestrated by a \emph{scheduler}.
In each step, the scheduler randomly selects a participant and this participant executes transactions. 
It chooses a liquidity provider with probability $P_{lp} \geq 0$ and a trader with probability $P_t \geq 0$, where $P_{lp}+P_t = 1$.
The scheduler assigns the liquidity provider $\lpac{}{}{A}$ \textit{token~A} and~$\lpac{}{}{B}$ \textit{token~B}, where $(\lpac{}{}{A}, \lpac{}{}{B})$ follows a random distribution~$D_{lp}$.
		
The trader is either a \textit{BA} trader aiming to obtain \textit{token~A} or an  \textit{AB} trader aiming to obtain \textit{token~B}.
The probability of drawing an \textit{AB} trader (resp., \textit{BA} trader) is $P_t^{ \textit{AB}} \geq 0$ (resp., $P_t^{ \textit{BA}} \geq 0$), with $P_t^{ \textit{AB}} +  P_t^{ \textit{BA}} = P_t$.
To avoid repetition, we only show the case of \textit{BA} traders; the expressions for \textit{AB} traders are symmetric.
The scheduler assigns the \textit{BA} trader an amount of~$\uac{}{}{\textit{BA}}$ \textit{token~A} to obtain, following a random distribution~$D^{\textit{BA}}$. 


\subparagraph*{External Market and Arbitrageurs.}
Following prior work~\cite{goyal2023finding,milionis2022automated}, we assume there is an \emph{external market} providing the price of \textit{token~A} and \textit{token~B}, $p^A$ and $p^B$, respectively.
These prices do not change due to trades; they serve as objective prices for \textit{token~A} and \textit{token~B}.

If the AMM sells tokens at a price lower than the external market, a principal can buy tokens from the AMM and sell them in the external market to make profits, and vice versa; this is called an \emph{arbitrage}.
Previous work~\cite{milionis2022automated, goyal2023finding, canidio2023batching, chan2024mechanism} assumes active and rational arbitrageurs who can use the external market and arbitrage to maximize their utility.
We follow the assumption~\cite{milionis2022automated} that there are immediate arbitrages without trading fees when the token price in the AMM is different from the external market.


\subparagraph*{Threat Model and Our Goal.}
As is typical in blockchain security research~\cite{zhang2019lay,hou2019squirrl,mirkin2020bdos,carlsten2016instability,yaish2023uncle,tsabary2021mad,mccorry2019smart,daian2020flash}, rational participants are potential attackers that might deviate from the intended protocol to maximize their own utility.
In our context, attackers might, for example, overload the system or create temporary imbalances that reduce parallelism.
Our goal is to design a set of AMM shards such that all rational participants behave in a way that improves the overall throughput of the system.

\section{SAMM: Sharded AMM} \label{sec:samm}

We present SAMM, an AMM comprising multiple shards. 
Each shard is an independent AMM~(\S\ref{sec:SAMM:structure}), and SAMM's goal is to have operations distributed among the different shards~(\S\ref{sec:SAMM:properties}).
To enforce this desired behavior despite attacks and untrusted users, we present a novel trading fee function~(\S\ref{sec:SAMM_trading_fee}) and find parameters to fulfill SAMM's goal~(\S\ref{sec:SAMM_parameter}). 


		\subsection{SAMM Structure}\label{sec:SAMM:structure}

SAMM enables parallel processing by deploying multiple AMM shards, each with an independent liquidity pool.
There is no data dependency among the shards and there are no global elements. 
The operation of each shard is based on the Constant Product Market Maker (See detailed overview in \fv{\S\ref{app:preliminaries}}{\fullver}), which keeps the product of the two tokens constant.

When liquidity providers add tokens to a shard, they receive \emph{share tokens}, which signify their portion of its liquidity. 
Assume the shard has $\poolcontent{}{}{A}$ \textit{token~A} and $\poolcontent{}{}{B}$ \textit{token~B} deposited, and $\poolcontent{}{}{S}$ is the number of all share tokens distributed before the operation.
The amount of share tokens acquired by the liquidity provider is~$\pooloutput{}{}{S} = \poolcontent{}{}{S} \times \min\left\{\frac{\poolinput{}{}{A}}{\poolcontent{}{}{A}}, \frac{\poolinput{}{}{B}}{\poolcontent{}{}{B}} \right\}$.
If a liquidity provider pays $\poolinput{}{}{S}$ share tokens to withdraw deposited tokens, she gets $\pooloutput{}{}{A}= \frac{\poolinput{}{}{S}}{\poolcontent{}{}{S}} \times \poolcontent{}{}{A}$ \textit{token~A} and~{$\pooloutput{}{}{B}= \frac{\poolinput{}{}{S}}{\poolcontent{}{}{S}}\times \poolcontent{}{}{B}$} \textit{token~B}.

When a trader wants to get $\pooloutput{}{}{A}$ \textit{token~A} from the AMM, ignoring trading fees, she needs to pay the \emph{net amount} 
\begin{equation}\label{eq:netamount}
	\textit{net}(\poolcontent{}{}{A}, \poolcontent{}{}{B}, \pooloutput{}{}{A}) = \frac{\poolcontent{}{}{B} \times \pooloutput{}{}{A}}{\poolcontent{}{}{A} - \pooloutput{}{}{A}} \,\, .
\end{equation}
The \emph{net price} for a single \textit{token~A} is $\frac{\poolcontent{}{}{B}}{\poolcontent{}{}{A} - \pooloutput{}{}{A}}$ \textit{token~B}.
This value increases as the output amount of \textit{token~A}, $\pooloutput{}{}{A}$, increases, which is the \emph{Slippage} of the trade.


		\subsection{Desired Properties}\label{sec:SAMM:properties} 

Most operations in an AMM are trades (We observe $99.5\%$ in historical blockchain records, see \fv{Appendix~\ref{app:uniswapdata}}{\fullver}).
Our main goal is therefore that traders distribute the workload evenly among the shards.
The following properties will suffice. 

\subparagraph*{c-Non-Splitting Property.}
Our first goal is that trades only use a single shard, i.e., traders do not split a trade into smaller ones on multiple shards.
This avoids the overhead caused by trade splits.
So, the cost of a single-trade transaction should be less than the combined cost of multiple, split-trade transactions. 
However, this principle faces a significant challenge as highlighted by Equation~\ref{eq:netamount}: When the output amount is not small enough in comparison to the shard's deposited amount, the resulting slippage could incentivize traders to split their transactions to mitigate this slippage.
Consequently, we refine our requirement to ensure that this principle is adhered to only when the output amount is relatively small compared to the deposited amount, specifically when the ratio is below a predefined constant, $c$.
Indeed, the prominent pairs in Ethereum's Uniswap v2 data support this approach, with $99\%$ of trades having a ratio of output amount to deposited amount below $0.0052$ (See \fv{Appendix~\ref{app:uniswapdata}}{\fullver} for more details).
We call this the \emph{$c$-non-splitting} property.

\begin{restatable}[$c$-non-splitting]{prop}{propertyconcavity}\label{concavity}
	Let $m \geq 2$. 
	Given a set of $m$ output amounts $\{O^A_1, O^A_2, \ldots, O^A_m\}$ such that all amounts are positive, i.e., $\forall 1 \leq j \leq m: O^A_j > 0$,
	denote by $\tilde{O}^A$ the sum of the amounts in the set, $\tilde{O}^A = \sum_{j=1}^m \pooloutput{j}{}{A}$.
	For the constant $0 < c < 1$ and the deposited amount of tokens $\poolcontent{}{}{A}$ and $\poolcontent{}{}{B}$, if $\frac{\tilde{O}^A}{\poolcontent{}{}{A}} \leq c$ , then the cost of trading $\tilde{O}^A$ \textit{token~A} is less than the sum of the cost of trading $\pooloutput{j}{}{A}$ \textit{token~A} for $1 \leq j \leq m$, i.e.,~$G_{\textit{SAMM}}(\poolcontent{}{}{A},\poolcontent{}{}{B}, \tilde{O}^A) < \sum_{j=1}^m G_{\textit{SAMM}}(\poolcontent{}{}{A},\poolcontent{}{}{B}, \pooloutput{j}{}{A}).$
\end{restatable}

\subparagraph*{c-Smaller-Better Property.}
Our second goal is to maintain balanced volumes in all shards.
This is crucial because the volume affects the slippage. 
When there are differences in shard sizes, with some being much smaller than others, the slippage in trading within these smaller shards is significantly greater than in larger ones. 
This discrepancy can result in transactions clustering in the larger shards rather than spread evenly, reducing parallelism.
Moreover, an attacker could disrupt the system by adding or removing liquidity from certain shards, thereby harming overall system performance.

We address this by incentivizing liquidity providers to allocate their tokens to the shards with lower volumes.
Intuitively, this ensures that smaller shards receive more frequent fees from traders when shard volumes are not balanced.
This incentivizes liquidity providers to deposit tokens in these smaller shards.
Therefore, the system would converge to the state where all shards have balanced volumes, and traders then randomly select shards for trading.
As with the $c$-non-splitting property, large transactions suffer from high slippage, resulting in an advantage for larger shards.
So here too, we refine this requirement to scenarios where the ratio of the traded amount to the deposited amount is below a threshold~$c$.
We call this the \emph{$c$-smaller-better property}.

\begin{restatable}[$c$-smaller-better]{prop}{propertysmallerpoolsmallercost}\label{smallerpoolsmallercost}
	Given an output amount ${\pooloutput{}{}{A} > 0}$, for any two shards with deposited token amounts $(\poolcontent{i}{}{A},\poolcontent{i}{}{B})$ and $(\poolcontent{j}{}{A},\poolcontent{j}{}{B})$, respectively, and $\poolcontent{i}{}{A} < \poolcontent{j}{}{A}$.
	For the constant $0 < c < 1$, if
	$\pooloutput{}{}{A} / \poolcontent{j}{}{A} < \pooloutput{}{}{A} / \poolcontent{i}{}{A} \leq c$ and 
	$\poolcontent{i}{}{A}/\poolcontent{i}{}{B} = \poolcontent{j}{}{A}/ \poolcontent{j}{}{B}$, then the cost of trading $\pooloutput{}{}{A}$ \textit{token~A} in the smaller shard is less than that in the larger shard, i.e.,~${G_{\textit{SAMM}}(\poolcontent{i}{}{A},\poolcontent{i}{}{B}, \pooloutput{}{}{A}) < G_{\textit{SAMM}}(\poolcontent{j}{}{A},\poolcontent{j}{}{B}, \pooloutput{}{}{A})}$.
\end{restatable}


\subparagraph*{$c$ value.}
If the above properties are satisfied for a particular~$c$ but a trade occurs with a larger ratio of output amount to the deposited amount, traders may split their transactions or tend to use larger shards to minimize their cost.
Therefore, $c$ should be as large as possible to ensure such occurrences are rare. 


\subparagraph*{Note.}
CPMMs charge a constant ratio of the net amount as the trading fee\fvonly{ (\S\ref{sec:preliminaries:tf})}.
Therefore, their cost function does not satisfy either property due to slippage \fv{(\S\ref{app:limitation_cpmm})}{(see \fullver)}.


		\subsection{Trading Fee Design}\label{sec:SAMM_trading_fee}

To satisfy properties~\ref{concavity} and~\ref{smallerpoolsmallercost}, we first generalize the trading fee function to provide flexibility in the incentive design.
Then, we propose a specific trading fee function.


\subparagraph*{Generic Trading Fee Function.}
The gross amount of a trade comprises the trading fee and the net amount, with the latter being determined by the CPMM curve. 
To maintain the foundational characteristics of AMMs, we do not modify the CPMM curve. 
Instead, we generalize the trading fee function beyond simply charging a ratio of the net amount as in previous work (e.g.,~\cite{angeris2020improved,goyal2023finding,angeris2022optimal}). 

Denote the trading fee function of the AMM by $\textit{tf}(\poolcontent{}{}{A}, \poolcontent{}{}{B}, \pooloutput{}{}{A})$, which takes the deposited amount of \textit{token~A},~$\poolcontent{}{}{A}$, and \textit{token~B}, $\poolcontent{}{}{B}$, in the shard and the output amount of \textit{token~A}, $\pooloutput{}{}{A}$, and outputs the amount of \textit{token~B} the trader needs to pay as the trading fee.
Then, the gross amount of getting $\pooloutput{}{}{A}$ \textit{token~A} is: 
\begin{equation}\label{cost_function}
	G(\poolcontent{}{}{A},\poolcontent{}{}{B}, \pooloutput{}{}{A}) = \textit{tf}(\poolcontent{}{}{A}, \poolcontent{}{}{B}, \pooloutput{}{}{A}) + \textit{net}(\poolcontent{}{}{A},\poolcontent{}{}{B}, \pooloutput{}{}{A}) \,\, .
\end{equation}

\subparagraph*{Bounded-Ratio Trading Fee Function.}
In order to achieve the desired properties, we need flexibility for the trading fee function design.
A monomial function is easy to parameterize and to differentiate, and is sufficient to achieve most of our goals.
The function takes the variables available on a trade, $\poolcontent{}{}{A}, \poolcontent{}{}{B}, \pooloutput{}{}{A}$.
It is parameterized by four values,~$\beta_1, \beta_2, \beta_3, \beta_4$: $\textit{tf}(R^A, R^B, O^A; \beta_1, \beta_2, \beta_3, \beta_4) \coloneqq \beta_1 (\poolcontent{}{}{A}) ^{\beta_2}  (\poolcontent{}{}{B})^{\beta_3}(\pooloutput{}{}{A})^{\beta_4}.$

The monomial is unbounded: Without limits, the trading fee can be too high to attract traders or too low to reward liquidity providers.
Analogously to the fixed ratio in prior work, we thus add adjustable bounds~$r_{\min}, r_{\max}$ and a base offset~$\beta_5$ to fine-tune the fee's absolute value, yielding the bounded-ratio polynomial:
\begin{multline}\label{tf_bounded_ratio}
    \textit{tf}_{\textit{BRP}}(\poolcontent{}{}{A}, \poolcontent{}{}{B}, \pooloutput{}{}{A}; \beta_1, \beta_2, \beta_3, \beta_4, \beta_5) \coloneqq \frac{\poolcontent{}{}{B}}{\poolcontent{}{}{A}} \pooloutput{}{}{A} \times
     \\  \max \{ 
    r_{\min}, 
    \min \{ r_{\max},
    \beta_1 (\poolcontent{}{}{A})^{\beta_2} (\poolcontent{}{}{B})^{\beta_3} (\pooloutput{}{}{A})^{\beta_4} + \beta_5 
    \} 
    \}\,\,.
\end{multline}

The ratio ${\poolcontent{}{}{B} / \poolcontent{}{}{A}}$ represents the market price of \textit{token~A} relative to \textit{token~B} \fv{(\S\ref{sec:preliminaries:trade})}{(see \fullver)}.
The product $\frac{\poolcontent{}{}{B}}{\poolcontent{}{}{A}}\pooloutput{}{}{A}$ is thus the~trader's net payment in terms of \textit{token~B} without slippage.
Then ${\max\{
	r_{\min} , 
	\min\{
	r_{\max},
	\beta_1 (\poolcontent{}{}{A}) ^{\beta_2}  (\poolcontent{}{}{B})^{\beta_3}(\pooloutput{}{}{A})^{\beta_4} + \beta_5
	\}	
	\}}$ acts as the constrained ratio of the trading fee to that net payment.
We hereinafter omit the parameters and write~${\textit{tf}_{\textit{BRP}}(\poolcontent{}{}{A}, \poolcontent{}{}{B}, \pooloutput{}{}{A})}$ for brevity.


        \subsection{Parameter Selection}\label{sec:SAMM_parameter}

Now we turn to the selection of parameters for the SAMM trading fee function.
First, we identify the necessary conditions under which the bounded-ratio polynomial trading fee function aligns with the $c$-smaller-better property.

\begin{restatable}{propos}{theoremnecessary}\label{theorem:parameter1}
	Let ${\textit{tf}_{\textit{SAMM}}(\poolcontent{}{}{A}, \poolcontent{}{}{B},  \pooloutput{}{}{A}) = \textit{tf}_{\textit{BRP}}(\poolcontent{}{}{A}, \poolcontent{}{}{B},  \pooloutput{}{}{A})}$, then the following conditions are necessary for $c$-smaller-better
	to hold for $G_{\textit{SAMM}}(\poolcontent{}{}{A},\poolcontent{}{}{B}, \pooloutput{}{}{A})$:
	(1) $\beta_3 = 0$; (2) $\beta_2 + \beta_4 = 0$; (3) $\beta_1 < 0$; (4) $0 < \beta_4 \leq 1$; (5) $r_{\min} < \beta_5 \leq r_{\max}$; and (6) $\frac{\beta_5 - r_{\min}}{-\beta_1} \geq c^{\beta_4}$.
\end{restatable}

We require the polynomial value to be between~$r_{\min}$ and~$r_{\max}$.
Then, we need to make sure that the derivative of the gross amount with respect to the shard size is non-negative to ensure the $c$-smaller-better property.
Required items come directly from these two restrictions.
\fv{We defer the proof to Appendix~\ref{app:proof_necessary}.}{The proof is in \fullver.}

Based on Proposition~\ref{theorem:parameter1}, we require that $\beta_1 < 0$, ${\beta_2 + \beta_4 = 0}$, $\beta_3 =0$, $0 <\beta_4 \leq 1$, and~$r_{\min} < \beta_5 \leq r_{\max}$.
Next, we identify additional conditions that are sufficient for both properties to hold.

\begin{restatable}{thm}{theoremsufficient}
	\label{theorem:parameter2}
	Let $\textit{tf}_{\textit{SAMM}}(\poolcontent{}{}{A}, \poolcontent{}{}{B},  \pooloutput{}{}{A}) = \textit{tf}_{\textit{BRP}}(\poolcontent{}{}{A}, \poolcontent{}{}{B},  \pooloutput{}{}{A})$, if ${\beta_1 < 0}, {\beta_2 + \beta_4 = 0},$ ${\beta_3 =0}, {0 <\beta_4 \leq 1},  {r_{\min} < \beta_5 \leq r_{\max}}$ and $\frac{\beta_5 - r_{\min}}{-\beta_1} \geq c^{\beta_4}$, then the following items are sufficient for the $c$-non-splitting and $c$-smaller-better properties to hold for $G_{\textit{SAMM}}(\poolcontent{}{}{A},\poolcontent{}{}{B}, \pooloutput{}{}{A})$:
	(1) $\beta_1\beta_4(\beta_4+1)c^{\beta_4-1}(1-c)^3\leq -2$; and~(2)~$-\beta_1\beta_4 \geq\frac{c^{1-\beta_4}}{(1-c)^2}$.
\end{restatable}

The $c$-smaller-better property is satisfied when the derivative of the gross amount is positive.
It is sufficient to ensure the $c$-non-splitting property when the gross amount is concave with respect to the output amount, which is ensured by a negative second derivative over the output amount.
\fv{The proof is in Appendix~\ref{app:proof_sufficient}.}{The proof is in \fullver.}

By setting $\beta_2 = -1, \beta_3 = 0, \beta_4 = 1, \beta_5 = r_{\max}$, and choosing $\beta_1 < -1$, the fee function satisfies the above requirements, leaving just three parameters: 
\begin{equation*}
	\textit{tf}_{\textit{SAMM}}(\poolcontent{}{}{A}, \poolcontent{}{}{B},  \pooloutput{}{}{A}) =   \frac{\poolcontent{}{}{B}}{\poolcontent{}{}{A}} \times \pooloutput{}{}{A} \times \max
	\left\{	r_{\min}, \beta_1\times \frac{\pooloutput{}{}{A}}{\poolcontent{}{}{A}} + r_{\max} 
	\right\} \,\, . 
\end{equation*}
By setting $\beta_4 = 1$ and $\beta_5 = r_{\max}$ in the sufficient condition of the $c$-non-splitting and $c$-smaller-better properties (Theorem~\ref{theorem:parameter2}), the sufficient condition becomes:
\begin{restatable}{cor}{theoremspecialpara}\label{corollary:parameter}
	For any $\beta_1 < -1$ and $c$ satisfying $c \leq \min\left\{ 1 - (-\beta_1)^{-\frac{1}{3}}, \frac{r_{\max} - r_{\min}}{-\beta_1}\right\}$,	the SAMM cost function $G_{\textit{SAMM}}(\poolcontent{}{}{A},\poolcontent{}{}{B}, \pooloutput{}{}{A})$ satisfies the $c$-non-splitting and $c$-smaller-better properties.
\end{restatable}
For instance, the parameters $\beta_1 = -1.05$, $r_{\max} = 0.012$, $r_{\min} = 0.001$, and $c = 0.0104$ meet the specified criteria. 
According to historical records \fv{(\S\ref{app:uniswapdata})}{(see \fullver)}, over $99\%$ of Uniswap v2 transactions have a ratio below $0.0052$, suggesting that if its liquidity is split into two shards, $99\%$ of transactions would fall within our targeted range. 
Specifically, in the AMMs with the highest trading volumes, USDC-ETH and USDT-ETH, the ratio remains below $0.00128$,
which can manage eight shards. 
Increasing the number of shards could be achieved by adjusting the value of~$c$ by increasing~$r_{\max}$ and~$\beta_1$.


    \section{Analysis of Game-Theoretic Security}\label{sec:game}

The model gives rise to a game (\S\ref{GT:game_model}) played by traders and liquidity providers.
Our goal (\S\ref{GT:property}) is to ensure that at equilibrium, traders randomly select shards to trade in without splitting, thereby enhancing throughput.
SAMM achieves this~(\S\ref{sec:GT:equilibrium}) and converges to a state where all shards have balanced volumes, overcoming attacks that unbalance the shards.


		\subsection{Game Model}\label{GT:game_model}

We derive from the model a sequential game with discrete steps $k = 0, 1, \cdots$.
Denote the game, parameterized by the number $n$ of shards and their trading fee function~$\textit{tf}$, by~$\Gamma_n(\textit{tf})$.

\subparagraph*{System State.}
In $\Gamma_n(\textit{tf})$, the state of each shard $\textit{shard}_i$ in step $k$ consists of the amount of deposited~\textit{token~A}, the amount of deposited~\textit{token~B}, and the amount of share tokens, $\poolcontent{i}{k}{A}, \poolcontent{i}{k}{B}, \poolcontent{i}{k}{S}$, respectively.
Recall that share tokens are not deposited in the shard but are held by liquidity providers, so $\poolcontent{i}{}{S}$ is the total amount of $\textit{shard}_i$'s share tokens held by liquidity providers.
Denote by $\poolcontentvector{}{k}{}$ the state of all AMM contracts in step~$k$,
${\left(
\left(\poolcontent{1}{k}{A},\poolcontent{1}{k}{B}, \poolcontent{1}{k}{S} \right),
\cdots, \left(\poolcontent{n}{k}{A},\poolcontent{n}{k}{B}, \poolcontent{n}{k}{S}	 \right) \right)}$.

\subparagraph*{Liquidity Provider Actions.}
The liquidity provider decides the amount of tokens she deposits in each shard.
Denote the amount of \textit{token~A} and \textit{token~B} she deposits in $\textit{shard}_i$ by~$\lpac{i}{}{A} , \lpac{i}{}{B} \geq 0$, respectively.
Recall that the scheduler assigns the liquidity provider~$\lpac{}{}{A}$ \textit{token~A} and~$\lpac{}{}{B}$ \textit{token~B}.
The total amount of tokens deposited should not exceed the amount she holds:
\begin{equation*}
	\forall 1 \leq i \leq n , \lpac{i}{}{A}, \lpac{i}{}{B} \geq 0,
	\sum_{i=1}^n \lpac{i}{}{A} \leq \lpac{}{}{A}, \sum_{i=1}^n \lpac{i}{}{B} \leq \lpac{}{}{B} \,\, .
\end{equation*}

The action of a liquidity provider is thus the vector $a_{lp} = \left(\left(\lpac{1}{}{A}, \lpac{1}{}{B}\right), \cdots, \left(\lpac{n}{}{A}, \lpac{n}{}{B}\right)\right)$.
When the liquidity provider takes this action with the system state $\poolcontentvector{}{k}{}$, the liquidity provider receives $\pooloutput{i}{}{S}$ share tokens from $\textit{shard}_i$, where ${\pooloutput{i}{}{S} = \poolcontent{i}{k}{S} \times \min\left\{\frac{\lpac{i}{}{A}}{\poolcontent{i}{k}{A}}, \frac{\lpac{i}{}{B}}{\poolcontent{i}{k}{B}} \right\} }$.

There is no arbitrage opportunity only if ${\poolcontent{i}{}{A}}/{\poolcontent{i}{}{B}} = {p^B}/{p^A}$ \fv{(\S\ref{sec:preliminaries:trade})}{(see \fullver)}, so arbitrageurs enforce this equality. 
When~${\lpac{i}{}{A}/\lpac{i}{}{B} = p^B/p^A}$, increasing just one of $\lpac{i}{}{A}$ or $\lpac{i}{}{B}$ would not increase the share token the liquidity provider receives, which means more payment without more revenue.
Therefore, we only consider actions where ${\lpac{i}{}{A}}/{\lpac{i}{}{B}} = {p^B}/{p^A}$ and require ${\lpac{}{}{A}}/{\lpac{}{}{B}} = {p^B}/{p^A}$.
The action space of a liquidity provider is denoted by $\mathcal{A}_{\textit{lp}}(\lpac{}{}{A} , \lpac{}{}{B} )$:
\begin{equation*}
    \mathcal{A}_{lp}(\lpac{}{}{A} , \lpac{}{}{B} ) =
    \left\{ a_{lp}  \left|
    \begin{aligned}
        &\forall 1 \leq i \leq n, \lpac{i}{}{A} = \frac{p^B}{p^A} \lpac{i}{}{B} \geq 0, \\
        &\sum_{i=1}^n \lpac{i}{}{A} \leq \lpac{}{}{A}, \sum_{i=1}^n \lpac{i}{}{B} \leq \lpac{}{}{B}
    \end{aligned}
    \right. \right\} \,\, .
\end{equation*}

We denote the updated state of shards from the previous state $\poolcontentvector{}{}{}$ and the action of a liquidity provider $a_{lp}$ by $\poolcontentvector{}{}{} + a_{lp}$.
Then for ${\poolcontentvector{}{}{\prime} = \left(\left(\poolcontent{1}{}{A^{\prime}},\poolcontent{1}{}{B^{\prime}}, \poolcontent{1}{}{S^{\prime}} \right), \cdots \left(\poolcontent{n}{}{A^{\prime}},\poolcontent{n}{}{B^{\prime}}, \poolcontent{n}{}{S^{\prime}}	 \right) \right) = \poolcontentvector{}{}{} + a_{lp}}$, we have
\begin{equation*}
	\poolcontent{i}{}{A^{\prime}} = \poolcontent{i}{}{A} + \lpac{i}{}{A}, 
	\poolcontent{i}{}{B^{\prime}} = \poolcontent{i}{}{B} + \lpac{i}{}{B}, 
	\poolcontent{i}{}{S^{\prime}} = \poolcontent{i}{}{S} + 	\pooloutput{i}{}{S} = \left(1 + \frac{ \lpac{i}{}{A} }{ \poolcontent{i}{}{A} }\right)\poolcontent{i}{}{S}.
\end{equation*}

After the liquidity addition operation, there is no arbitrage opportunity for the arbitrageurs since $\frac{\poolcontent{i}{}{A} + \lpac{i}{}{A}}{\poolcontent{i}{}{B} + \lpac{i}{}{B}} = \frac{\poolcontent{i}{}{A}}{\poolcontent{i}{}{B}} = \frac{p^B}{p^A}$ \fv{(\S\ref{sec:preliminaries:trade})}{(see \fullver)}.
Then the update of state in step $k+1$ is~${\poolcontentvector{}{k+1}{} = \poolcontentvector{}{k}{} + a_{lp}.}$

\subparagraph*{Trader Actions.}
The action of a \textit{BA} trader determines the amount of \textit{token~A} she acquires from each shard.
Denote by~$\uac{i}{}{BA} \geq 0$ the amount of \textit{token~A} she acquires in $\textit{shard}_i$.
Recall that the scheduler assigns the \textit{BA} trader $\uac{}{}{\textit{BA}}$ \textit{token~A} to acquire in total.
The action of a \textit{BA} trader is thus the vector $a^{\textit{BA}} = \left(\uac{1}{}{\textit{BA}}, \cdots, \uac{n}{}{\textit{BA}} \right)$.
The action space of a \textit{BA} trader is denoted by $\ActionSpaceBA(\uac{}{}{\textit{BA}})$, which is the set of all feasible actions: 
\begin{equation}\label{restriction_ba}
	\ActionSpaceBA(\uac{}{}{\textit{BA}}) = \left\{ a^{\textit{BA}}  \left| \forall 1 \leq i \leq n,\uac{i}{}{\textit{BA}}  \geq 0, \sum_{i=1}^n \uac{i}{}{\textit{BA}} = \uac{}{}{\textit{BA}} \right. \right\} \,\, .
\end{equation}

After the trade operation and the arbitrage, $\poolcontent{i}{}{A}$ and $\poolcontent{i}{}{B}$ remain unchanged \fv{(\S\ref{sec:preliminaries:trade})}{(see \fullver)}.
Consequently, the state of the shards remains unchanged in the subsequent step: ${\forall 1 \leq i \leq n, \poolcontent{i}{k+1}{A} = \poolcontent{i}{k}{A}, \poolcontent{i}{k+1}{B} = \poolcontent{i}{k}{B}}$.


\subparagraph*{Utility and Strategies.}
For traders and liquidity providers, we first discuss their revenue and then define their strategies and utility, respectively.
Players determine the value of tokens according to the external market, namely $p^A$ and $p^B$.

Consider a \textit{BA} trader whose goal is to acquire $\uac{}{}{\textit{BA}}$ units of \textit{token~A}. 
This trader needs to pay the gross amount and may derive some fixed reward from getting these tokens.
We consider her revenue only as the inverse of the gross amount in terms of \textit{token~B} times the value of each \textit{token~B}: 
\begin{equation}\label{utilityba}
	U^{\textit{BA}}(	\poolcontentvector{}{}{},  a^{\textit{BA}}) =
	-p^B \times \sum_{i}
	G(\poolcontent{i}{}{A},\poolcontent{i}{}{B}, \uac{i}{}{\textit{BA}}) 
	= -p^B\times\sum_{i}
	G(\poolcontent{i}{}{A},\frac{p^A}{p^B}\poolcontent{i}{}{A}, \uac{i}{}{\textit{BA}})
	\,\, .
\end{equation}

For the trader aiming to get $\uac{}{}{\textit{BA}}$ \textit{token~A}, the strategy of the trader $\pi^{BA}( \poolcontentvector{}{}{}, \uac{}{}{\textit{BA}}, a^{\textit{BA}})$ takes $\poolcontentvector{}{}{}$, $\uac{}{}{\textit{BA}}$ and an action~$a^{\textit{BA}}$ as input, then outputs the probability of taking action~$a^{\textit{BA}}$.
The total probability of all feasible actions should be 1: ${\sum_{a^{\textit{BA}} \in \ActionSpaceBA(\uac{}{}{\textit{BA}})} \pi^{BA}( \poolcontentvector{}{}{}, \uac{}{}{\textit{BA}}, a^{\textit{BA}}) = 1}.$

The utility of the trader over the strategy is a function of the system state $\poolcontentvector{}{}{}$, the assigned requirement $\uac{}{}{\textit{BA}}$, and the strategy of the trader $\pi^{BA}$. 
It is the expected revenue under the distribution of actions:
\begin{equation}\label{utilitybapi}
	U^{BA}(\poolcontentvector{}{}{},  \uac{}{}{\textit{BA}}, \pi^{BA}) =
	\sum_{a^{\textit{BA}} \in \ActionSpaceBA(\uac{}{}{\textit{BA}})}
	\left(
	\pi^{BA}( \poolcontentvector{}{}{}, \uac{}{}{\textit{BA}}, a^{\textit{BA}}) \times
	U^{BA}(\poolcontentvector{}{}{}, a^{\textit{BA}})
	\right)
	\,\, .
\end{equation}

The revenue of a liquidity provider comes from the trading fees paid by traders.
Ignoring the effect of other liquidity providers, all future steps are identical due to arbitrageurs, so the long-term mean revenue is proportional to the next-step mean revenue.
Hence, we consider the \emph{myopic} setting (as in, e.g.,~\cite{roughgarden2021transaction,ferreira2021dynamic}), where the liquidity provider regards the revenue in the next step as her utility.

Denote the revenue of a liquidity provider with her action $a_{lp} = \left(\left(\lpac{1}{}{A}, \lpac{1}{}{B}\right), \cdots, \left(\lpac{n}{}{A}, \lpac{n}{}{B}\right)\right)$, the action of the \textit{BA} trader in the next step $a^{\textit{BA}} = \left(\uac{1}{}{\textit{BA}}, \cdots, \uac{n}{}{\textit{BA}} \right)$ and the system state $\poolcontentvector{}{}{}= \left(\left(\poolcontent{1}{}{A},\poolcontent{1}{}{B}, \poolcontent{1}{}{S} \right), \cdots \left(\poolcontent{n}{}{A},\poolcontent{n}{}{B}, \poolcontent{n}{}{S}	 \right) \right)$, by the function $U_{lp}(\poolcontentvector{}{}{}, a_{lp},a^{BA})$.
In the next step, $\textit{shard}_i$ receives a trading fee of ${\textit{tf}( \poolcontent{i}{}{A} + \lpac{i}{}{A},  \poolcontent{i}{}{B} + \lpac{i}{}{B}, \uac{i}{}{\textit{BA}}) }$. 
The liquidity provider receives a fraction of that fee proportional to her fraction of share tokens out of all shares in the shard.
Therefore, the revenue function is
\begin{align}\label{TFBA}
	U_{lp}(\poolcontentvector{}{}{}, \, a_{lp},a^{BA})=&p^B \times \sum_{i=1}^n
	 \left\{\textit{tf}( \poolcontent{i}{}{A} + \lpac{i}{}{A},  \poolcontent{i}{}{B} + \lpac{i}{}{B}, \uac{i}{}{\textit{BA}})  
	 \times \frac{\frac{ \lpac{i}{}{A} }{ \poolcontent{i}{}{A} }\poolcontent{i}{}{S}}{ \left(1 + \frac{ \lpac{i}{}{A} }{ \poolcontent{i}{}{A} }\right)\poolcontent{i}{}{S}}  
	\right\}\nonumber \\
	=&p^B \times \sum_{i=1}^n
	 \left\{\textit{tf}( \poolcontent{i}{}{A} + \lpac{i}{}{A},  \poolcontent{i}{}{B} + \lpac{i}{}{B}, \uac{i}{}{\textit{BA}})  
	 \times \frac{\lpac{i}{}{A}}{ \lpac{i}{}{A} +  \poolcontent{i}{}{A}} 
	\right\} \,\, .
\end{align}

For the liquidity provider with $\lpac{}{}{A}$ \textit{token~A} and $\lpac{}{}{B}$ \textit{token~B}, the strategy of the liquidity provider $\pi_{lp}( \poolcontentvector{}{}{}, \lpac{}{}{A}, \lpac{}{}{B})$ takes~$\poolcontentvector{}{}{}$,~$\lpac{}{}{A}$,~$\lpac{}{}{B}$ and an action~$a_{lp}$ as input, and outputs the probability of taking action~$a_{lp}$.
The total probability of all feasible actions should be 1,
\begin{equation}
	\sum_{a_{lp} \in \mathcal{A}_{lp}(\lpac{}{}{A} , \lpac{}{}{B} )} \pi_{lp}( \poolcontentvector{}{}{}, \lpac{}{}{A}, \lpac{}{}{B}, a_{lp}) = 1 \,\, .
\end{equation}

The utility of the liquidity provider over strategies is a function of the system state $\poolcontentvector{}{}{}$, the amount of tokens she is assigned $\lpac{}{}{A}$, $\lpac{}{}{B}$, and the strategy of the liquidity provider and traders, $\pi_{lp}, \pi^{BA}, \pi^{AB}$; denote it by $U_{lp}(\poolcontentvector{}{}{},\lpac{}{}{A}, \lpac{}{}{B}, \pi_{lp},  \pi^{BA}, \pi^{AB})$.
Before calculating this, we show the revenue given the action of the liquidity provider and the strategies of traders, denoted by $U_{lp}(\poolcontentvector{}{}{},\lpac{}{}{A}, \lpac{}{}{B}, a_{lp},  \pi^{BA}, \pi^{AB})$.
It takes the system state $\poolcontentvector{}{}{}$, the action of the liquidity provider $a_{lp}$, the strategies of traders~$\pi^{BA}$ and~$\pi^{AB}$ as input, then outputs the expected utility over the strategies and distributions of traders.
The strategy of traders is affected by the state after the liquidity provider's action, namely $\poolcontentvector{}{}{} + a_{lp}$.
Denote by $E_{ \uac{}{}{\textit{BA}} \sim D^{BA}}[ f(\cdot) ]$ the expected value of~$f(\cdot)$ with $\uac{}{}{\textit{BA}} $ sampled from $D^{BA}$.
The revenue~is:
\begin{multline*}
	U_{lp}(\poolcontentvector{}{}{},\lpac{}{}{A}, \lpac{}{}{B}, a_{lp},  \pi^{BA}, \pi^{AB}) = 
	 P_t^{\textit{BA}}\times E_{ \uac{}{}{\textit{BA}} \sim D^{BA}}
	\left[ 
	\sum_{a^{\textit{BA}} \in \ActionSpaceBA(\uac{}{}{\textit{BA}})}
	\left( 
		\substack{\pi^{BA}( \poolcontentvector{}{}{}+a_{lp}, \uac{}{}{\textit{BA}}, a^{\textit{BA}})
	\times \\
		U_{lp}(\poolcontentvector{}{}{}, a_{lp},a^{BA})}
	\right)
	\right]
  \\ +
 P_t^{\textit{AB}} \times E_{ \uac{}{}{\textit{AB}} \sim D^{AB}}
	\left[
	\sum_{a^{\textit{AB}} \in \mathcal{A}_{\textit{AB}}(\uac{}{}{\textit{AB}})}
	\left(
		\substack{
			\pi^{AB}( \poolcontentvector{}{}{}+a_{lp}, \uac{}{}{\textit{AB}}, a^{\textit{AB}})
		\times \\
		U_{lp}(\poolcontentvector{}{}{}, a_{lp},a^{AB})
		}	
	\right)
	\right]
	\,\, .
\end{multline*}

To simplify the presentation, we assume the liquidity provider is always followed by a \textit{BA} trader. 
The expressions for an \textit{AB} trader are symmetric. 
Then, the above equation can be simplified as
\begin{equation}\label{Ulp_action_strategy_simplified}
	U_{lp}(\poolcontentvector{}{}{},\lpac{}{}{A}, \lpac{}{}{B}, a_{lp},  \pi^{BA}) = 
	 E_{ \uac{}{}{\textit{BA}} \sim D^{BA}}
	\left[ 
	\sum_{a^{\textit{BA}} \in \ActionSpaceBA(\uac{}{}{\textit{BA}})}
	\left( 
		\substack{\pi^{BA}( \poolcontentvector{}{}{}+a_{lp}, \uac{}{}{\textit{BA}}, a^{\textit{BA}})
	\times \\
		U_{lp}(\poolcontentvector{}{}{}, a_{lp},a^{BA})}
	\right)
	\right]
	\,\, .
\end{equation}

Then, the utility function of the liquidity provider over strategies is the expected utility under the distribution of actions:
\begin{equation}\label{Ulp_strategy_strategy}
	U_{lp}(\poolcontentvector{}{}{},\lpac{}{}{A}, \lpac{}{}{B}, \pi_{lp},  \pi^{BA}, \pi^{AB}) = 
	\sum_{a_{lp} \in \mathcal{A}_{lp}(\lpac{}{}{A} , \lpac{}{}{B} )}
	\left(
	\substack{\pi_{lp}( \poolcontentvector{}{}{}, \lpac{}{}{A}, \lpac{}{}{B}, a_{lp}) \\ \times
	U_{lp}(\poolcontentvector{}{}{},\lpac{}{}{A}, \lpac{}{}{B}, a_{lp},  \pi^{BA}, \pi^{AB})}
	\right)
	\,\, .
\end{equation}

\subparagraph*{Solution Concept.}
In a Subgame-Perfect Nash Equilibrium (SPNE), players cannot gain higher utility by changing strategies at any step~\cite{selten1965spieltheoretische}, knowing that subsequent players will take their best responses. 

When a trader takes an action in a given step, her utility is influenced solely by her immediate strategy and the current state of AMMs, as outlined in Equation~\ref{utilityba}.
Crucially, future actions do not affect this calculation, allowing the trader to directly optimize her utility, thereby establishing dominant strategies.

In the case of a liquidity provider being chosen in a step, the situation is different.
Given their myopic viewpoint, liquidity providers only need to account for the strategy of the trader in the ensuing step. 
Their actions in subsequent steps do not affect their own utility. 
Thus, the sequential game is reduced to a two-stage Stackelberg game and SPNE to a Stackelberg Equilibrium~\cite{von2010market}.

To formalize this, we denote the strategies of the liquidity provider, the \textit{BA} trader, and the \textit{AB} trader in the SPNE by~$\tau_{lp}$,~$\tau^{BA}$ and~$\tau^{AB}$, respectively.
The \textit{BA} trader would always get the optimal utility in equilibrium, namely $\forall \poolcontentvector{}{}{}, \uac{}{}{\textit{BA}} $, we have ${U^{\textit{BA}}(\tau^{BA}, \poolcontentvector{}{}{}, \uac{}{}{\textit{BA}})
 = \max_{\pi^{BA}}U^{\textit{BA}}(\pi^{BA}, \poolcontentvector{}{}{}, \uac{}{}{\textit{BA}})}.$

Note that $\tau^{BA}$ is a best response for the \textit{BA} trader.
The strategy of liquidity provider in equilibrium is just the optimal strategy when traders adopt their best response, namely for all~$\poolcontentvector{}{}{}$, $\lpac{}{}{A}$ and~$\lpac{}{}{B}$, we have ${U_{lp}(\tau_{lp}, \poolcontentvector{}{}{}, \tau^{BA}, \tau^{AB},\lpac{}{}{A}, \lpac{}{}{B}) = 
	\max_{\pi_{lp}}U_{lp}(\pi_{lp}, \poolcontentvector{}{}{}, \tau^{BA}, \tau^{AB},\lpac{}{}{A}, \lpac{}{}{B}).}$


\subparagraph*{Game Assumptions.}
We briefly discuss three non-trivial model assumptions.

First, we did not explicitly consider liquidity removal in the game.
If a liquidity provider removes all her liquidity, it just changes the state of the AMM pools and our analysis still holds.
If the liquidity provider strategically removes liquidity, it is equivalent to removing all liquidity first and then adding liquidity again.
Therefore, we do not lose generality by ignoring liquidity removals.

Second, in a blockchain, each transaction consumes resources measured in a metric called \emph{gas}.
The transaction pays a fee according to its gas consumption. 
In practice, gas fees are volatile, often very low in high-performance blockchains (e.g., in Sui under $\$0.005$ per transaction on average~\cite{suigasfee,suiprice}).
The model conservatively neglects gas fees to avoid arbitrary assumptions on gas costs.
Gas fees disincentivize trade splitting, which helps with parallelization. 
In particular, with a set gas fee, Properties~\ref{concavity} and~\ref{smallerpoolsmallercost} continue to hold.
Therefore, it can only strengthen the equilibrium analysis.

Third, perfect arbitrageurs are a common assumption in the literature~\cite{milionis2022automated, goyal2023finding, canidio2023batching, chan2024mechanism}.
We nonetheless demonstrate later (\S\ref{sec:simulation}) that our results hold with real workloads, without this assumption.
In addition, we do not model arbitrage that trades across SAMM shards or between SAMM and other decentralized exchanges (DEXs).
We call this type of arbitrage \emph{DEX-to-DEX}.
Any DEX-to-DEX arbitrage can be replaced by two independent trades, each between a DEX and the external market, yielding the same payoff.
Therefore, DEX-to-DEX arbitrage is weakly dominated by arbitrages with the external market.
Empirically, DEX-to-DEX arbitrage accounts for less than $10\%$ of observed arbitrage activity~\cite{greenfieldStatArbAMMPart1}.

		\subsection{Desired Property}\label{GT:property}
Our goal is to improve the throughput by allowing parallelism. 
Specifically, we would like traders to evenly distribute their transactions among all AMM shards without splitting them.
That is, a dominant strategy for the \textit{BA} trader should be to randomly select an AMM shard to acquire all her needed \textit{token~A}.
Denote the action of getting all $\uac{}{}{\textit{BA}}$ \textit{token~A} in~$\textit{shard}_i$~by 
\begin{equation}\label{eq:single_action}
	a_i^{\textit{BA}}(\uac{}{}{\textit{BA}}) = \left(0, \cdots, \uac{i}{}{\textit{BA}} = \uac{}{}{\textit{BA}}, \cdots, 0\right) \,\, .
\end{equation}
Denote the set of these actions by $\SingleAction(\uac{}{}{\textit{BA}})  = \left\{ a_i^{\textit{BA}}(\uac{}{}{\textit{BA}})
	\left| 
		1 \leq i \leq n
	\right. \right\}  \subset\mathcal{A}_t^{\textit{BA}}(\uac{}{}{\textit{BA}}).$
The strategy that uniformly at random selects an AMM contract to acquire all her needed \textit{token~A} is the \emph{perfect parallelism strategy}:
\begin{restatable}{defi}{perfectparallelismstrategy}
	The perfect parallelism strategy of the \textit{BA} trader is
	\begin{equation*}
		\hat\tau^{BA}(\poolcontentvector{}{}{}, \uac{}{}{\textit{BA}}, a^{\textit{BA}}) = \begin{cases}
			\frac{1}{n}, & \text{if } a^{\textit{BA}} \in \SingleAction(\uac{}{}{\textit{BA}})  \\
			0, & \text{otherwise}
		  \end{cases} .
	\end{equation*}
\end{restatable}

Our goal is thus to have the perfect parallelism strategy be a dominant strategy:

\begin{restatable}{prop}{propertyperfectparallelism}
	\label{property_perfect_parallelism}
	The perfect parallelism strategy of the \textit{BA} trader is a dominant strategy:
	\begin{equation*}
		\forall \pi^{BA}:
		U^{\textit{BA}}(\pi^{BA}, \poolcontentvector{}{}{}, \uac{}{}{\textit{BA}}) 
		\leq 
		U^{\textit{BA}}(\hat\tau^{BA}, \poolcontentvector{}{}{}, \uac{}{}{\textit{BA}})
		\,\, .
	\end{equation*}
\end{restatable}

Using multiple CPMMs does not satisfy the perfect parallelism property and would be counterproductive: Each trader would split her transactions across all AMM contracts.
Thus, although the total number of trades increases due to parallelism, the satisfied trade demand is no higher than a single AMM contract and possibly lower since the total throughput might only increase sublinearly in the number of AMM contracts.
\fv{See more details in Appendix~\ref{app:proof_uniswap_splitting}.}{See more details in \fullver.}

			\subsection{SAMM Equilibrium}\label{sec:GT:equilibrium}
We analyze the SPNE of the above game.
We only provide a roadmap here, from trader strategy to liquidity provider strategy\fv{, and defer the full analysis to Appendix~\ref{app:gtAnalysis}.}{; the complete analysis is in \fullver.}

\subsubsection{Trader Strategy}
Consider the case that the system state is $\poolcontentvector{}{}{}= \left(
	\left(\poolcontent{1}{}{A},\poolcontent{1}{}{B}, \poolcontent{1}{}{S} \right),
	\cdots \left(\poolcontent{n}{}{A},\poolcontent{n}{}{B}, 
	\poolcontent{n}{}{S}	 \right) \right)$.
As discussed in Section~\ref{sec:SAMM_parameter}, the SAMM gross amount satisfies the $c$-non-splitting property and $c$-smaller-better property for a certain $0 < c < 1$.
We assume that the required amount of \textit{token~A}, $\uac{}{}{\textit{BA}}$, is at most a fraction~$c$ of the amount of deposited \textit{token~A} in all shards, i.e.,~${\forall 1 \leq i \leq n, \uac{}{}{\textit{BA}} \leq c\poolcontent{i}{}{A} \,\,}.$

The $c$-non-splitting property and $c$-smaller-better property give a trader the incentive to randomly select one of the smallest shards to trade all her required tokens.
Recall that $a_i^{\text{BA}}(\uac{}{}{\textit{BA}})$ is the action of acquiring all $\uac{}{}{\textit{BA}}$ \textit{token~A} in $\textit{shard}_i$ (Equation~\ref{eq:single_action}).
We define the set of actions that trade in one of the smallest shards:
\newcommand{\defsmallestactionset}{The \emph{Smallest Shard Action Set} is the set of actions that acquire all $\uac{}{}{\textit{BA}}$ \textit{token~A} in one of the smallest shards under state~$\poolcontentvector{}{}{}$: $\SmallestAction(\uac{}{}{\textit{BA}}, \poolcontentvector{}{}{}) = \left\{
	a_i^{\textit{BA}}(\uac{}{}{\textit{BA}})| \forall j, \poolcontent{i}{}{A} \leq \poolcontent{j}{}{A}
	\right\}.$}
\begin{restatable}{defi}{smallestpoolactionset}\label{def:smallest_action_set}
	\defsmallestactionset
\end{restatable}
The cardinality of $\SmallestAction(\uac{}{}{\textit{BA}}, \poolcontentvector{}{}{})$ is the number of smallest shards in $\poolcontentvector{}{}{}$.
We denote this by $\SmallestNum(\poolcontentvector{}{}{}) = \left|\SmallestAction(\uac{}{}{\textit{BA}}, \poolcontentvector{}{}{}) \right|$.

We prove that randomly selecting a shard to trade without splitting is a dominant strategy for the trader, and it is strictly better than splitting or trading in a larger shard:
\begin{restatable}{thm}{theoremstrictsingle}\label{theorem:strict_single}
	In $\Gamma_n(tf_{\textit{SAMM}})$, considering the following dominant strategy of the \textit{BA} trader which randomly selects one of the smallest shards to acquire all required tokens:
	\begin{align*}
		\tau^{BA}(\poolcontentvector{}{}{}, \uac{}{}{\textit{BA}}, a^{\textit{BA}}) = \begin{cases}
			\frac{1}{n_{\min}(\poolcontentvector{}{}{})}, & \text{if } a^{\textit{BA}} \in \SmallestAction(\uac{}{}{\textit{BA}}, \poolcontentvector{}{}{})  \\
			0, & \text{otherwise}
		  \end{cases},
	\end{align*}
	then for all strategies $\pi^{BA}$ that have a positive probability of actions not trading in one of the smallest shards, i.e., ${\exists a^{\textit{BA}} = \left(\uac{1}{}{\textit{BA}}, \cdots, \uac{i}{}{\textit{BA}}, \cdots, \uac{n}{}{\textit{BA}} \right)\notin \SmallestAction(\uac{}{}{\textit{BA}}, \poolcontentvector{}{}{})}$, ${\pi^{BA}(\poolcontentvector{}{}{}, \uac{}{}{\textit{BA}}, a^{\textit{BA}}) > 0}$, the utility of the \textit{BA} trader is strictly lower than with strategy $\tau^{BA}$:
	\begin{align*}
		U^{\textit{BA}}(\tau^{BA}, \poolcontentvector{}{}{}, \uac{}{}{\textit{BA}}) > U^{\textit{BA}}(\pi^{BA}, \poolcontentvector{}{}{}, \uac{}{}{\textit{BA}}) \,\, .
	\end{align*}
\end{restatable}
\begin{proof}[Proof Sketch]
	Due to the $c$-non-splitting property, trading in a single shard is strictly better than trading in multiple shards.
	Then the revenue of trading in one of the smallest shards is strictly better than that in any other shard due to the $c$-smaller-better property.
	\cronly{The complete proof is in \fullver.}
\end{proof}

If shards are balanced, which means all shards have the same amount of deposited tokens, we have $n_{\min}(\poolcontentvector{}{}{}) = n$.
Then it is a dominant strategy for the trader to randomly select one of the $n$ shards to trade, as we intended:
\newcommand{\contenttheoremperfectparallelism}{In $\Gamma_n(tf_{\textit{SAMM}})$, the system state is $\poolcontentvector{}{}{} = \left(
	\left(\poolcontent{1}{}{A},\poolcontent{1}{}{B}, \poolcontent{1}{}{S} \right),
	\cdots \left(\poolcontent{n}{}{A},\poolcontent{n}{}{B}, \poolcontent{n}{}{S}	 \right) \right)$.
If $\forall i, j, \poolcontent{i}{}{A} = \poolcontent{j}{}{A}$ and $\poolcontent{i}{}{B} = \poolcontent{j}{}{B}$, then the perfect parallelism strategy
\begin{equation*}
	\hat\tau^{BA}(\poolcontentvector{}{}{}, \uac{}{}{\textit{BA}}, a^{\textit{BA}}) = \begin{cases}
		\frac{1}{n}, & \text{if } a^{\textit{BA}} \in \SingleAction(\uac{}{}{\textit{BA}})  \\
		0, & \text{otherwise}
	  \end{cases}
\end{equation*}
is a dominant strategy for the \textit{BA} trader.}
\begin{restatable}{cor}{theoremperfectparallelism}\label{theorem:perfect_parallelism}
	\contenttheoremperfectparallelism
\end{restatable}

\subsubsection{Liquidity Provider Strategy}
Having shown that traders behave as intended if the shards are balanced, we consider the strategies of liquidity providers in equilibrium.
They should maintain the shard balance. 
Moreover, their incentives should rebalance the system state even in the face of attacks that break this balance.

We want a liquidity provider to fill up smaller shards to keep shards balanced.
We call such an action the \emph{fillup action}, where if the liquidity provider adds tokens to a shard, then the shard is the smallest shard after this action.
We denote the fillup action by~$a_{lp}^{\textit{fill}}(\poolcontentvector{}{}{}, \lpac{}{}{A}, \lpac{}{}{B})$.
Note that there is only one such action.

\newcommand{\deffillupaction}{The fillup action of a liquidity provider $a_{lp}^{\textit{fill}}(\poolcontentvector{}{}{}, \lpac{}{}{A}, \lpac{}{}{B}) = \left(\left(\fillupac{1}{}{A}, \fillupac{1}{}{B}\right), \cdots, \left(\fillupac{n}{}{A}, \fillupac{n}{}{B}\right)\right)$ is the action where if the liquidity provider adds tokens to a shard, then the shard is one of the smallest shards after this action:	
$$	\forall 1 \leq i \leq n:\;\;(1)\;\;\fillupac{i}{}{A} \geq 0,
	\sum_{i=1}^n\fillupac{i}{}{A} = \lpac{}{}{A}\,\, ,
	(2)\;\;\forall \fillupac{i}{}{A}>0, \forall j: \fillupac{i}{}{A} + \poolcontent{i}{}{A} \leq \fillupac{j}{}{A} + \poolcontent{j}{}{A} \,\, .
$$}
\begin{restatable}{defi}{deffillup_action}\label{def:fillup_action}
	\deffillupaction
\end{restatable}

The \emph{fillup strategy} only uses the fillup action: 
\newcommand{\deffillupstrategy}{The fillup strategy of a liquidity provider $\tau^{\textit{fill}}_{lp}(\poolcontentvector{}{}{}, \lpac{}{}{A}, \lpac{}{}{B})$ is the strategy that only takes the fillup action: 
\begin{equation*}
	\tau^{\textit{fill}}_{lp}( \poolcontentvector{}{}{}, \lpac{}{}{A}, \lpac{}{}{B}, a_{lp}) = \begin{cases}
		1, & \text{if } a_{lp} = a_{lp}^{\textit{fill}}(\poolcontentvector{}{}{}, \lpac{}{}{A}, \lpac{}{}{B})  \\
		0, & \text{otherwise}
	  \end{cases} \,\, .
\end{equation*}}
\begin{restatable}{defi}{fillupstrategy}\label{def:fillup_strategy}
	\deffillupstrategy
\end{restatable}

\begin{figure}[t]
	\centering
	\includegraphics[width=0.63\textwidth]{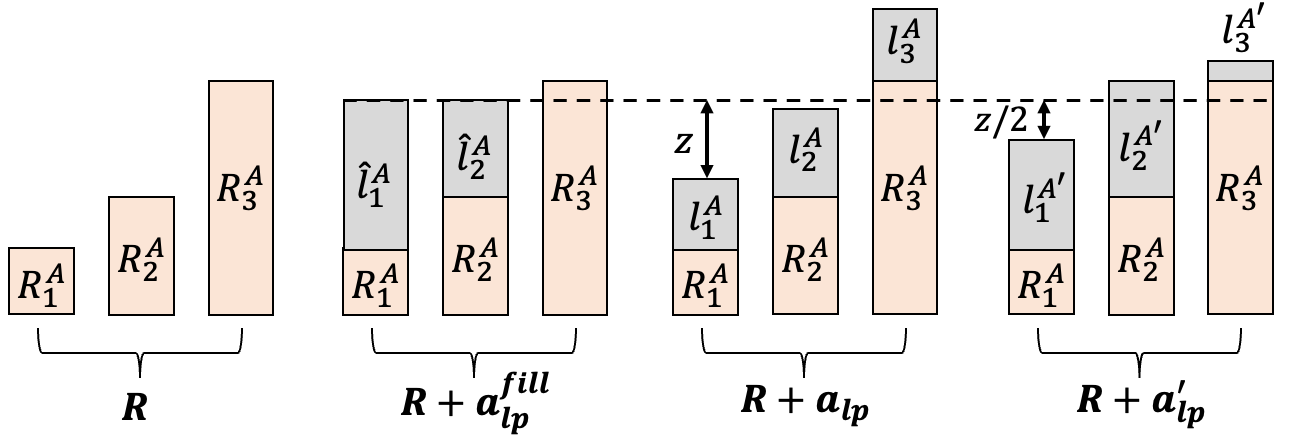}
	\caption{An example construction of $a_{lp}'$.}
	\label{fig:action}
\end{figure}

We first prove that, when all shards have identical sizes, the fillup action is to add tokens to all shards evenly, which is the best response of the liquidity provider:
\newcommand{\contenttheoremperfectparallelismandlp}{Denote by $\hat{a}_{lp} = \left(\left(\frac{1}{n}\lpac{}{}{A},  \frac{1}{n}\lpac{}{}{B}\right),\cdots\right)$ the action of evenly depositing tokens in all shards.
	In $\Gamma_n(tf_{\textit{SAMM}})$, if for all $i$ and $j$ that the liquidity amounts are the same, $\poolcontent{i}{}{A} = \poolcontent{j}{}{A}$ and $\poolcontent{i}{}{B} = \poolcontent{j}{}{B}$, the liquidity provider strategy which only takes action $\hat{a}_{lp}$,
	\begin{equation*}
		\tau_{lp}( \poolcontentvector{}{}{}, \lpac{}{}{A}, \lpac{}{}{B}, a_{lp}) = \begin{cases}
			1, & \text{if } a_{lp} = \hat{a}_{lp}  \\
			0, & \text{otherwise}\nonumber 
		  \end{cases}\nonumber \\
	\end{equation*}
	and any best response of the trader constitutes an SPNE.}
\begin{restatable}{thm}{theoremperfectparallelismandlp}\label{theorem_perfect_parallelism_and_lp}
	\contenttheoremperfectparallelismandlp
\end{restatable}
\begin{proof}[Proof Sketch]
	Traders prefer trading in smaller shards to reduce their costs.
	At the same time, fees are higher in larger shards.
	So the liquidity provider should increase her share in the smallest shards. 
	This dual objective is optimally achieved by uniformly distributing tokens across all shards.
	\cronly{The complete proof is in \fullver.}
\end{proof}

The above theorem indicates that once the system reaches a balanced state, this state is stable.
We now show that even if the system reaches an unbalanced state, maybe due to an attacker, it converges to a balanced state since the liquidity provider uses the fillup strategy.

We show that the fillup strategy is the only best response in all SPNE:

\newcommand{\contenttheoremalwaysfillup}{In $\Gamma_n(\textit{tf}_{\textit{SAMM}})$, in all SPNE, the liquidity provider's best response is the fillup strategy:
	\begin{equation*}
		\tau^{\textit{fill}}_{lp}( \poolcontentvector{}{}{}, \lpac{}{}{A}, \lpac{}{}{B}, a_{lp}) = \begin{cases}
			1, & \text{if } a_{lp} = a_{lp}^{\textit{fill}}(\poolcontentvector{}{}{}, \lpac{}{}{A}, \lpac{}{}{B})  \\
			0, & \text{otherwise}
		  \end{cases}
		\,\,.
	\end{equation*}}

\begin{restatable}{thm}{theoremalwaysfillup}\label{theorem:always_fillup}
	\contenttheoremalwaysfillup
\end{restatable}

\begin{proof}[Proof Sketch]
	Given any action $a_{lp}$ that is not the fillup action, we can construct a new action $a_{lp}'$ that is strictly better than~$a_{lp}$.
	By ensuring the smallest shard in $\poolcontentvector{}{}{} + a_{lp}'$ is larger than that in $\poolcontentvector{}{}{} + a_{lp}$, we increase trading fees garnered from each transaction. 
	Moreover, this smallest shard is also the smallest before the action, maximizing the liquidity provider's share. 
	Consequently, the liquidity provider earns higher revenue under $a_{lp}'$ than under $a_{lp}$. 
	Thus, any strategy incorporating an action other than the fillup action is not optimal. Figure~\ref{fig:action} illustrates an example of the~$a_{lp}'$ construction.
	\cronly{The complete proof is in \fullver.}
\end{proof}

Theorem~\ref{theorem:always_fillup} does not prove the existence of an SPNE.
We prove one exists by constructing an SPNE for a general (perhaps unbalanced) starting state~\fv{(\S\ref{app:gtAnalysis})}{(see \fullver)}: Liquidity providers follow the fillup strategy, and traders use the smallest shard with the largest share of the previous liquidity provider.

In summary, we showed that the system achieves a stable state with perfect parallelism.
After events that create heterogeneous shard sizes (like imbalancing liquidity additions by an attacker or asymmetric liquidity removals by liquidity providers), liquidity providers rebalance shards once sufficient liquidity has been provided.
Malicious attackers incur ongoing losses, and when their budget is exhausted, the system returns to a balanced state.
These results indicate that SAMM is robust to such attacks.


    \section{Evaluation}\label{sec:evaluation}

To evaluate the performance, we use the state-of-the-art blockchains Sui (\S\ref{sec:evaluation_setup}) and Solana\fvonly{~(\S\ref{app:solana})}.
We find that the throughput of a single-contract AMM is limited (\S\ref{sec:evaluation_single}), and that it improves with the number of SAMM shards (\S\ref{sec:evaluation_samm}).
Our analysis shows that further improvement is possible by increasing the platform's parallelism~(\S\ref{sec:evaluation_theory}).


		\subsection{Experimental Setup}\label{sec:evaluation_setup}

We conduct experiments on the Sui~\cite{blackshear2023sui} blockchain.
Smart contracts in Sui are independent \emph{objects}, and Sui executes transactions on different objects in parallel.
We implement SAMM\footnote{\url{https://github.com/MountainGold/SAMM-Sui-Evaluation}} in the Move language~\cite{blackshear2019move}.
We deploy a local testnet, which follows the default configuration, consisting of~4 validators maintaining consensus of the blockchain.
The latency of transactions is always higher than~1 second due to Sui's consensus protocol.
We issue transactions using Sui's Rust RPC interface. 

As a baseline, we test the latency of simple token transfers. 
As expected, unencumbered by smart-contract coordination constraints, the latency is consistently smaller than 200~msec (Figure~\ref{fig:multiple_latency}) in Sui at $2360\textit{tps}$ (outside the figure range).
This throughput is about twice the maximum rate observed in our AMM experiments on Sui.
The latency remains under one second because token transfers in Sui do not require immediate consensus among validators.

We also conduct experiments on Solana~\cite{yakovenko2018solana}, which has a different architecture and consensus protocol.
Details of the Solana experimental setup are in \fv{Appendix~\ref{app:solana}}{\fullver}.

For all reported results, we use a machine with~2~TB of memory and~256 CPU cores.
We run~100 trader processes for Sui (and 50 for Solana).
Several experiments with more traders or with the testnet deployed and transactions sent from separate machines produce similar results.
This shows that workload generation does not introduce bottlenecks.

Each trader sends transactions at random intervals following an exponential distribution with an expected frequency of~$\lambda$.
A trader can send another transaction before the previous one is confirmed.
Note this experiment is only for performance evaluation, so traders follow the perfect parallelism strategy. 
We vary the overall frequency of transactions by setting different values of the individual $\lambda$ values.
In each test, we set a target throughput. 
We first warm up the system by sending transactions for $500$ seconds and then measure actual throughputs and latencies for the following $100$ seconds. 
If the latency is stable during the~100 measurement seconds, we report the mean value.


		\subsection{Single-Contract Bottleneck}
		\label{sec:evaluation_single}

To demonstrate the bottleneck of a single AMM, we first deploy a standard CPMM.
We use the OmniSwap contract~\cite{OmniSwap} in Sui and the token-swap contract from the Solana Program Library~\cite{spl2024}, which are generalizations of Uniswap~v2.

For the Sui experiment, Figure~\ref{fig:multiple_latency} (${n=1}$) shows latencies of transaction processing (Y~axis) in workloads with varying transaction frequencies (X~axis) using a single OmniSwap contract.
We test each frequency~5 times and calculate the truncated average, excluding the two extreme values.
Error bars show all measured values.
The average latency increases gradually with the transaction frequency up to~$214\textit{tps}$ before crossing the 3-second line.
With higher frequencies, transactions frequently fail.

Our Solana experiments produce similar results~\fv{(\S\ref{app:solana})}{(see \fullver)}.
The throughput bottleneck of a single CPMM contract in both Solana-Default and Solana-NoBlockLimit is $129\textit{tps}$.


			\subsection{SAMM Evaluation}
			\label{sec:evaluation_samm}

We implement SAMM by modifying the fee mechanism of CPMM contracts and deploying a varying number of shards (contracts).
Each trader randomly selects a shard to trade.

Figure~\ref{fig:multiple_latency} reports the average latency under varying demand (\textit{tps}) across different shard counts in our Sui experiments.
The latency in the case of one SAMM contract is indistinguishable from a single CPMM contract.
In some instances, more than half of the transactions failed, which is marked with~$\times$ in the graph.
As the number of SAMM contracts increases, the system can process higher demand.

To quantify SAMM's performance enhancements in Sui, we evaluate the throughput with a varying number of shards. 
We set a latency cap of $3$ seconds.
As depicted in Figure~\ref{fig:multiple_latency}, the latency rises rapidly once it surpasses $2$ seconds.
Therefore, choosing other latency caps beyond $2$ seconds does not significantly affect the results.
The throughput with~$n$ shards is thus the highest frequency that produced a latency lower than 3 seconds.
Figure~\ref{fig:multiple_tps} shows the throughput increases almost linearly in the beginning and then converges to a bound.
With $32$ shards, the maximal throughput exceeds $1185\textit{tps}$, more than five times the throughput of a single OmniSwap contract.

On Solana, we test two configurations: \emph{Solana-Default} uses the standard mainnet gas limits, while \emph{Solana-NoBlockLimit} removes the per-block gas cap to expose the architecture's limit \fv{(\S\ref{app:solana})}{(see \fullver)}.
Figure~\ref{fig:multiple_tps_solana} shows that the maximal throughput scales almost linearly up to~$4$ shards under Solana-Default and up to~$16$ shards under Solana-NoBlockLimit, with $17$ shards failing to match the \textit{tps} of $16$ shards, marked with~$\times$.

\begin{figure}[t]
	\centering
	\begin{subfigure}{0.5\textwidth}
		\centering
		\includegraphics[width=\linewidth]{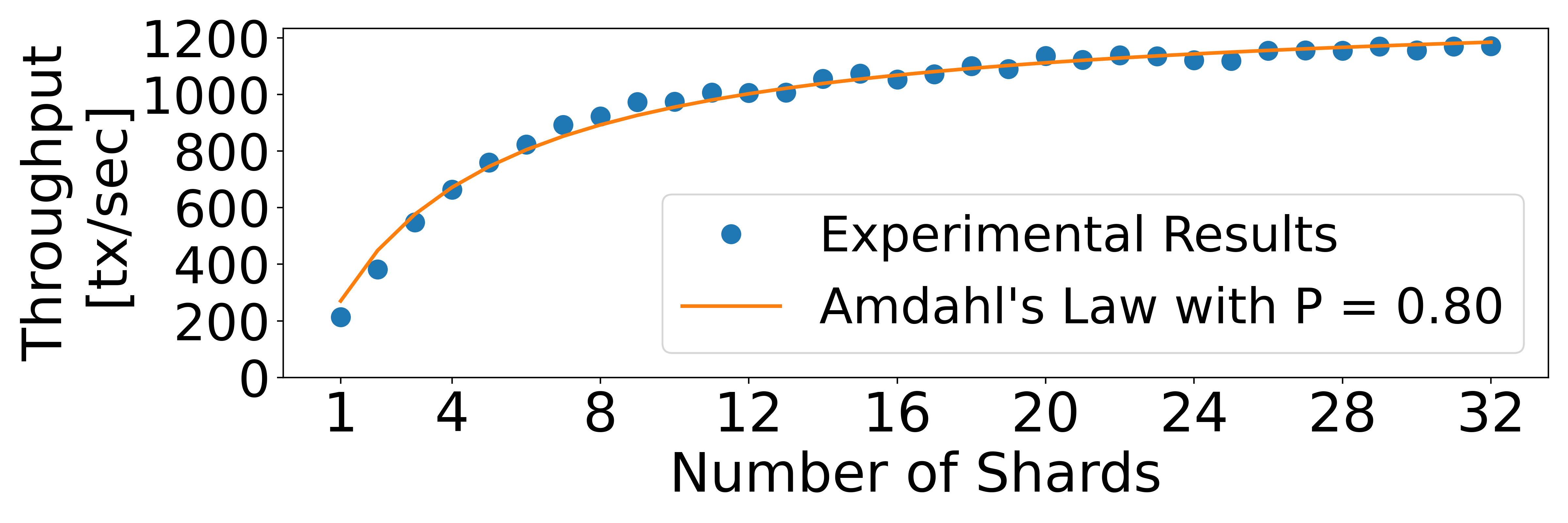}
		\caption{SAMM's experimental throughput in Sui.}
		\label{fig:multiple_tps}
	\end{subfigure}%
	\hfill
	\begin{subfigure}{0.5\textwidth}
		\centering
		\includegraphics[width=\linewidth]{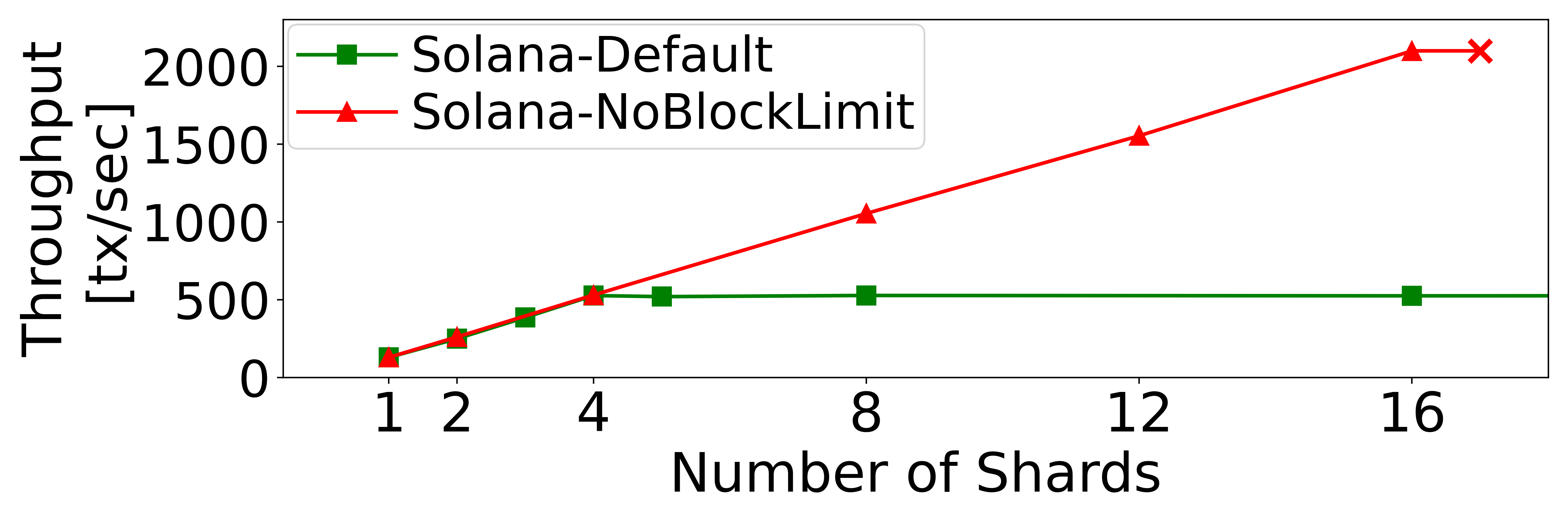}
		\caption{SAMM's experimental throughput in Solana.}
		\label{fig:multiple_tps_solana}
	\end{subfigure}
	\caption{Maximal throughput as a function of the number of SAMM shards in Sui and Solana.}
	\label{fig:samm_tps}
\end{figure}


			\subsection{Parallelization in the Underlying Platform}\label{sec:evaluation_theory}

When there are $n$ concurrently operating AMM shards, each with a maximum throughput of $T_{\text{max}}$, the total maximum throughput,~$T_{\text{total}}(n)$, is influenced by the fraction $P$ of the transaction that can be parallelized, according to Amdahl's Law. 
This law states that the speedup ratio, $S(n)$, is the total throughput relative to a single AMM's throughput, according to the expression $S(n) = \frac{1}{(1-P) + \frac{P}{n}}$. 
For fully sequential systems like Ethereum, $P = 0$, resulting in no throughput gain.
For stability, Solana sets a gas limit that keeps throughput well below its serial bottleneck. 
Its throughput is thus linear in~$n$, up to the artificial limit (Figure~\ref{fig:multiple_tps_solana}). 

In Sui, we fit the experimental data with Amdahl's law and conclude the parallelizable part of a transaction is ${P = 0.80}$ (${R^{2}=0.99}$).
Since the serial components of transactions are invariant, throughput improvements are inherently limited. 
According to the fitted curve, the throughput is bounded by $1330\textit{tps}$.
We find that if transactions in Sui have a larger parallelizable portion, the throughput improvement due to SAMM is even better \fv{(\S\ref{app:evaluation_heavier})}{(see \fullver)}.

\section{Incentive-Compatibility Validation and Secondary Effects}\label{sec:simulation}

We simulate SAMM on real trading data~(\S\ref{sec:simulation:setup}): Incentive compatibility holds even without assuming perfect arbitrageurs~(\S\ref{sec:simulation:distribution}).
SAMM also raises liquidity-provider revenue at trader costs comparable to a standard CPMM, without incentivizing sandwich attacks or amplifying losses and overhead under price fluctuations~(\S\ref{sec:simulation:secondary}).


			\subsection{Simulation Setup}
			\label{sec:simulation:setup}

We simulate SAMM with workloads from public CPMM trading data, namely the Uniswap~v2 Ethereum AMM for the pair of tokens USDC and ETH (hereinafter the \emph{workload}). 
We assume both Uniswap and SAMM operate at the maximal throughput measured in our performance evaluation. 
In each test, we simulate over 1 second, following a 1-second warm-up period.
We repeat each simulation 100 times and calculate averages.
Note that a $1$-second simulation corresponds to several hours' trading in the real world.
We conservatively choose the throughput of SAMM according to the performance evaluation of Sui where throughput improves sublinearly with the number of shards.

We simulate SAMM using~1 to~32 shards and Uniswap.
For SAMM with~$n$ shards, we evenly distribute the liquidity of the workload among the~$n$ shards to establish the initial state of SAMM. 
We test three different values of~$c$, namely~$0.003$, $0.005$, and~$0.01$. 
We still do not consider the gas fee in the simulation, as it would only make the results better since it discourages trade splitting (\S\ref{GT:game_model}).

In each test, we replay the trades in the workload.
Each trader selects the shard offering the lowest price and may split the trade into several smaller trades if this reduces her costs.

Our simulation differs from the game-theoretic analysis in two key aspects.
First, the simulation does not assume perfect arbitrageurs, as there are no synthetic arbitrage transactions.
Second, the trade amount of some transactions exceeds the $c$ fraction of a single shard's liquidity.
This is particularly relevant in our simulation, where each shard holds only $1/n$ of the total liquidity, occasionally resulting in trade splits.
See more details in \fv{Appendix~\ref{app:simulation_setup}}{\fullver}.

\begin{figure}[t!]
	\centering
	\includegraphics[width=\linewidth]{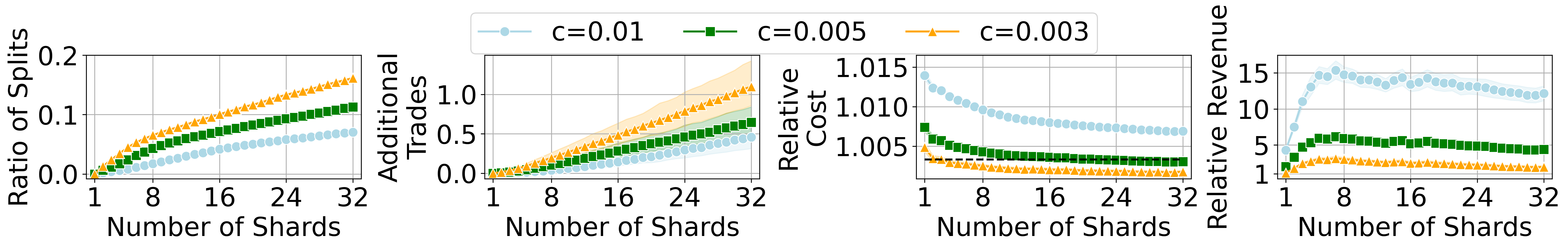}\\[-12pt]
	\subcaptionbox{Ratio of transaction splits.\label{fig:splits:ratio}}[0.24\linewidth]{}\hfill
	\subcaptionbox{Additional trades due to splits.\label{fig:splits:trades}}[0.24\linewidth]{}\hfill
	\subcaptionbox{Traders' costs relative to the external market.\label{fig:tradercost}}[0.24\linewidth]{}\hfill
	\subcaptionbox{Relative lp's revenue compared to Uniswap.\label{fig:lprevenue}}[0.24\linewidth]{}
	\caption{Trade splits, trader costs, and lp's (liquidity providers') revenue in SAMM.}
	\label{fig:samm_combined}
\end{figure}

			\subsection{Incentive-Compatibility Validation}
			\label{sec:simulation:distribution}

We analyze the proportion of trades that are split into multiple transactions within SAMM.
With $c=0.01$, trade splits are infrequent, occurring in less than $8\%$ of trades, even with 32 shards (Figure~\ref{fig:splits:ratio}). 
Furthermore, the total number of trades (the trader might split a trade into more than~2 parts) results in less than a~$45\%$ increase in overall trade count~(Figure~\ref{fig:splits:trades}), a relatively minor increment given the five-fold throughput improvement.
When the number of shards is larger or $c$ is smaller, the proportion of trade splits is higher.
We also confirm that the distribution of trades in SAMM is balanced~\fv{(\S\ref{app:simulation_distribution})}{(see \fullver)}.
This confirms that our theoretical analysis holds without assumptions of perfect arbitrageurs or trade size limits.

\subsection{Secondary Effects of Sharding}\label{sec:simulation:secondary}

Sharding results in smaller shards, since liquidity is split.
This results in higher slippage for trades.
We analyze how this affects participant revenues and AMM vulnerabilities.

\subparagraph*{Traders' Costs and Liquidity Provider Revenue.}
We now analyze how sharding affects the revenues of traders and liquidity providers.
In a trade with a set output amount and fixed deposited amounts of tokens in a shard, adjusting the fee parameters is zero-sum for the trader and liquidity provider: with higher fees, the liquidity provider gains more revenue, while the trader incurs higher costs (Equation~\ref{cost_function}).
However, sharding allows more trades to be processed, so liquidity providers can gain more revenue while each trader's cost still remains comparable to that in a standard CPMM.

Our simulation also validates this economic impact.
Figure~\ref{fig:tradercost} illustrates traders' costs: while increasing with the fee parameter~$c$, costs decrease as the number of shards increases, remaining comparable to Uniswap or even lower. 
This shows how higher throughput can offset increased slippage. 
Figure~\ref{fig:lprevenue} reveals a significant increase in total liquidity provider revenue; it climbs with the number of shards (reaching up to 15 times that of Uniswap) due to more trades, before eventually declining as individual fees per trade diminish in very small shards. 
These trends demonstrate a key advantage of SAMM: It increases the exchange volume, benefiting all participants, and is not a zero-sum game that simply reallocates costs.
Hence, SAMM attracts participants when competing with other AMMs.
More mathematical analysis and empirical details are provided in \fv{Appendix~\ref{app:revenue}}{\fullver}.

We also observe that SAMM outperforms Uniswap v2 on \emph{volume capacity}, which measures the ability of handling large trades with small price impact \fv{(\S\ref{app:volume_capacity})}{(see \fullver)}.

The analysis of revenue, volume capacity and incentive-compatibility (\S\ref{sec:simulation:distribution}) implies a trade-off between participant incentives and system performance.
A larger $c$ results in fewer trade splits and higher revenue for liquidity providers, but the cost for traders increases.
Fixing $c$ and increasing the number of shards results in lower trader costs but more trade splits, which can negatively impact system performance.
Additionally, liquidity provider revenue peaks at a certain number of shards.
System designers can tune~$c$ and the number of shards to balance participant incentives and system performance.

\subparagraph*{Sandwich Attacks and Price-Fluctuation Losses and Overhead.}
A critical consideration is whether the increased slippage inherent in smaller shards exacerbates known AMM vulnerabilities, specifically sandwich attacks and losses due to price fluctuations. 
Our analysis demonstrates that, counter-intuitively, the profitability of sandwich attacks decreases as liquidity decreases, rendering smaller shards less sensitive \fv{(\S\ref{app:sandwich})}{(see \fullver)}.
When the price of tokens in the external market fluctuates, arbitrageurs can arbitrage between the external market and the AMM, causing losses for the liquidity providers, called loss-versus-rebalancing (LVR)~\cite{milionis2022automated, fritsch2024measuring, milionis2023automated}.
For CPMM-based AMMs like SAMM, with perfect arbitrageurs, the loss per unit of liquidity is independent of the size of the AMM~\cite[Example 3]{milionis2022automated}, meaning the loss ratio is unaffected by sharding.
Thus, SAMM enhances throughput without amplifying these economic risks.

Note that the fluctuation of token prices in the external market may lead to arbitrage opportunities, and with more shards there are more such opportunities.
However, in the worst case, each external price move triggers one arbitrage per shard per block. 
An $n$-shard system (say, $n=32$) delivering hundreds of \textit{tps} would suffer at most $n$ arbitrage transactions per block. 
This overhead is minor.

    \section{Conclusion} \label{sec:conclusion}
We present SAMM, a scalable AMM that employs smart contract sharding.
SAMM addresses incentive-driven security challenges through a novel trading fee function, which we validate via game-theoretic analysis and simulations on real trade traces.
We implement and deploy SAMM on local Sui and Solana testnets, demonstrating throughput improvements of over 5x and 16x, respectively, up to the limits of the underlying platforms. 
SAMM is directly applicable and has already attracted industry interest.
Nonetheless, how to choose and adapt the fee parameters and the threshold~$c$ as trade-size distributions shift, and how sensitive the system is to their miscalibration, remain to be studied and verified in real-world deployments.
Our results provide the first demonstration that high-demand smart contracts can be securely scaled via sharding.
They also highlight the importance of minimizing serial bottlenecks in smart-contract platforms to enable further scaling \fv{(See Appendix~\ref{app:evaluation_heavier})}{(see \fullver)}.

\section*{Acknowledgments}
This work was supported in part at the Technion by Aly Kaufman Fellowship, IC3, the Sui Foundation, and the Avalanche Foundation.

\bibliography{reference}

\iffullversion

\appendix

	\section{Smart Contract Sharding and Effects} \label{sec:contracts}

	Some smart contracts are naturally composed of independent components that can be partitioned into multiple independent instances, called \emph{shards}, with each shard running in parallel. 
	For example, \emph{Centralized Limit Order Book} (CLOB) exchanges partition orders by trading pair, with each pair managed by an independent shard~\cite{deepbookV3SuiDocs,phoenixCli}. 
	Similarly, AMMs typically manage different token pairs in separate pools, which can be executed as independent shards~\cite{orcaWhirlpools,cetusDeveloperDocs}.

	A contract may also contain multiple components that are naturally isolated and can be sharded.
	For instance, \emph{lending protocols}~\cite{morphoMarketDocs,kaminoU,aaveV3TechPaper} and \emph{CDP (Collateralized Debt Position)-based stablecoins}~\cite{bucketProtocolDocs,hubbleProtocolDocs} commonly maintain per-asset \emph{vaults} (reserves) with mostly independent accounting.
	Operations confined to a single vault can be parallelized.
	However, some operations remain inherently global, such as actions that depend on aggregate quantities (e.g., total debt or total supply), including borrowing and liquidation logic in lending markets and mint/burn logic in stablecoins.

	These two patterns typically do not introduce new security challenges because the shards correspond to \emph{pre-existing} separations of state and functionality: each shard serves a distinct use case (e.g., a specific trading pair or asset market), and users choose shards based on that external differentiation rather than cross-shard effects.
	As a result, sharding primarily improves performance without materially changing the incentive structure.

	In contrast, to parallelize the sequential core of a high-demand contract, one may split a \emph{single logical service} into multiple \emph{identical} shards.
	Here, shards are interchangeable and serve the same purpose, so users can choose among shards strategically.

	However, strategic behavior can directly undermine parallelization: Participants may concentrate activity on a subset of shards or split activity across shards to exploit shard-specific pricing or state.
	This behavior is naturally incentivized in several settings.
	In lending, borrowers may spread positions across shards to obtain lower rates if the interest rate increases with utilization or total borrows.
	In CLOBs, traders may split orders across shards to reduce market impact and information leakage.
	Mitigating these game-theoretic risks is necessary for sharding to scale with demand in practice.

	At present, sharding is not yet urgent for lending protocols, CDP-based stablecoins, and CLOBs, since their sequential demands (below~$1\textit{tps}$, \S\ref{app:demand}) are still far from saturating single-thread execution.
	We therefore focus on AMMs, whose demand is growing rapidly and may soon reach the single-thread performance limit~(\S\ref{app:growingdemand}).

\section{Demand for Popular Smart Contracts}\label{app:demand}

We analyze the monthly average demand for popular smart contracts, including \emph{CDP (Collateralized Debt Position)-based stablecoins}~\cite{bucketProtocolDocs,hubbleProtocolDocs} that maintain a peg to a fiat currency, \emph{lending protocols}~\cite{morphoMarketDocs,kaminoU,aaveV3TechPaper} that enable users to lend and borrow crypto assets, and \emph{centralized limit order books} (CLOBs)~\cite{deepbookV3SuiDocs,phoenixCli}, which partition orders across trading pairs. 
We use data from Dune (\url{https://dune.com/}), a comprehensive database for blockchain data.

For stablecoins, most demand comes from parallelizable token transfers. 
In contrast, non-parallelizable operations such as minting and burning occur at rates well below $1\textit{tps}$.
For lending protocols, we study Aave~\cite{aaveV3TechPaper}, the most widely used lending protocol, across all major versions on Ethereum, Polygon, and Arbitrum. The dominant user actions are non-parallelizable: deposits and withdrawals account for most activity, yet remains below $1\textit{tps}$ in total.
For CLOBs, overall activity can be high, but demand is typically spread across many trading pairs, making per-pair demand low and largely parallelizable. 
We evaluate Phoenix v1 on Solana~\cite{phoenixCli} and find that even its most active trading pair has monthly average demand below $1\textit{tps}$.

\section{Growing Demand for AMM}\label{app:growingdemand}
\begin{figure}[t]
	\centering
	\includegraphics[width=0.8\textwidth]{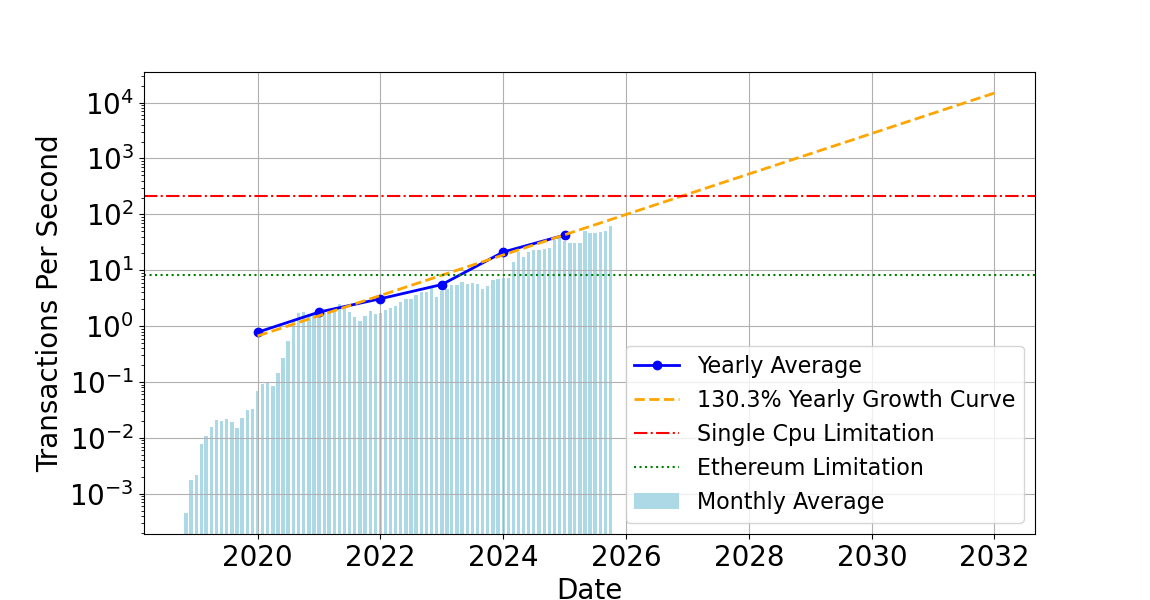}
	\caption{Monthly trades in Uniswap $1$-$3$}
	\label{fig:growingdemand}
\end{figure}

We analyze the demand for AMMs by calculating the number of trades per second on the prominent Uniswap versions 1-3, from its deployment in November 2018 to March 2024. 
As in Appendix~\ref{app:demand}, we use data from Dune.
Figure~\ref{fig:growingdemand} illustrates the exponentially increasing demand, from $0.78$ average trades per second in 2020 to $9.54$ in 2024.
On the Ethereum blockchain, most Uniswap transactions are executed through versions 2 and 3, incurring gas costs of \num{152809} and \num{184523} respectively. 
Given that the average total gas per block is \num{15000000} and the block interval is $12$ seconds, Ethereum can facilitate up to $8.18$ trades per second for v2 and $6.77$ trades per second for v3. 
However, in March 2024, the average monthly trades on Uniswap reached $13.9$ per second, nearly doubling Ethereum's processing capacity.
Therefore, most demand of Uniswap is processed through off-chain solutions~\cite{adams2024layer}.

The demand curve matches an exponential function, reflecting its rise in popularity since 2020 and consistent growth rate, yielding a yearly demand growth of $130.3\%$ ($R^2 = 0.999$).
The fitted curve suggests that demand for Uniswap will surpass the single CPU processing capacity of $214\textit{tps}$ by 2027. 
As we show this is beyond what even the state-of-the-art Sui can sustain.

Since AMMs handling different token pairs are independent, their executions can also be parallelized.
We analyze the ratio of the maximum throughput required by a single token pair and find that, within a given month, the peak hourly trades for a single pair account for approximately $30\%-60\%$ of the total in Uniswap v2, and $13\%-53\%$ across all Uniswap implementations.
This indicates that, even with parallelization across different token pairs, the point at which a single pair becomes the throughput bottleneck is only postponed by about 1 to 2 years.


\section{Constant Product Market Maker (CPMM)} \label{app:preliminaries} 

The basis of SAMM is the Constant Product Market Maker AMM (CPMM, e.g.~\cite{adams2021uniswap,OmniSwap, SushiSwap, Curve}). 
We review its operation here. 
It uses a share-based solution to manage liquidity addition and removal operations (\S\ref{sec:preliminaries:liquidity}) and keeps the product of the deposited amount of two tokens constant in trade operations~(\S\ref{sec:preliminaries:trade}).
Liquidity providers earn revenue from the trading fees~(\S\ref{sec:preliminaries:tf}).


			\subsection{Liquidity Addition and Removal}\label{sec:preliminaries:liquidity}
Most CPMMs (e.g.,~\cite{zhang2018formal, adams2020uniswap, OmniSwap, QuickSwap}) use a fungible share token to manage liquidity addition and removal.
These tokens represent a liquidity provider's share in the AMM.

When liquidity providers add tokens to an AMM, they receive \emph{share tokens} which signify their portion of the AMM. 
To recall, $\poolcontent{}{}{A}$ and $\poolcontent{}{}{B}$ denote the amounts of \textit{token~A} and \textit{token~B} already deposited in the AMM.
Similarly, $\poolinput{}{}{A}$ and $\poolinput{}{}{B}$ represent the quantities of \textit{token~A} and \textit{token~B} that the liquidity provider contributes through a liquidity addition operation.
Let $\poolcontent{}{}{S}$ represent the total amount of all share tokens distributed before the operation.
The amount of share tokens acquired by the liquidity provider in this operation, $\pooloutput{}{}{S}$, is given by 
\begin{equation*}
	\pooloutput{}{}{S} = \poolcontent{}{}{S} \times \min\left\{\frac{\poolinput{}{}{A}}{\poolcontent{}{}{A}}, \frac{\poolinput{}{}{B}}{\poolcontent{}{}{B}} \right\} \,\, .
\end{equation*}

The term $\min\left\{\frac{\poolinput{}{}{A}}{\poolcontent{}{}{A}}, \frac{\poolinput{}{}{B}}{\poolcontent{}{}{B}} \right\}$ signifies the ratio of the input token to the deposited token. 
The $\min$ function serves to ensure that the ownership accurately reflects the liquidity provider's contribution relative to the scarcer asset. 
It prevents situations where a liquidity provider adds a large amount of a certain token to unfairly obtain a larger share of tokens in the AMM.

Liquidity providers have the option to withdraw tokens from the AMM with the liquidity removal operation, which takes the number of share tokens as input and outputs \textit{token~A} and \textit{token~B}. 
Let $\poolinput{}{}{S}$ represent the amount of input share tokens, and $\poolcontent{}{}{S}$ denote the total amount of all share tokens referred to the AMM before the execution. 
The amounts of \textit{token~A} and \textit{token~B} withdrawn are $\pooloutput{}{}{A}$ and $\pooloutput{}{}{B}$:
\begin{equation*}
	\pooloutput{}{}{A}= \frac{\poolinput{}{}{S}}{\poolcontent{}{}{S}} \times \poolcontent{}{}{A}, 
	\pooloutput{}{}{B}= \frac{\poolinput{}{}{S}}{\poolcontent{}{}{S}}\times \poolcontent{}{}{B} \,\, .
\end{equation*}

Note that $\poolcontent{}{}{S}$ is not the amount of share tokens deposited in the AMM, but the total amount of share tokens owned by all liquidity providers. 

			\subsection{CPMM Trades}\label{sec:preliminaries:trade}
Recall that in a trade operation, a trader sends $\poolinput{}{}{A}$ \textit{token~A} (resp., $\poolinput{}{}{B}$ \textit{token~B}) and gets $\pooloutput{}{}{B}$ \textit{token~B} (resp., $\pooloutput{}{}{A}$ \textit{token~A}).
A trade is thus defined by a tuple~$(\poolinput{}{}{A}, \pooloutput{}{}{A}, \poolinput{}{}{B}, \pooloutput{}{}{B})$, where all values are non-negative, 
$
	\poolinput{}{}{A}, \pooloutput{}{}{A}, \poolinput{}{}{B}, \pooloutput{}{}{B} \geq 0 
$. 
The trade is either \textit{token~A} for \textit{token~B} or \textit{token~B} for \textit{token~A}, i.e., \(\poolinput{}{}{A} = \pooloutput{}{}{B} = 0\) or \(\poolinput{}{}{B} = \pooloutput{}{}{A} = 0\).

After the trade, the amount of deposited tokens is updated to $\poolcontent{}{}{A} + \poolinput{}{}{A} - \pooloutput{}{}{A}$ and $\poolcontent{}{}{B} + \poolinput{}{}{B} - \pooloutput{}{}{B}$, respectively.
Ignoring fees, the CPMM chooses the outputs~($\pooloutput{}{}{A}$ and~$\pooloutput{}{}{B}$) by setting an invariant called the \emph{trading function} $\Phi^{\textit{net}}_\textit{CPMM}(\poolcontent{}{}{A}, \poolcontent{}{}{B}, \poolinput{}{}{A}, \pooloutput{}{}{A}, \poolinput{}{}{B}, \pooloutput{}{}{B})$~\cite{angeris2020improved} which is the product of the amount of \textit{token~A} and \textit{token~B} after the trade, i.e.,
\begin{equation*} 
	\Phi^{\textit{net}}_\textit{CPMM}(\poolcontent{}{}{A}, \poolcontent{}{}{B}, \poolinput{}{}{A}, \pooloutput{}{}{A}, \poolinput{}{}{B}, \pooloutput{}{}{B}) \coloneqq 
	 (\poolcontent{}{}{A} + \poolinput{}{}{A} - \pooloutput{}{}{A}) \times (\poolcontent{}{}{B} + \poolinput{}{}{B}-\pooloutput{}{}{B}) \,\, . 
\end{equation*}

Note that $\Phi^{\textit{net}}_\textit{CPMM}(\poolcontent{}{}{A}, \poolcontent{}{}{B}, 0,0,0,0)$ is the product of the amount of \textit{token~A} and \textit{token~B} before the trade.

A trade $(\poolinput{}{}{A}, \pooloutput{}{}{A}, \poolinput{}{}{B}, \pooloutput{}{}{B})$ is legal if the trading function remains constant, i.e.,
\begin{equation} \label{leagaltrade1}
	\Phi^{\textit{net}}_\textit{CPMM}(\poolcontent{}{}{A}, \poolcontent{}{}{B}, \poolinput{}{}{A}, \pooloutput{}{}{A}, \poolinput{}{}{B}, \pooloutput{}{}{B})  
	= 
	\Phi^{\textit{net}}_\textit{CPMM}(\poolcontent{}{}{A}, \poolcontent{}{}{B}, 0,0,0,0) \,\, . 
\end{equation}

It indicates that the product of the amounts of \textit{token~A} and \textit{token~B} after the trade is the same as the product of the amounts before the trade, hence then names Constant Product Market Maker, i.e.,
\begin{equation} \label{leagaltrade2}
	(\poolcontent{}{}{A} +  \poolinput{}{}{A} - \pooloutput{}{}{A}) \times (\poolcontent{}{}{B}+ \poolinput{}{}{B}-\pooloutput{}{}{B}) 
	= \poolcontent{}{}{A} \times \poolcontent{}{}{B} \,\, . 
\end{equation}

If the trader gets $\pooloutput{}{}{A}$ \textit{token~A} (resp., $\pooloutput{}{}{B}$ \textit{token~B}), according to Equation~\ref{leagaltrade2}, she pays
\begin{align}\label{net_amount_pre}
	\poolinput{}{}{B} = \frac{\poolcontent{}{}{A} \times \poolcontent{}{}{B}}{\poolcontent{}{}{A} - \pooloutput{}{}{A}} - \poolcontent{}{}{B} =  \frac{\poolcontent{}{}{B} \times \pooloutput{}{}{A}}{\poolcontent{}{}{A} - \pooloutput{}{}{A}} \,\, , 
\end{align}
and similarly for an \textit{AB} trader. 

Note that we ignored fees in the above equations. 
We call the payment without fees (Equation~\ref{net_amount_pre}) \emph{net amount}, and denote by
\begin{equation*}
	\textit{net}(\poolcontent{}{}{A}, \poolcontent{}{}{B}, \pooloutput{}{}{A}) = \frac{\poolcontent{}{}{B} \times \pooloutput{}{}{A}}{\poolcontent{}{}{A} - \pooloutput{}{}{A}} \,\, .
\end{equation*}

Denote the amount of \textit{token~B} that the trader needs to pay to get a single \textit{token~A} by $p^{AB} = \frac{\poolinput{}{}{B}}{\pooloutput{}{}{A}}$.
From the above equation, $p^{AB} = \frac{\poolcontent{}{}{B}}{\poolcontent{}{}{A} - \pooloutput{}{}{A}}$.
This value increases as the output amount of \textit{token~A}, $\pooloutput{}{}{A}$, increases, which is the \emph{Slippage} of the trade.
When the output amount of \textit{token~A}, $\pooloutput{}{}{A}$, approaches zero, the token price is not influenced by the slippage.
We call it the \emph{reported price} of \textit{token~A} relative to \textit{token~B} and denote it by
\begin{equation*}
	p_{\textit{reported}}^{AB} \coloneqq  \lim_{\pooloutput{}{}{A} \to 0} \frac{\poolcontent{}{}{B}}{\poolcontent{}{}{A} - \pooloutput{}{}{A}}  = \frac{\poolcontent{}{}{B}}{\poolcontent{}{}{A}}\,\, . 
\end{equation*}
When the reported price of an AMM is different from the price in the external market without trading fees, i.e. $p_{\textit{reported}}^{AB} \neq \frac{p^A}{p^B}$, there is an arbitrage opportunity for arbitrageurs to make profits.
Therefore, due to the arbitrageurs, the reported price of the AMM is always equal to the price in the external market without trading fees~\cite{milionis2022automated}.
That is:
\begin{equation}\label{eq:reportedpriceisprice}
	\frac{p^A}{p^B} = p_{\textit{reported}}^{AB} =
	\frac{\poolcontent{}{}{B}}{\poolcontent{}{}{A}} \,\, .
\end{equation}

Since trading fees are not added to the AMM (as defined in Section~\ref{sec:model}), the product of $\poolcontent{}{}{A}$ and $\poolcontent{}{}{B}$ remains constant after each trade (Equation~\ref{leagaltrade2}).
Then, arbitrageurs keep the ratio of~$\poolcontent{}{}{A}$ and~$\poolcontent{}{}{B}$ equal to~$\frac{p^A}{p^B}$ (Equation~\ref{eq:reportedpriceisprice}).
Therefore, $\poolcontent{}{}{A}$ and $\poolcontent{}{}{B}$ remain the same after the trade and arbitrage.


			\subsection{CPMM Trading Fee}\label{sec:preliminaries:tf}

AMMs charge a trading fee for each trade operation, which the trader pays.
These trading fees form the revenue of liquidity providers.
In CPMMs, the trading fee is a constant fraction ${1 - \gamma \in [0,1]}$ of input tokens~\cite{angeris2020improved}.
This is achieved by selecting the trading function $\Phi_\textit{CPMM}(\poolcontent{}{}{A}, \poolcontent{}{}{B}, \poolinput{}{}{A}, \pooloutput{}{}{A}, \poolinput{}{}{B}, \pooloutput{}{}{B})$ as~\cite{angeris2020improved}
\begin{equation*}
	\Phi_\textit{CPMM}(\poolcontent{}{}{A}, \poolcontent{}{}{B}, \poolinput{}{}{A}, \pooloutput{}{}{A}, \poolinput{}{}{B}, \pooloutput{}{}{B}) \coloneqq 
	(\poolcontent{}{}{A} + \gamma \poolinput{}{}{A} - \pooloutput{}{}{A}) \times (\poolcontent{}{}{B}+\gamma \poolinput{}{}{B}-\pooloutput{}{}{B}) \,\, . 
\end{equation*}

Since a trade $(\poolinput{}{}{A}, \pooloutput{}{}{A}, \poolinput{}{}{B}, \pooloutput{}{}{B})$ is legal if the trading function remains constant (Equation~\ref{leagaltrade1}), 
to get $\pooloutput{}{}{A}$ \textit{token~A}, the trader pays the gross amount
\begin{align}\label{net_amount_cpmm}
	G_{\textit{CPMM}}(\poolcontent{}{}{A}, \poolcontent{}{}{B}, \pooloutput{}{}{A}) &= \frac{1}{\gamma}	\left(
	\frac{\poolcontent{}{}{A} \times \poolcontent{}{}{B}}{\poolcontent{}{}{A} - \pooloutput{}{}{A}} - \poolcontent{}{}{B}
	\right) \nonumber \\
	&= \frac{1}{\gamma}	\left(
	\frac{\poolcontent{}{}{B} \times \pooloutput{}{}{A}}{\poolcontent{}{}{A} - \pooloutput{}{}{A}}
	\right)\,\, .
\end{align}
Compared to the net amount (Equation~\ref{net_amount_pre}), the trader pays additional tokens to complete the trade; this is the \emph{trading fee}. 
In the CPMM case, it is 
\begin{equation*}
	 \frac{1-\gamma}{\gamma}	\left(
		\frac{\poolcontent{}{}{B} \times \pooloutput{}{}{A}}{\poolcontent{}{}{A} - \pooloutput{}{}{A}}
		\right)\,\, .
\end{equation*}

In the prominent Uniswap v2~\cite{adams2020uniswap}, the ratio is~${1 - \gamma = 0.003}$. 

\section{Uniswap v2 Statistics}\label{app:uniswapdata}
We analyze the five Uniswap~v2 AMMs with different pairs of tokens with the highest number of trade transactions from Ethereum block 12,000,000 to 19,500,000 (from 2021-03-08 to 2024-03-23, about 3 years).
First, we find that more than $99.5\%$ of the transactions are trade operations.
Therefore, we only focus on the throughput of trade operations.
Second, we calculate the ratio of output tokens to deposited tokens for each transaction and find that most transactions are small compared with the liquidity size.
Among all trades, the average ratio of output tokens to deposited tokens is less than $0.036\%$, and more than $99\%$ of the trades have a ratio of less than $0.52\%$.
Such a phenomenon is consistent in all token pairs.
Specifically, in the most active ones, USDC-ETH and USDT-ETH, over $99\%$ of trades exhibit a ratio of output to deposited tokens below $0.00128\%$.
Figure~\ref{fig:uniswap_v2_ratio} shows in log scale the one-complement of the cumulative distribution function of the ratio of output tokens to deposited tokens in all token pairs and in five selected pairs.

\begin{figure}
	\centering
	\includegraphics[width=0.47\textwidth]{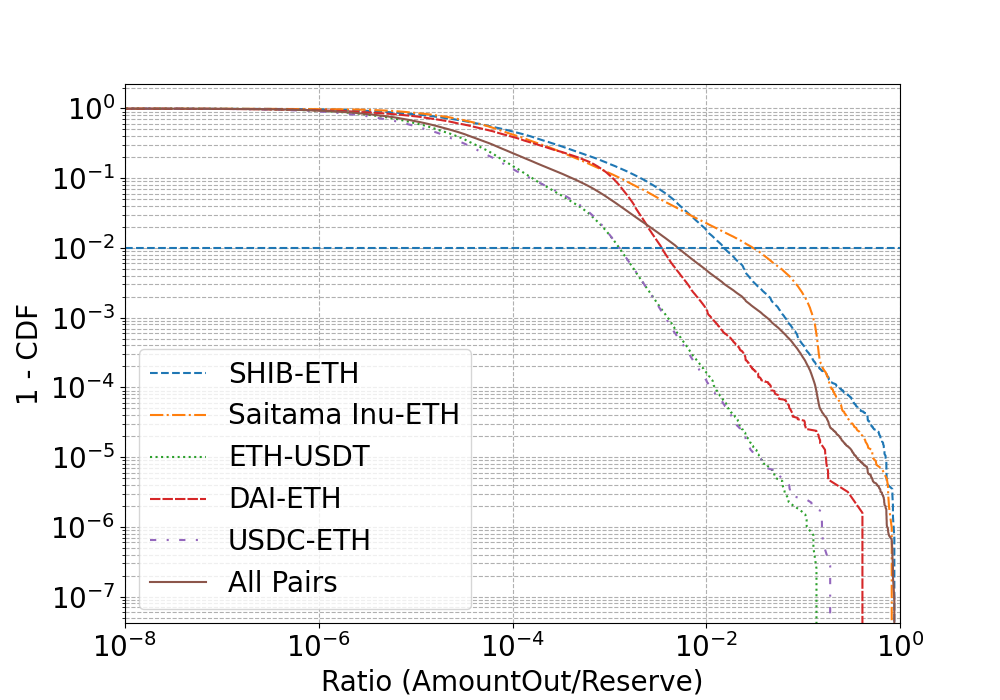}
	\caption{The ratio of output tokens to deposited tokens in Uniswap v2}
	\label{fig:uniswap_v2_ratio}
\end{figure}

\section{CPMM Does Not Satisfy Either Property}\label{app:limitation_cpmm}
Simply deploying multiple CPMM shards does not satisfy our desired properties.

\begin{restatable}{thm}{thmcpmmlimitation} \label{thm:cpmm_limitation}
	For any value of $0 < c < 1$, the CPMM cost function $G_{\textit{CPMM}}(\poolcontent{}{}{A},\poolcontent{}{}{B}, \pooloutput{}{}{A})$ does not satisfy either the $c$-non-splitting or the $c$-smaller-better properties.
\end{restatable}

\begin{proof}
	For any $0< c < 1$, it is sufficient to find a single case where each property does not hold.
	For the $c$-non-splitting property, we consider the cost of getting $\tilde{O}^A = c\poolcontent{}{}{A}$ \textit{token~A} and $\pooloutput{1}{}{A} = \pooloutput{2}{}{A} = \frac{c}{2}\poolcontent{}{}{A}$, the gross amount of getting $\tilde{O}^A = c\poolcontent{}{}{A}$ is larger than getting $\pooloutput{1}{}{A}$ and $\pooloutput{2}{}{A}$ respectively:
	\begin{align*}
		&G_{\textit{CPMM}}(\poolcontent{}{}{A},\poolcontent{}{}{B}, \tilde{O}^A) \\
		\stackrel{(\ref{net_amount_cpmm})}{=} &\frac{1}{\gamma}	\left(
		\frac{\poolcontent{}{}{B} \times c\poolcontent{}{}{A}}{\poolcontent{}{}{A} - c\poolcontent{}{}{A}}
		\right)\\
		=& \frac{1}{\gamma}	\left(
			\frac{\poolcontent{}{}{B} \times \frac{c}{2}\poolcontent{}{}{A}}{\poolcontent{}{}{A} - c\poolcontent{}{}{A}}
			\right)+ \frac{1}{\gamma}	\left(
				\frac{\poolcontent{}{}{B} \times \frac{c}{2}\poolcontent{}{}{A}}{\poolcontent{}{}{A} - c\poolcontent{}{}{A}}
				\right) \\
		>&\frac{1}{\gamma}	\left(
			\frac{\poolcontent{}{}{B} \times \frac{c}{2}\poolcontent{}{}{A}}{\poolcontent{}{}{A} - \frac{c}{2}\poolcontent{}{}{A}}
			\right)+ \frac{1}{\gamma}	\left(
				\frac{\poolcontent{}{}{B} \times \frac{c}{2}\poolcontent{}{}{A}}{\poolcontent{}{}{A} - \frac{c}{2}\poolcontent{}{}{A}}
				\right) \\
			\stackrel{(\ref{net_amount_cpmm})}{=}& G_{\textit{CPMM}}(\poolcontent{}{}{A},\poolcontent{}{}{B}, \pooloutput{1}{}{A}) + G_{\textit{CPMM}}(\poolcontent{}{}{A},\poolcontent{}{}{B}, \pooloutput{2}{}{A}) \,\, .
	\end{align*}
	Therefore, the CPMM cost function does not satisfy the $c$-non-splitting property.
	
	Next, we turn to the $c$-smaller-better property.
	Considering two shards with deposited token amounts $(\poolcontent{i}{}{A},\poolcontent{i}{}{B})$ and $(\poolcontent{j}{}{A},\poolcontent{j}{}{B})$, respectively,  where $\poolcontent{i}{}{A} < \poolcontent{j}{}{A}, \frac{\poolcontent{i}{}{A}}{\poolcontent{i}{}{B}} =  \frac{\poolcontent{j}{}{A}}{\poolcontent{j}{}{B}}$, for any output amount $0 < \pooloutput{}{}{A} \leq c\poolcontent{i}{}{A}$, consider the gross amount of getting $\pooloutput{}{}{A}$ \textit{token~A}, we have
	\begin{align*}
		G_{\textit{CPMM}}(\poolcontent{i}{}{A},\poolcontent{i}{}{B}, \pooloutput{}{}{A}) \stackrel{(\ref{net_amount_cpmm})}{=}& \frac{1}{\gamma}	\left(
		\frac{\poolcontent{i}{}{B} \times \pooloutput{}{}{A}}{\poolcontent{i}{}{A} - \pooloutput{}{}{A}}
		\right) \nonumber \\
		=& \frac{1}{\gamma}	\left(
		\frac{ \frac{\poolcontent{i}{}{B}}{\poolcontent{i}{}{A}}
			 \times \pooloutput{}{}{A}}{ 1 - \frac{\pooloutput{}{}{A}}{\poolcontent{i}{}{A}}}\right) \nonumber \\
		=&\frac{1}{\gamma}	\left(
		\frac{ \frac{\poolcontent{j}{}{B}}{\poolcontent{j}{}{A}}
				 \times \pooloutput{}{}{A}}{ 1 - \frac{\pooloutput{}{}{A}}{\poolcontent{i}{}{A}}}\right) \nonumber \\
		=&\frac{1}{\gamma}	\left(
			\frac{ \poolcontent{j}{}{B}
			\times \pooloutput{}{}{A}}{ \poolcontent{j}{}{A} - \frac{\poolcontent{j}{}{A}}{\poolcontent{i}{}{A}}\pooloutput{}{}{A}}\right) \nonumber \\	 
		>& \frac{1}{\gamma}	\left(
		\frac{\poolcontent{j}{}{B} \times \pooloutput{}{}{A}}{\poolcontent{j}{}{A} - \pooloutput{}{}{A}}
		\right) \nonumber \\
		\stackrel{(\ref{net_amount_cpmm})}{=}& G_{\textit{CPMM}}(\poolcontent{j}{}{A},\poolcontent{j}{}{B}, \pooloutput{}{}{A}) \,\, .
	\end{align*}
	Therefore, the CPMM cost function does not satisfy the $c$-smaller-better property as well.
\end{proof}

\section{Proof of Proposition~\ref{theorem:parameter1}}\label{app:proof_necessary}
\begin{restate}{Proposition~\ref{theorem:parameter1}}
	Let $\textit{tf}_{\textit{SAMM}}(\poolcontent{}{}{A}, \poolcontent{}{}{B},  \pooloutput{}{}{A}) = \textit{tf}_{\textit{BRP}}(\poolcontent{}{}{A}, \poolcontent{}{}{B},  \pooloutput{}{}{A})$, then the following conditions are necessary for 
	$c$-smaller-better
	to hold for $G_{\textit{SAMM}}(\poolcontent{}{}{A},\poolcontent{}{}{B}, \pooloutput{}{}{A})$:
	\begin{enumerate}
		\item $\beta_3 = 0$, \label{itm:c3Zero}\label{itm1}
		\item $\beta_2 + \beta_4 = 0$, \label{itm:c2c4Zero}
		\item $\beta_1 < 0$, \label{itm:c1Negative} 	
		\item $0 < \beta_4 \leq 1$, \label{itm:c4range}
		\item $r_{\min} < \beta_5 \leq r_{\max}$,\label{itm:c5range} and 
		\item $\frac{\beta_5 - r_{\min}}{-\beta_1} \geq c^{\beta_4}$.  \label{itm:51c4}\label{itmLast}
	\end{enumerate}
\end{restate}
\begin{proof}
	We would first give some common prefixes for the required items, and then obtain them one by one.
	
	The $c$-smaller-better property considers two shards~$i$ and~$j$, assuming their reported prices are identical. 
	Denote the inverse of this price by~$c^\textit{AB}$, so $\frac{\poolcontent{i}{}{A}}{\poolcontent{i}{}{B}}=\frac{\poolcontent{j}{}{A}}{\poolcontent{j}{}{B}} = \frac{1}{c^{\textit{AB}}} > 0$. The property requirement (Equation~\ref{smallerpoolsmallercost}) becomes 
	\begin{equation*}
		G_{\textit{SAMM}}(\poolcontent{i}{}{A},c^{\textit{AB}}\poolcontent{i}{}{A}, \pooloutput{}{}{A}) < G_{\textit{SAMM}}(\poolcontent{j}{}{A},c^{\textit{AB}}\poolcontent{j}{}{A}, \pooloutput{}{}{A}) \,\, .
	\end{equation*}
	The function $G_{\textit{SAMM}}(\poolcontent{}{}{A},c^{\textit{AB}}\poolcontent{}{}{A}, \pooloutput{}{}{A})$ is differentiable, and by assumption~$\poolcontent{i}{}{A} < \poolcontent{i}{}{B}$, therefore a necessary condition 
	for this inequality to hold is
	\begin{equation} \label{differentialNonNegative}
		\frac{dG_{\textit{SAMM}}(\poolcontent{}{}{A},c^{\textit{AB}}\poolcontent{}{}{A}, \pooloutput{}{}{A})}{d\poolcontent{}{}{A}} \geq 0 \,\, . 
	\end{equation}
	
	If the polynomial value is greater than the bound, $\beta_1 (\poolcontent{}{}{A}) ^{\beta_2}(\poolcontent{}{}{B})^{\beta_3} (\pooloutput{}{}{A})^{\beta_4} + \beta_5 >  r_{\max}$, then from Equations~\ref{eq:netamount},~\ref{cost_function}, and~\ref{tf_bounded_ratio}, the gross amount becomes 
	\begin{equation} \label{eq:grossAmountAboveBound}
		G_{\textit{SAMM}}(\poolcontent{}{}{A},c^{\textit{AB}}\poolcontent{}{}{A}, \pooloutput{}{}{A}) = c^{\textit{AB}}r_{\max}\pooloutput{}{}{A} + \frac{c^{\textit{AB}}\poolcontent{}{}{A} \times \pooloutput{}{}{A}}{\poolcontent{}{}{A} - \pooloutput{}{}{A}} \,\, ,
	\end{equation}
	so the derivative is 
	\begin{equation}\label{impossible_dif}
		\frac{dG_{\textit{SAMM}}(\poolcontent{}{}{A},c^{\textit{AB}}\poolcontent{}{}{A}, \pooloutput{}{}{A})}{d\poolcontent{}{}{A}} = - \frac{c^{\textit{AB}}(\pooloutput{}{}{A})^2}{(\poolcontent{}{}{A} - \pooloutput{}{}{A})^2} < 0 \,\, ,
	\end{equation}
	contradicting Equation~\ref{differentialNonNegative}.
	Therefore, for the property to hold we need $\beta_1  (\poolcontent{}{}{A}) ^{\beta_2} (\poolcontent{}{}{B})^{\beta_3}  (\pooloutput{}{}{A})^{\beta_4} + \beta_5 \leq  r_{\max}$.
	A similar situation occurs when $\beta_1 = 0$ or ${\beta_1  (\poolcontent{}{}{A}) ^{\beta_2} (\poolcontent{}{}{B})^{\beta_3}  (\pooloutput{}{}{A})^{\beta_4} + \beta_5 <  r_{\min}}$, where 
	\begin{equation*}
		G_{\textit{SAMM}}(\poolcontent{}{}{A},c^{\textit{AB}}\poolcontent{}{}{A}, \pooloutput{}{}{A}) = 
		c^{\textit{AB}} \max\{
			r_{\min} , 
			\min\{
			r_{\max}, \beta_5
			\}	
			\}	\pooloutput{}{}{A} + \frac{c^{\textit{AB}}\poolcontent{}{}{A} \times \pooloutput{}{}{A}}{\poolcontent{}{}{A} - \pooloutput{}{}{A}} \,\, ,
	\end{equation*}
	or
	\begin{equation*}
		G_{\textit{SAMM}}(\poolcontent{}{}{A},c^{\textit{AB}}\poolcontent{}{}{A}, \pooloutput{}{}{A}) = c^{\textit{AB}}r_{\min}\pooloutput{}{}{A} + \frac{c^{\textit{AB}}\poolcontent{}{}{A} \times \pooloutput{}{}{A}}{\poolcontent{}{}{A} - \pooloutput{}{}{A}} \,\, .
	\end{equation*}
	Therefore, we also need $\beta_1 \neq 0$, and ${\beta_1  (\poolcontent{}{}{A}) ^{\beta_2} (\poolcontent{}{}{B})^{\beta_3}  (\pooloutput{}{}{A})^{\beta_4} + \beta_5 \geq  r_{\min}}$, to avoid a similar contradiction with Equation~\ref{impossible_dif}. 
	Thus, we require~$\beta_1 \neq 0$ and the polynomial to be in the range 
	\begin{equation}\label{limited_poly}
	r_{\min} \leq \beta_1  (\poolcontent{}{}{A}) ^{\beta_2} (c^{\textit{AB}}\poolcontent{}{}{A})^{\beta_3}  (\pooloutput{}{}{A})^{\beta_4} + \beta_5 \leq  r_{\max} \,\, .
	\end{equation}
	
	Since the output amount~$\pooloutput{}{}{A}$ can be arbitrarily close to zero, for~$(\pooloutput{}{}{A})^{\beta_4}$ to be bounded we require 
	\begin{equation}\label{c4pos}
		\beta_4 \geq 0 \,\, . 
	\end{equation}
	Since $\pooloutput{}{}{A}$ can be arbitrarily close to zero, $\beta_1  (\poolcontent{}{}{A}) ^{\beta_2} (c^{\textit{AB}}\poolcontent{}{}{A})^{\beta_3}  (\pooloutput{}{}{A})^{\beta_4}$ can be arbitrarily close to zero, so from Equation~\ref{limited_poly} we require that
	\begin{equation}\label{c51}
		r_{\min} \leq \beta_5 \leq  r_{\max} \,\, .
	\end{equation}
	
	Here, we start to derive necessary items from the properties.
	We rewrite Equation~\ref{limited_poly} as
	\begin{equation*}
		r_{\min} \leq \left(c^{\textit{AB}}\right)^{\beta_3} \beta_1 (\poolcontent{}{}{A}) ^{\beta_2 + \beta_3+\beta_4} \left(
			\frac{\pooloutput{}{}{A}}{\poolcontent{}{}{A}}
		\right)^{\beta_4} + \beta_5 \leq  r_{\max} \,\, . 
	\end{equation*}
	The $c$-smaller-better property should hold for all reported prices, that is, this inequality should hold for all $c^{\textit{AB}}$.
	But if~$\beta_3 \neq 0$, the first element $\left(c^{\textit{AB}}\right)^{\beta_3}$ can be arbitrarily large, and since $\beta_1 (\poolcontent{}{}{A}) ^{\beta_2 + \beta_3+\beta_4} \left( \frac{\pooloutput{}{}{A}}{\poolcontent{}{}{A}} \right)^{\beta_4} \neq 0$, the whole expression $\left(c^{\textit{AB}}\right)^{\beta_3} \beta_1 (\poolcontent{}{}{A}) ^{\beta_2 + \beta_3+\beta_4} \left(
		\frac{\pooloutput{}{}{A}}{\poolcontent{}{}{A}}
	\right)^{\beta_4} + \beta_5$ is unbounded. 
	Therefore, we obtain Item~\ref{itm:c3Zero}: 
	\begin{equation}\label{c3value}
		\beta_3 = 0 \,\,.
	\end{equation}
	
	Similarly, we need to bound the expression $(\poolcontent{}{}{A}) ^{\beta_2 + \beta_3+\beta_4} \left( \frac{\pooloutput{}{}{A}}{\poolcontent{}{}{A}} \right)^{\beta_4}$. Since all positive values for~$\poolcontent{}{}{A}$ are possible and all output ratios are bounded~$0 < \frac{\pooloutput{}{}{A}}{\poolcontent{}{}{A}} < c$, to keep the expression bounded for all such values we require $\beta_2 + \beta_3 + \beta_4 = 0$, and since we already saw~$\beta_3 = 0$, we obtain Item~\ref{itm:c2c4Zero}:
	\begin{equation}\label{c2c4}
		\beta_2 + \beta_4 = 0 \,\, .
	\end{equation}
	
	Now, if~$\beta_2 = \beta_4 = 0$, the expression of Equation~\ref{limited_poly} becomes 
	$\beta_1  (\poolcontent{}{}{A}) ^{\beta_2} (c^{\textit{AB}}\poolcontent{}{}{A})^{\beta_3}  (\pooloutput{}{}{A})^{\beta_4} + \beta_5 = \beta_1 + \beta_5$, so the gross amount is 
	\begin{equation*}
		G_{\textit{SAMM}}(\poolcontent{}{}{A},c^{\textit{AB}}\poolcontent{}{}{A}, \pooloutput{}{}{A}) =  c^{\textit{AB}} \max\{
			r_{\min} , 
			\min\{
			r_{\max}, \beta_1 + \beta_5
			\}	
			\}	\pooloutput{}{}{A} + \frac{c^{\textit{AB}}\poolcontent{}{}{A} \times \pooloutput{}{}{A}}{\poolcontent{}{}{A} - \pooloutput{}{}{A}} \,\, . 
	\end{equation*}
	So, as in Equation~\ref{eq:grossAmountAboveBound}, the derivative is negative, a contradiction. 
	Therefore, we have~$\beta_4 \neq 0$, and due to Equation~\ref{c4pos} we obtain
	\begin{equation}\label{c2c4label}
		\beta_4 >0, \beta_2 < 0 \,\, .
	\end{equation}
	
	Combining the constraints we just found (Equations~\ref{c3value} and~\ref{c2c4}) into the gross amount expression (Equations~\ref{cost_function} and \ref{tf_bounded_ratio}) we have
	\begin{equation*}
		G_{\textit{SAMM}}(\poolcontent{}{}{A},c^{\textit{AB}}\poolcontent{}{}{A}, \pooloutput{}{}{A}) = 
		c^{\textit{AB}}\pooloutput{}{}{A} \times \left(\beta_1  (\poolcontent{}{}{A}) ^{-\beta_4}  (\pooloutput{}{}{A})^{\beta_4} + \beta_5\right)
		+ \frac{c^{\textit{AB}}\poolcontent{}{}{A} \times \pooloutput{}{}{A}}{\poolcontent{}{}{A} - \pooloutput{}{}{A}} 
	\end{equation*}
	and its derivative is 
	\begin{equation}\label{eq:simpliedderive}
	\frac{dG_{\textit{SAMM}}(\poolcontent{}{}{A},c^{\textit{AB}}\poolcontent{}{}{A}, \pooloutput{}{}{A})}{d\poolcontent{}{}{A}}=
		-c^{\textit{AB}}\beta_1\beta_4  \left(
			\frac{\pooloutput{}{}{A}}{\poolcontent{}{}{A}}
		\right)^{\beta_4+1} - \frac{c^{\textit{AB}}(\pooloutput{}{}{A})^2}{(\poolcontent{}{}{A} - \pooloutput{}{}{A})^2} 
	\end{equation}
	
	The second element~${\frac{c^{\textit{AB}}(\pooloutput{}{}{A})^2}{(\poolcontent{}{}{A} - \pooloutput{}{}{A})^2}}$ is positive, so to keep the derivative non-negative, the first element 
	$c^{\textit{AB}}\beta_1\beta_4  \left( \frac{\pooloutput{}{}{A}}{\poolcontent{}{}{A}} 	\right)^{\beta_4+1}$
	must be negative. 
	Since $c^{\textit{AB}} > 0$, $\beta_4 >0$, and $\left(\frac{\pooloutput{}{}{A}}{\poolcontent{}{}{A}} \right)^{\beta_4+1} > 0$, we obtain Item~\ref{itm:c1Negative}:
	\begin{equation*}
		\beta_1 < 0 \,\, . 
	\end{equation*}
	
	Since the derivative is non-negative, from Equation~\ref{eq:simpliedderive} we have
	\begin{equation*}
		-c^{\textit{AB}}\beta_1\beta_4  \left(
			\frac{\pooloutput{}{}{A}}{\poolcontent{}{}{A}}
		\right)^{\beta_4+1} \geq \frac{c^{\textit{AB}}(\pooloutput{}{}{A})^2}{(\poolcontent{}{}{A} - \pooloutput{}{}{A})^2} 
	\end{equation*}
	Since $\pooloutput{}{}{A} < c \times \poolcontent{}{}{A} < \poolcontent{}{}{A}$, we have $(\poolcontent{}{}{A} - \pooloutput{}{}{A})^2 < (\poolcontent{}{}{A})^2$.
	Therefore, we have 
	\begin{equation*}
		-c^{\textit{AB}}\beta_1\beta_4  \left(
			\frac{\pooloutput{}{}{A}}{\poolcontent{}{}{A}}
		\right)^{\beta_4+1} \geq \frac{c^{\textit{AB}}(\pooloutput{}{}{A})^2}{(\poolcontent{}{}{A})^2} \,\, .
	\end{equation*}
	By multiplying both sides by $\frac{(\poolcontent{}{}{A})^{\beta_4 + 1}}{c^{\textit{AB}}(\pooloutput{}{}{A})^{\beta_4 + 1}}$ which is positive, the above inequality is equal to:
	\begin{equation*}
		-\beta_1\beta_4 \geq \left(
			\frac{\pooloutput{}{}{A}}{\poolcontent{}{}{A}}
		\right)^{1-\beta_4} \,\, .
	\end{equation*}
	Since $\frac{\pooloutput{}{}{A}}{\poolcontent{}{}{A}}$ can be arbitrarily close to zero, if $\beta_4 > 1$, the right side of the above inequality can be arbitrarily large, so we require $\beta_4 <= 1$.
	Combining equation~\ref{c2c4label}, we obtain Item~\ref{itm:c4range}:
	\begin{equation*}
		0 < \beta_4 \leq 1 \,\, .
	\end{equation*}

	From Equation~\ref{limited_poly}, we have 
	\begin{equation*}
		\beta_1  (\poolcontent{}{}{A}) ^{\beta_2} (\poolcontent{}{}{B})^{\beta_3}  (\pooloutput{}{}{A})^{\beta_4} + \beta_5  \geq r_{\min}.
	\end{equation*}
	And since $\beta_1 < 0$, we have $\beta_1  (\poolcontent{}{}{A}) ^{\beta_2} (\poolcontent{}{}{B})^{\beta_3}  (\pooloutput{}{}{A})^{\beta_4} < 0$ so $\beta_5 > r_{\min}$. This allows us to make the first inequality of Equation~\ref{c51} strict, which gives us Item~\ref{itm:c5range}: 
	\begin{equation*}
		r_{\min} < \beta_5 \leq r_{\max} \,\, .
	\end{equation*}
	
	From Equations~\ref{limited_poly} and \ref{c2c4}, we have
	\begin{equation*}
		r_{\min} \leq \beta_1 \left(
			\frac{\pooloutput{}{}{A}}{\poolcontent{}{}{A}}
		\right)^{\beta_4} + \beta_5 \leq r_{\max} \,\, .
	\end{equation*}
	
	Since $\beta_1 <0$ and~$\beta_5 \leq r_{\max}$, the right side of the above inequality always holds.
	From the left side, we have
	\begin{equation*}
		\frac{\beta_5 - r_{\min}}{-\beta_1} \geq \left(
			\frac{\pooloutput{}{}{A}}{\poolcontent{}{}{A}}
		\right)^{\beta_4} \,\, .
	\end{equation*}
	Since the above equation holds for all $\frac{\pooloutput{}{}{A}}{\poolcontent{}{}{A}} \in (0,c)$, we obtain Item~\ref{itm:51c4}:
	\begin{equation*}
		\frac{\beta_5 - r_{\min}}{-\beta_1} \geq c^{\beta_4} \,\, . 
	\end{equation*}
	
	We have now shown all constraints~\ref{itm1}--\ref{itmLast} hold. 
	\end{proof}


\section{Proof of Theorem~\ref{theorem:parameter2}}\label{app:proof_sufficient}

\begin{restate}{Theorem~\ref{theorem:parameter2}}
	Let $\textit{tf}_{\textit{SAMM}}(\poolcontent{}{}{A}, \poolcontent{}{}{B},  \pooloutput{}{}{A}) = \textit{tf}_{\textit{BRP}}(\poolcontent{}{}{A}, \poolcontent{}{}{B},  \pooloutput{}{}{A})$, if $\beta_1 < 0, \beta_2 + \beta_4 = 0, \beta_3 =0, 0 <\beta_4 \leq 1,  r_{\min} < \beta_5 \leq r_{\max}$ and $\frac{\beta_5 - r_{\min}}{-\beta_1} \geq c^{\beta_4}$, then following items are sufficient for the $c$-non-splitting and $c$-smaller-better properties to hold for $G_{\textit{SAMM}}(\poolcontent{}{}{A},\poolcontent{}{}{B}, \pooloutput{}{}{A})$:
	\begin{enumerate}
		\item $\beta_1\beta_4(\beta_4+1)c^{\beta_4-1}(1-c)^3\leq -2$\label{itm:beta1beta4}
		\item $-\beta_1\beta_4 \geq\frac{c^{1-\beta_4}}{(1-c)^2}$\label{itm:beta1beta4c}
	\end{enumerate}
\end{restate}

\begin{proof}
	Initially, we expand and simplify the form of the gross amount function according to our assumptions.
	Then, we prove that Item~\ref{itm:beta1beta4} is sufficient for the $c$-non-splitting property to hold.
	Finally, we prove that Item~\ref{itm:beta1beta4c} is sufficient for the $c$-smaller-better property to hold.

	Since $\beta_1< 0, r_{\min} < \beta_5$, we have $\beta_1  (\poolcontent{}{}{A}) ^{\beta_2} (c^{\textit{AB}}\poolcontent{}{}{A})^{\beta_3}  (\pooloutput{}{}{A})^{\beta_4} + \beta_5 \leq  r_{\max}$.
	Since $\frac{\beta_5 - r_{\min}}{-\beta_1} \geq c^{\beta_4}$, we have $r_{\min} \leq \beta_5 + \beta_1c^{\beta_4}$.

	Since $\beta_2 = - \beta_4, \beta_3 =0, \frac{\pooloutput{}{}{A}}{\poolcontent{}{}{A}}\leq c$, we have

	\begin{align*}
		\beta_1  (\poolcontent{}{}{A}) ^{\beta_2} (c^{\textit{AB}}\poolcontent{}{}{A})^{\beta_3}  (\pooloutput{}{}{A})^{\beta_4} + \beta_5 =& \beta_1\times \left(\frac{\pooloutput{}{}{A}}{\poolcontent{}{}{A}} \right)^{\beta_4} + \beta_5\\
		&\geq \beta_5 + \beta_1c^{\beta_4}\\
		&\geq r_{\min}
	\end{align*}
	Then we can expand the trading fee function and the gross amount function to:
	\begin{equation*}
		\textit{tf}_{\textit{SAMM}}(\poolcontent{}{}{A}, \poolcontent{}{}{B},  \pooloutput{}{}{A}) = \frac{\poolcontent{}{}{B}}{\poolcontent{}{}{A}} \times \pooloutput{}{}{A} \times \left(
		\beta_1\times \left(\frac{\pooloutput{}{}{A}}{\poolcontent{}{}{A}} \right)^{\beta_4} + \beta_5
		\right)
	\end{equation*}
	and
	\begin{multline}\label{proof_sufficient_cost}
		G_{\textit{SAMM}}(\poolcontent{}{}{A},\poolcontent{}{}{B}, \pooloutput{}{}{A}) = \\
		\frac{\poolcontent{}{}{B}}{\poolcontent{}{}{A}} \times \pooloutput{}{}{A} \times \left(
			\beta_1\times \left(\frac{\pooloutput{}{}{A}}{\poolcontent{}{}{A}} \right)^{\beta_4} + \beta_5
			\right)
		+ \frac{\poolcontent{}{}{B} \times \pooloutput{}{}{A}}{\poolcontent{}{}{A} - \pooloutput{}{}{A}} \,\, .
	\end{multline}

	Now we start to prove the $c$-non-splitting property holds.
	We first show that if the $c$-non-splitting property holds for $m=2$, then it holds for any $m>2$.
	Then, we prove that it holds for $m=2$ when Item~\ref{itm:beta1beta4} holds.
	
	Consider the case of $m=2$, where for $\pooloutput{1}{}{A}, \pooloutput{2}{}{A}> 0$, the gross amount of acquiring $\pooloutput{1}{}{A}+ \pooloutput{2}{}{A}$ is less than the sum of the gross amounts of acquiring $\pooloutput{1}{}{A}$ and $\pooloutput{2}{}{A}$:
	\begin{equation}\label{eq:2non_splitting}
		G_{\textit{SAMM}}(\poolcontent{}{}{A},\poolcontent{}{}{B}, \pooloutput{1}{}{A}+ \pooloutput{2}{}{A}) <
		  G_{\textit{SAMM}}(\poolcontent{}{}{A},\poolcontent{}{}{B}, \pooloutput{1}{}{A}) + G_{\textit{SAMM}}(\poolcontent{}{}{A},\poolcontent{}{}{B}, \pooloutput{2}{}{A}) \,\, .
	\end{equation}
	For $m=3$, we can get that the total gross amount of acquiring $\pooloutput{1}{}{A}$, $\pooloutput{2}{}{A}$ and $\pooloutput{3}{}{A}$ separately is less than the gross amounts of acquiring $\pooloutput{1}{}{A} + \pooloutput{2}{}{A} + \pooloutput{3}{}{A}$ in one time:
	\begin{align*}
		&\sum_{j=1}^3 G_{\textit{SAMM}}(\poolcontent{}{}{A},\poolcontent{}{}{B}, \pooloutput{j}{}{A})\nonumber \\ = &\left(
			G_{\textit{SAMM}}(\poolcontent{}{}{A},\poolcontent{}{}{B}, \pooloutput{1}{}{A}) + G_{\textit{SAMM}}(\poolcontent{}{}{A},\poolcontent{}{}{B}, \pooloutput{2}{}{A})
		\right) \nonumber \\
		& + G_{\textit{SAMM}}(\poolcontent{}{}{A},\poolcontent{}{}{B}, \pooloutput{3}{}{A})\nonumber\\
		\stackrel{(\ref{eq:2non_splitting})}{>} & G_{\textit{SAMM}}(\poolcontent{}{}{A},\poolcontent{}{}{B}, \pooloutput{1}{}{A} + \pooloutput{2}{}{A}) + G_{\textit{SAMM}}(\poolcontent{}{}{A},\poolcontent{}{}{B}, \pooloutput{3}{}{A})\nonumber\\
		\stackrel{(\ref{eq:2non_splitting})}{>}& G_{\textit{SAMM}}(\poolcontent{}{}{A},\poolcontent{}{}{B}, \pooloutput{1}{}{A} + \pooloutput{2}{}{A} + \pooloutput{3}{}{A}) \,\, .
	\end{align*}

	This can be easily generalized to any $m > 2$.
	Therefore, we only need to prove the $c$-non-splitting property for $m=2$, which is shown in Equation~\ref{eq:2non_splitting}.

	Since $G_{SAMM}(\poolcontent{}{}{A},\poolcontent{}{}{B}, 0) = 0$, Equation~\ref{eq:2non_splitting} is equivalent to
	\begin{multline}
		G_{SAMM}(\poolcontent{}{}{A},\poolcontent{}{}{B}, 0) + G_{\textit{SAMM}}(\poolcontent{}{}{A},\poolcontent{}{}{B}, \pooloutput{1}{}{A}+ \pooloutput{2}{}{A}) < \\  G_{\textit{SAMM}}(\poolcontent{}{}{A},\poolcontent{}{}{B}, \pooloutput{1}{}{A}) + G_{\textit{SAMM}}(\poolcontent{}{}{A},\poolcontent{}{}{B}, \pooloutput{2}{}{A}) \nonumber \,\, .
	\end{multline}
	The above inequality holds when $G_{\textit{SAMM}}(\poolcontent{}{}{A},\poolcontent{}{}{B}, \pooloutput{}{}{A})$ is strictly concave over $\pooloutput{}{}{A}$.
	A sufficient condition for strict concavity is the second derivative of $G_{\textit{SAMM}}(\poolcontent{}{}{A},\poolcontent{}{}{B}, \pooloutput{}{}{A})$ to be negative for all $0 < \pooloutput{}{}{A} < c\times \poolcontent{}{}{A}$:
	\begin{equation*}
		\frac{d^2G_{\textit{SAMM}}(\poolcontent{}{}{A},\poolcontent{}{}{B}, \pooloutput{}{}{A})}{d\left(\pooloutput{}{}{A}\right)^2} < 0 \,\, .
	\end{equation*}
	From Equation~\ref{proof_sufficient_cost}, the second derivative of the gross amount function is
	\begin{align*}
		&\frac{d^2G_{\textit{SAMM}}(\poolcontent{}{}{A},\poolcontent{}{}{B}, \pooloutput{}{}{A})}{d\left(\pooloutput{}{}{A}\right)^2} \\
		=& \beta_1\beta_4(\beta_4+1)\left(
			\poolcontent{}{}{A}
		\right)^{-\beta_4-1}\poolcontent{}{}{B}\left(\pooloutput{}{}{A}\right)^{\beta_4-1} + \frac{2\poolcontent{}{}{A}\poolcontent{}{}{B}}{\left(\poolcontent{}{}{A} - \pooloutput{}{}{A}\right)^3} \\
		< & \beta_1\beta_4(\beta_4+1)\left(
			\poolcontent{}{}{A}
		\right)^{-\beta_4-1}\poolcontent{}{}{B}\left(\pooloutput{}{}{A}\right)^{\beta_4-1}  + \frac{2\poolcontent{}{}{A}\poolcontent{}{}{B}}{\left(\poolcontent{}{}{A} - c\poolcontent{}{}{A}\right)^3}\\
		=& \beta_1\beta_4(\beta_4+1)\left(
			\poolcontent{}{}{A}
		\right)^{-\beta_4-1}\poolcontent{}{}{B}\left(\pooloutput{}{}{A}\right)^{\beta_4-1} + \frac{2\poolcontent{}{}{B}}{(1-c)^3\left(\poolcontent{}{}{A}\right)^2}.
	\end{align*} 
	Therefore, a sufficient condition to make the second derivative negative is
	\begin{equation*}
		\beta_1\beta_4(\beta_4+1)\left(
			\poolcontent{}{}{A}
		\right)^{-\beta_4-1}\poolcontent{}{}{B}\left(\pooloutput{}{}{A}\right)^{\beta_4-1} + \frac{2\poolcontent{}{}{B}}{(1-c)^3\left(\poolcontent{}{}{A}\right)^2} \leq 0 .
	\end{equation*}
	The above inequality is equal to 
	\begin{equation}\label{eq:sufficient_form}
		\beta_1\beta_4(\beta_4+1)(1-c)^3\left(
			\frac{\pooloutput{}{}{A}}{\poolcontent{}{}{A}}	
		\right)^{\beta_4-1}\leq -2 \,\, .
	\end{equation}
	The above condition is sufficient for the $c$-non-splitting property.

	Since $\beta_4 \leq 1$ and $\frac{\pooloutput{}{}{A}}{\poolcontent{}{}{A}} < c$, we have 
	\begin{equation*}
		\beta_1\beta_4(\beta_4+1)(1-c)^3\left(
			\frac{\pooloutput{}{}{A}}{\poolcontent{}{}{A}}	
		\right)^{\beta_4-1} \leq \beta_1\beta_4(\beta_4+1)c^{\beta_4-1}(1-c)^3 \,\, .
	\end{equation*}

	Therefore, Item~\ref{itm:beta1beta4} is sufficient for Equation~\ref{eq:sufficient_form}, which indicates that the $c$-non-splitting property holds under Item~\ref{itm:beta1beta4}.

	Now we turn to the $c$-smaller-better property.
	The $c$-smaller-better property considers two shards~$i$ and~$j$, assuming their reported prices are identical. 
	Denote the inverse of this price by~$c^\textit{AB}$, so $\frac{\poolcontent{i}{}{A}}{\poolcontent{i}{}{B}}=\frac{\poolcontent{j}{}{A}}{\poolcontent{j}{}{B}} = \frac{1}{c^{\textit{AB}}} > 0$.
	A sufficient condition for the $c$-smaller-better property is the derivative of the gross amount function over $\poolcontent{}{}{A}$ to be positive for all $0 < \pooloutput{}{}{A} < c\times \poolcontent{}{}{A}$:
	\begin{equation}\label{differential3}
		\frac{dG_{\textit{SAMM}}(\poolcontent{}{}{A},c^{\textit{AB}}\poolcontent{}{}{A}, \pooloutput{}{}{A})}{d\poolcontent{}{}{A}} > 0 \,\,
	\end{equation}
	From Equation~\ref{proof_sufficient_cost}, we have the derivative:
	\begin{multline*}
		\frac{dG_{\textit{SAMM}}(\poolcontent{}{}{A},c^{\textit{AB}}\poolcontent{}{}{A}, \pooloutput{}{}{A})}{d\poolcontent{}{}{A}} = \\
		-c^{\textit{AB}}\beta_1\beta_4  \left(
			\frac{\pooloutput{}{}{A}}{\poolcontent{}{}{A}}
		\right)^{\beta_4+1} - \frac{c^{\textit{AB}}(\pooloutput{}{}{A})^2}{(\poolcontent{}{}{A} - \pooloutput{}{}{A})^2} \,\, .
	\end{multline*}
	Since $\pooloutput{}{}{A} < c \times \poolcontent{}{}{A}$, we have a lower bound of the derivative:
	\begin{multline*}
		\frac{dG_{\textit{SAMM}}(\poolcontent{}{}{A},c^{\textit{AB}}\poolcontent{}{}{A}, \pooloutput{}{}{A})}{d\poolcontent{}{}{A}} > \\
		-c^{\textit{AB}}\beta_1\beta_4  \left(
			\frac{\pooloutput{}{}{A}}{\poolcontent{}{}{A}}
		\right)^{\beta_4+1} - \frac{c^{\textit{AB}}(\pooloutput{}{}{A})^2}{(\poolcontent{}{}{A} - c\poolcontent{}{}{A})^2} \,\, .
	\end{multline*}
	Therefore, it is a sufficient condition for Equation~\ref{differential3} to hold if the lower bound is non-navigate:
	\begin{equation*}
		-c^{\textit{AB}}\beta_1\beta_4  \left(
			\frac{\pooloutput{}{}{A}}{\poolcontent{}{}{A}}
		\right)^{\beta_4+1} - \frac{c^{\textit{AB}}(\pooloutput{}{}{A})^2}{(\poolcontent{}{}{A} - c\poolcontent{}{}{A})^2} \geq 0 \,\, .
	\end{equation*}
	The above equation is equivalent to
	\begin{equation}\label{eq:sufficient_form2}
		-\beta_1\beta_4 \geq \frac{1}{(1-c)^2} \left(
			\frac{\pooloutput{}{}{A}}{\poolcontent{}{}{A}}
		\right)^{1-\beta_4} \,\, .
	\end{equation}
	Since $ 0 < \frac{\pooloutput{}{}{A}}{\poolcontent{}{}{A}} < c$ and $\beta_4 \leq 1$, we have
	\begin{equation*}
		\frac{c^{1-\beta_4}}{(1-c)^2}
		\geq 
		\frac{1}{(1-c)^2} \left(
			\frac{\pooloutput{}{}{A}}{\poolcontent{}{}{A}}
		\right)^{1-\beta_4}
		 \,\, .
	\end{equation*}

	Therefore, Item~\ref{itm:beta1beta4c} is sufficient for Equation~\ref{eq:sufficient_form2} to hold, which indicates the $c$-smaller-better property holds under Item~\ref{itm:beta1beta4c}.

	In summary, we have shown that Item~\ref{itm:beta1beta4} is sufficient for the $c$-non-splitting property to hold, and Item~\ref{itm:beta1beta4c} is sufficient for the $c$-smaller-better property to hold.

\end{proof}


\section{CPMM Equilibrium}\label{app:proof_uniswap_splitting}

In CPMM, the gross amount is linear to the net amount.
Since traders can suffer less from slippage by splitting a trade, the gross amount of splitting a trade is less than the gross amount of trading the same amount at one time.
\begin{restatable}{thm}{theoremuniswapsplitting}
	\label{theorem_uniswap_splitting}
	In $\Gamma_n(\textit{tf}_\textit{CPMM})$, the following strategy $\tau^{BA}(\poolcontentvector{}{}{}, \uac{}{}{\textit{BA}})$ is the only dominant strategy for the \textit{BA} trader, where
	\begin{equation*}
		\tau^{BA}(\poolcontentvector{}{}{}, \uac{}{}{\textit{BA}}) = \begin{cases}
			1, & \text{if } a^{\textit{BA}} = \left(\frac{\poolcontent{1}{}{A}}{\sum \poolcontent{i}{}{A}}\uac{}{}{\textit{BA}}, \cdots, \frac{\poolcontent{n}{}{A}}{\sum \poolcontent{i}{}{A}}\uac{}{}{\textit{BA}} \right)  \\
			0, & \text{otherwise}
		  \end{cases}
	\end{equation*}	
\end{restatable}

\begin{proof}
	We start by calculating the best response for the \textit{BA} trader.
	Since the utility $U^{BA}(\poolcontentvector{}{}{},  \uac{}{}{\textit{BA}}, \pi^{BA})$ is a linear combination of the revenue over actions $U^{BA}(\poolcontentvector{}{}{},  \uac{}{}{\textit{BA}}, a^{\textit{BA}})$ (Equation~\ref{utilitybapi}), we can first calculate the optimal action for the \textit{BA} trader.
	Combining the CPMM gross amount (Equation~\ref{net_amount_cpmm}) and the revenue function (Equation~\ref{utilityba}), we obtain
	\begin{align*}
		U^{\textit{BA}}(	\poolcontentvector{}{}{}, \ a^{\textit{BA}}) &= -\sum_i\frac{1}{\gamma}\frac{\frac{p^A}{p^B}\poolcontent{i}{}{A} \uac{i}{}{\textit{BA}}}{\poolcontent{i}{}{A} - \uac{i}{}{\textit{BA}}} \nonumber \\
		&= - \frac{p^A}{\gamma p_B}\sum_i\frac{\poolcontent{i}{}{A} \uac{i}{}{\textit{BA}}}{\poolcontent{i}{}{A} - \uac{i}{}{\textit{BA}}} \,\, ,
	\end{align*}
	where the sum of $\uac{i}{}{\textit{BA}}$ is the total required amount of \textit{token~A}	(Equation~\ref{restriction_ba}).
	
	We first consider the case of two shards and then extend it to the general case.
	We define the function $f(\cdot, \cdot)$ which is proportional to the utility of the trader in $\textit{shard}_i$ and $\textit{shard}_j$:
	\begin{equation*}
		f(\uac{i}{}{\textit{BA}}, \uac{j}{}{\textit{BA}}) \coloneqq
		 - \left(\frac{\poolcontent{i}{}{A} \uac{i}{}{\textit{BA}}}{\poolcontent{i}{}{A} - \uac{i}{}{\textit{BA}}} + \frac{\poolcontent{j}{}{A} \uac{j}{}{\textit{BA}}}{\poolcontent{j}{}{A} - \uac{j}{}{\textit{BA}}}\right) \,\, . 
	\end{equation*}
	Denote by $z \coloneqq \uac{i}{}{\textit{BA}} + \uac{j}{}{\textit{BA}}\leq \uac{}{}{\textit{BA}}$.
	Then, we have
	\begin{equation*}
		f(\uac{i}{}{\textit{BA}}, z - \uac{i}{}{\textit{BA}}) = - \left(\frac{\poolcontent{i}{}{A} \uac{i}{}{\textit{BA}}}{\poolcontent{i}{}{A} - \uac{i}{}{\textit{BA}}} + \frac{\poolcontent{j}{}{A} (z - \uac{i}{}{\textit{BA}})}{\poolcontent{j}{}{A} - (z - \uac{i}{}{\textit{BA}})}\right)\nonumber \,\, .
	\end{equation*}
	The derivative of the above function is
	\begin{equation*}
		\frac{df(\uac{i}{}{\textit{BA}}, z - \uac{i}{}{\textit{BA}})}{d\uac{i}{}{\textit{BA}}} = \frac{(\poolcontent{i}{}{A})^2}{(\poolcontent{i}{}{A} - \uac{i}{}{\textit{BA}})^2} - \frac{(\poolcontent{j}{}{A})^2}{(\poolcontent{j}{}{A} - (z - \uac{i}{}{\textit{BA}}))^2}\nonumber  \,\, .
	\end{equation*}
	If $0 \leq \uac{i}{}{\textit{BA}} < \frac{\poolcontent{i}{}{A}}{\poolcontent{i}{}{A} + \poolcontent{j}{}{A}} z$, then $\frac{df(\uac{i}{}{\textit{BA}}, z - \uac{i}{}{\textit{BA}})}{d\uac{i}{}{\textit{BA}}} > 0$; if ${\frac{\poolcontent{i}{}{A}}{\poolcontent{i}{}{A} + \poolcontent{j}{}{A}} z < \uac{i}{}{\textit{BA}} \leq z}$, then $\frac{df(\uac{i}{}{\textit{BA}}, z - \uac{i}{}{\textit{BA}})}{d\uac{i}{}{\textit{BA}}} < 0$.
	Therefore, $f(\uac{i}{}{\textit{BA}}, z - \uac{i}{}{\textit{BA}})$ is maximized only when $\uac{i}{}{\textit{BA}} = \frac{\poolcontent{i}{}{A}}{\poolcontent{i}{}{A} + \poolcontent{j}{}{A}} z = \frac{\poolcontent{i}{}{A}}{\poolcontent{i}{}{A} + \poolcontent{j}{}{A}}(\uac{i}{}{\textit{BA}} + \uac{j}{}{\textit{BA}})$.
	
	The revenue $U^{\textit{BA}}(	\poolcontentvector{}{}{}, \uac{}{}{\textit{BA}}, a^{\textit{BA}})$ reaches the maximum only when $\forall i, j , \frac{\uac{i}{}{\textit{BA}}}{\uac{j}{}{\textit{BA}}} = \frac{\poolcontent{i}{}{A}}{\poolcontent{j}{}{A}}$, or the trader can replace $\uac{i}{}{\textit{BA}}$ and $\uac{j}{}{\textit{BA}}$ with $\frac{\poolcontent{i}{}{A}}{\poolcontent{i}{}{A} + \poolcontent{j}{}{A}}(\uac{i}{}{\textit{BA}} + \uac{j}{}{\textit{BA}})$ and $\frac{\poolcontent{j}{}{A}}{\poolcontent{i}{}{A} + \poolcontent{j}{}{A}}(\uac{i}{}{\textit{BA}} + \uac{j}{}{\textit{BA}})$ to get higher utility.
	
	Therefore, the only optimal action of the \textit{BA} trader is
	\begin{equation*}
		a^{\textit{BA}} = \left(\frac{\poolcontent{1}{}{A}}{\sum \poolcontent{i}{}{A}}, \cdots, \frac{\poolcontent{n}{}{A}}{\sum \poolcontent{i}{}{A}} \right) \,\, .
	\end{equation*}
	
	Then, the only dominant strategy of the \textit{BA} trader is
	\begin{equation*}
		\tau^{BA}(\poolcontentvector{}{}{}, \uac{}{}{\textit{BA}}) = \begin{cases}
			1, & \text{if } a^{\textit{BA}} = \left(\frac{\poolcontent{1}{}{A}}{\sum \poolcontent{i}{}{A}}, \cdots, \frac{\poolcontent{n}{}{A}}{\sum \poolcontent{i}{}{A}} \right)  \\
			0, & \text{otherwise}
		  \end{cases}
	\end{equation*}
	\end{proof}



\section{SAMM Equilibrium} \label{app:gtAnalysis}

We analyze the behavior of players in the game with SAMM.
We first prove that the trader randomly selects a shard to trade when the states of shards are balanced (\S\ref{gt:trader_strategy}).
Then, we show the system tends to the balanced state since liquidity providers invest their tokens in the smallest shards, reducing the difference in the volume of shards (\S\ref{gt:lp_strategy}).

		\subsection{Trader Strategy}\label{gt:trader_strategy}

Consider the case that the system state is $\poolcontentvector{}{}{}= \left(
	\left(\poolcontent{1}{}{A},\poolcontent{1}{}{B}, \poolcontent{1}{}{S} \right),
	\cdots \left(\poolcontent{n}{}{A},\poolcontent{n}{}{B}, 
	\poolcontent{n}{}{S}	 \right) \right)$.
As discussed in Section~\ref{sec:SAMM_parameter}, the SAMM gross amount satisfies the $c$-non-splitting property and $c$-smaller-better property for a certain $0 < c < 1$.
We assume that the required amount of \textit{token~A}, $\uac{}{}{\textit{BA}}$, is at most a fraction~$c$ of the amount of deposited \textit{token~A} in all shards, i.e., 
\begin{equation*}
	\forall 1 \leq i \leq n, \uac{}{}{\textit{BA}} \leq c\poolcontent{i}{}{A} \,\,.
\end{equation*}

\subsubsection{Traders' optimal action}
The $c$-non-splitting property and $c$-smaller-better property give a trader the incentive to randomly select one of the smallest shards to trade all her required tokens.
Recall that $a_i^{\text{BA}}(\uac{}{}{\textit{BA}})$ is the action of acquiring all $\uac{}{}{\textit{BA}}$ \textit{token~A} in $\textit{shard}_i$ (Equation~\ref{eq:single_action}).
We define the set of actions that trade in one of the smallest shards:
\begin{restate}{Definition~\ref{def:smallest_action_set}}
	\defsmallestactionset
\end{restate}

The cardinality of $\SmallestAction(\uac{}{}{\textit{BA}}, \poolcontentvector{}{}{})$ is the number of smallest shards in $\poolcontentvector{}{}{}$.
We denote this by 

\begin{equation*}
	\SmallestNum(\poolcontentvector{}{}{}) = \left|\SmallestAction(\uac{}{}{\textit{BA}}, \poolcontentvector{}{}{}) \right| \,\, .
\end{equation*}

When the trader selects one of the actions in the smallest shard action set $\SmallestAction(\uac{}{}{\textit{BA}}, \poolcontentvector{}{}{})$, she gets the highest revenue:

\newcommand{\contentlemmaoptimaltraderaction}{In $\Gamma_n(tf_{\textit{SAMM}})$, a trader wants to get $\uac{}{}{\textit{BA}}$ \textit{token~A} when the system state is $\poolcontentvector{}{}{}$.
Then for the action which obtains all $\uac{}{}{\textit{BA}}$ \textit{token~A} in one of the smallest shards with index $i^{\ast}$, where $a_{i^{\ast}}^{\text{BA}}(\uac{}{}{\textit{BA}}) \in \SmallestAction(\uac{}{}{\textit{BA}}, \poolcontentvector{}{}{})$, the trader has no less than the revenue of any other actions:
\begin{equation*}
	\forall a^{\textit{BA}} \in \ActionSpaceBA(\uac{}{}{\textit{BA}}), U^{\textit{BA}}(a_{i^{\ast}}^{\text{BA}}, \poolcontentvector{}{}{}) \geq U^{\textit{BA}}(a^{\textit{BA}}, \poolcontentvector{}{}{}) \,\, .
\end{equation*}}

\begin{restatable}{lem}{lemmaoptimaltraderaction}\label{lemma_optimal_trader_action}
	\contentlemmaoptimaltraderaction
\end{restatable}

\begin{proof}[Proof Sketch]
	Due to the $c$-non-splitting property, trading in a single shard is better than trading in multiple shards.
	Then the revenue of trading in one of the smallest shards is no less than that in any other shard due to the $c$-smaller-better property.
\end{proof}

\begin{proof}
	Consider an arbitrary action acquiring $\uac{}{}{\textit{BA}}$ \textit{token~B}, $a^{\textit{BA}} = \left(\uac{1}{}{\textit{BA}}, \cdots, \uac{n}{}{\textit{BA}} \right)\in \ActionSpaceBA(\uac{}{}{\textit{BA}})$.
	Given the state of the shard ${\poolcontentvector{}{}{} = \left(
		\left(\poolcontent{1}{}{A},\poolcontent{1}{}{B}, \poolcontent{1}{}{S} \right),
		\cdots \left(\poolcontent{n}{}{A},\poolcontent{n}{}{B}, \poolcontent{n}{}{S}	 \right) \right)}$, the utility of the trader following action $a^{\textit{BA}}$ is
	\begin{equation}\label{lemma_optimal_trader_action_proof1}
		U^{\textit{BA}}(\poolcontentvector{}{}{}, a^{\textit{BA}}) 
		= -p^B \times \sum_{j=1}^{n}G_{\textit{SAMM}}(\poolcontent{j}{}{A},\frac{p^A}{p_B}\poolcontent{j}{}{A}, \uac{j}{}{\textit{BA}}) \,\, .
	\end{equation}
	Since $a_{i^{\ast}}^{\text{BA}}(\uac{}{}{\textit{BA}}) \in \SmallestAction(\uac{}{}{\textit{BA}}, \poolcontentvector{}{}{})$, we have $\poolcontent{i^{\ast}}{}{A} \leq \poolcontent{j}{}{A}$.
	From the $c$-smaller-better property, for all $j$, the gross amount of getting $\uac{j}{}{\textit{BA}}$ in $\textit{shard}_{i^{\ast}}$ is no larger than that in $\textit{shard}_j$:
	\begin{equation*}
		G_{\textit{SAMM}}(\poolcontent{i^{\ast}}{}{A},\frac{p^A}{p^B}\poolcontent{i^{\ast}}{}{A}, \uac{j}{}{\textit{BA}})
		\leq G_{\textit{SAMM}}(\poolcontent{j}{}{A},\frac{p^A}{p^B}\poolcontent{j}{}{A}, \uac{j}{}{\textit{BA}}) \,\, .
	\end{equation*}
	Summing over all $j$, we have
	\begin{equation}\label{lemma_optimal_trader_action_proof2}
		\sum_{j=1}^{n}G_{\textit{SAMM}}(\poolcontent{i^{\ast}}{}{A},\frac{p^A}{p^B}\poolcontent{i^{\ast}}{}{A}, \uac{j}{}{\textit{BA}})
		\leq \sum_{j=1}^{n}G_{\textit{SAMM}}(\poolcontent{j}{}{A},\frac{p^A}{p^B}\poolcontent{j}{}{A}, \uac{j}{}{\textit{BA}}) \,\, .
	\end{equation}
	From $c$-non-splitting property,	since $\sum_{j=1}^{n} \uac{j}{}{\textit{BA}} = \uac{}{}{\textit{BA}}$, the gross amount of trading $\uac{}{}{\textit{BA}}$ in a shard is no larger than the sum of gross amount of trading $\uac{j}{}{\textit{BA}}$ in the same shard:
	\begin{equation*}
		G_{\textit{SAMM}}(\poolcontent{i^{\ast}}{}{A},\frac{p^A}{p^B}\poolcontent{i^{\ast}}{}{A}, \uac{}{}{\textit{BA}})\leq 
		\sum_{j=1}^{n}G_{\textit{SAMM}}(\poolcontent{i^{\ast}}{}{A},\frac{p^A}{p^B}\poolcontent{i^{\ast}}{}{A}, \uac{j}{}{\textit{BA}})
	\end{equation*}
	Combining with Equation~\ref{lemma_optimal_trader_action_proof2}, we obtain 
	\begin{equation*}
		G_{\textit{SAMM}}(\poolcontent{i^{\ast}}{}{A},\frac{p^A}{p^B}\poolcontent{i^{\ast}}{}{A}, \uac{}{}{\textit{BA}})
		\leq \sum_{j=1}^{n}G_{\textit{SAMM}}(\poolcontent{j}{}{A},\frac{p^A}{p^B}\poolcontent{j}{}{A}, \uac{j}{}{\textit{BA}}) \,\, .
	\end{equation*} 
	We thus conclude that the revenue of action ${a_{i^{\ast}}^{\text{BA}}(\uac{}{}{\textit{BA}}) = \left(0, \cdots, \uac{i^{\ast}}{}{\textit{BA}} = \uac{}{}{\textit{BA}}, 0, \cdots, 0 \right)}$ is maximal:
	\begin{align*}
		U^{\textit{BA}}( \poolcontentvector{}{}{}, a_{i^{\ast}}^{\text{BA}}(\uac{}{}{\textit{BA}})) =& -p^B \times G(\poolcontent{i^{\ast}}{}{A},\frac{p^A}{p^B}\poolcontent{i^{\ast}}{}{A}, \uac{}{}{\textit{BA}}) \nonumber \\
		\geq & -p^B \times \sum_{j=1}^{n}G_{\textit{SAMM}}(\poolcontent{j}{}{A},\frac{p^A}{p^B}\poolcontent{j}{}{A}, \uac{j}{}{\textit{BA}})\\
		\stackrel{(\ref{lemma_optimal_trader_action_proof1})}{=} & U^{\textit{BA}}( \poolcontentvector{}{}{}, a^{\textit{BA}}) \,\, . 
	\end{align*}
\end{proof}

\subsubsection{Using smallest shards is a dominant strategy}

Lemma~\ref{lemma_optimal_trader_action} indicates that when multiple AMM shards have the same smallest amount of deposited tokens, acquiring all tokens in any one of them has the highest utility.
Since the utility of a trader's strategy is the linear combination of the utility of actions, it is a dominant strategy for the trader to randomly select one of the smallest shards to acquire all required tokens:
\begin{restatable}{cor}{lemmadominanttraderstrategy}
	In $\Gamma_n(tf_{\textit{SAMM}})$, a dominant strategy for a \textit{BA} trader is to randomly select one of the smallest shards to acquire all required tokens:
	\begin{align*}
		\tau^{BA}(\poolcontentvector{}{}{}, \uac{}{}{\textit{BA}}, a^{\textit{BA}}) = \begin{cases}
			\frac{1}{n_{\min}(\poolcontentvector{}{}{})}, & \text{if } a^{\textit{BA}} \in \SmallestAction(\uac{}{}{\textit{BA}}, \poolcontentvector{}{}{})  \\
			0, & \text{otherwise}
		  \end{cases}
	\end{align*}
\end{restatable}

If $n_{\min}(\poolcontentvector{}{}{}) = n$, then all shards have the same amount of deposited tokens, and the trader randomly selects one of the $n$ shards.
\begin{restate}{Corollary~\ref{theorem:perfect_parallelism}}
	\contenttheoremperfectparallelism
\end{restate}

\subsubsection{All dominant strategies use smallest shards}

We have shown that only trading in one of the smallest shards is the dominant strategy for the trader.
However, to later determine the best response of liquidity providers, we need to know whether there are other dominant strategies.
We show that if the trader has a positive probability of taking the action of splitting a transaction or trading in a shard with not the smallest amount of deposited tokens, then she has strictly lower utility than randomly selecting one of the smallest shards to trade, as we intended:
\begin{restate}{Theorem~\ref{theorem:strict_single}}
	In $\Gamma_n(tf_{\textit{SAMM}})$, considering the following dominant strategy of the \textit{BA} trader which randomly selects one of the smallest shards to acquire all required tokens:
	\begin{align}\label{theorem_strict_single_strategy_BA}
		\tau^{BA}(\poolcontentvector{}{}{}, \uac{}{}{\textit{BA}}, a^{\textit{BA}}) = \begin{cases}
			\frac{1}{n_{\min}(\poolcontentvector{}{}{})}, & \text{if } a^{\textit{BA}} \in \SmallestAction(\uac{}{}{\textit{BA}}, \poolcontentvector{}{}{})  \\
			0, & \text{otherwise}
		  \end{cases},
	\end{align}
	then for all strategies $\pi^{BA}$ that have a positive probability of actions not trading in one of the smallest shards, i.e., ${\exists a^{\textit{BA}} = \left(\uac{1}{}{\textit{BA}}, \cdots, \uac{i}{}{\textit{BA}}, \cdots, \uac{n}{}{\textit{BA}} \right)\notin \SmallestAction(\uac{}{}{\textit{BA}}, \poolcontentvector{}{}{})}$, ${\pi^{BA}(\poolcontentvector{}{}{}, \uac{}{}{\textit{BA}}, a^{\textit{BA}}) > 0}$, the utility of the \textit{BA} trader is strictly lower than with strategy $\tau^{BA}$:
	\begin{align*}
		U^{\textit{BA}}(\tau^{BA}, \poolcontentvector{}{}{}, \uac{}{}{\textit{BA}}) > U^{\textit{BA}}(\pi^{BA}, \poolcontentvector{}{}{}, \uac{}{}{\textit{BA}}) \,\, .
	\end{align*}
\end{restate}
\begin{proof}[Proof Sketch]
	Since the utility of the \textit{BA} trader is a linear combination of the utility of actions, we only need to show that the action of not trading in one of the smallest shards has strictly lower revenue than the action of trading in one of the smallest shards, which can be deduced from $c$-smaller-better property and $c$-non-splitting property.
\end{proof}

\begin{proof}
	Denote the minimal amount of deposited \textit{token~A} among all shards by $\poolcontent{\min}{}{A} = \min_{1 \leq i\leq n} \poolcontent{i}{}{A}$, then all smallest shards have $\poolcontent{\min}{}{A}$ deposited \textit{token~A} and $\frac{p^A}{p^B}\poolcontent{\min}{}{A}$ deposited 	\textit{token~B}.
	Then, from Equation~\ref{utilitybapi}, the utility of the \textit{BA} trader under strategy $\tau^{BA}$ is
	\begin{equation*}
		U^{\textit{BA}}(\poolcontentvector{}{}{}, \uac{}{}{\textit{BA}}, \tau^{BA})  = 
		\sum_{a^{\textit{BA}} \in \ActionSpaceBA(\uac{}{}{\textit{BA}})} \tau^{BA}(\poolcontentvector{}{}{}, \uac{}{}{\textit{BA}}, a^{\textit{BA}}) \times U^{\textit{BA}}(\poolcontentvector{}{}{}, a^{\textit{BA}}) \,\, .
	\end{equation*}
	From the definition of $\tau^{BA}$ (Equation~\ref{theorem_strict_single_strategy_BA}), the above equation can be expanded to
	\begin{equation}\label{theorem_strict_single_proof1}
		U^{\textit{BA}}(\poolcontentvector{}{}{}, \uac{}{}{\textit{BA}}, \tau^{BA})  = 
		\sum_{a_i^{\textit{BA}}(\uac{}{}{\textit{BA}}) \in \SmallestAction(\uac{}{}{\textit{BA}})} \frac{1}{n_{\min}(\poolcontentvector{}{}{})} \times U^{\textit{BA}}(\poolcontentvector{}{}{}, a_i^{\textit{BA}}) \,\, .
	\end{equation}
	The revenue of the \textit{BA} trader under action $a_i^{\textit{BA}} \in \SmallestAction(\uac{}{}{\textit{BA}})$ is (using Equation~\ref{utilityba})
	\begin{align*}
		U^{\textit{BA}}(\poolcontentvector{}{}{}, a_i^{\textit{BA}}) =& -p^B \times \sum_{i=1}^{n}G_{\textit{SAMM}}(\poolcontent{i}{}{A},\poolcontent{i}{}{B}, \uac{i}{}{\textit{BA}}) \nonumber \\
		=& -p^B \times G_{\textit{SAMM}}(\poolcontent{i}{}{A}, \frac{p^A}{p^B}\poolcontent{i}{}{B}, \uac{}{}{\textit{BA}}) \,\, .
	\end{align*}
	Since $a_i^{\textit{BA}} \in \SmallestAction(\uac{}{}{\textit{BA}})$, we have $\poolcontent{i}{}{A} = \poolcontent{\min}{}{A}$, which means
	\begin{equation}\label{theorem_strict_single_proof2}
		U^{\textit{BA}}(\poolcontentvector{}{}{}, a_i^{\textit{BA}}) 
		= -p^B \times G_{\textit{SAMM}}(\poolcontent{\min}{}{A}, \frac{p^A}{p^B}\poolcontent{\min}{}{A}, \uac{}{}{\textit{BA}}) \,\,.
	\end{equation}
	Combining Equation~\ref{theorem_strict_single_proof1} and~\ref{theorem_strict_single_proof2}, we find the utility of the \textit{BA} trader with strategy $\tau^{BA}$
	\begin{align*}
		&U^{\textit{BA}}(\poolcontentvector{}{}{}, \uac{}{}{\textit{BA}}, \tau^{BA}) \nonumber\\
		=& \sum_{a_i^{\textit{BA}}(\uac{}{}{\textit{BA}}) \in \SmallestAction(\uac{}{}{\textit{BA}})} \frac{\left(
			-p^B \times G_{\textit{SAMM}}(\poolcontent{\min}{}{A}, \frac{p^A}{p^B}\poolcontent{\min}{}{A}, \uac{}{}{\textit{BA}})\right)}{n_{\min}(\poolcontentvector{}{}{})}   \nonumber \\
		= & -p^B \times G_{\textit{SAMM}}(\poolcontent{\min}{}{A}, \frac{p^A}{p^B}\poolcontent{\min}{}{A}, \uac{}{}{\textit{BA}})  
		 \,\, .
	\end{align*}
	Combining Equation~\ref{theorem_strict_single_proof2}, we have
	\begin{equation}\label{theorem_strict_single_proof15}
		U^{\textit{BA}}(\poolcontentvector{}{}{}, a_i^{\textit{BA}}) = U^{\textit{BA}}(\poolcontentvector{}{}{}, \uac{}{}{\textit{BA}}, \tau^{BA}) \,\, .
	\end{equation}

	Now, consider the utility of a strategy $\pi^{BA}$.
	Considering any strategy $\pi^{BA}$ that splits the trade, i.e., $\exists \tilde{a}^{\textit{BA}} = \left(\uac{1}{}{\textit{BA}}, \cdots, \uac{i}{}{\textit{BA}}, \cdots, \uac{n}{}{\textit{BA}} \right)	 \in \ActionSpaceBA(\uac{}{}{\textit{BA}}), \pi^{BA}(\poolcontentvector{}{}{}, \uac{}{}{\textit{BA}}, a^{\textit{BA}}) > 0, \exists i \neq j, \uac{i}{}{\textit{BA}} > 0, \uac{j}{}{\textit{BA}} > 0$.

	Now we consider the revenue of the deviation actions.
	Considering any strategy $\pi^{BA}$ where $\exists \tilde{a}^{\textit{BA}} 
	\in \ActionSpaceBA(\uac{}{}{\textit{BA}}) \setminus \SmallestAction(\uac{}{}{\textit{BA}}, \poolcontentvector{}{}{})
	, \pi^{BA}(\poolcontentvector{}{}{}, \uac{}{}{\textit{BA}}, a^{\textit{BA}}) > 0.$
	There are two kinds of deviation actions.
	One is that the trader splits the transaction, namely $\tilde{a}^{\textit{BA}} \in \SingleAction(\uac{}{}{\textit{BA}}, \poolcontentvector{}{}{})$.
	$ \exists i \neq j, \uac{i}{}{\textit{BA}} > 0, \uac{j}{}{\textit{BA}} > 0$.
	The other is that the trader trades in a non-smallest shard, that is, $\tilde{a}^{\textit{BA}} \in \mathcal{A}^S(\uac{}{}{\textit{BA}}) \setminus \SmallestAction(\uac{}{}{\textit{BA}}, \poolcontentvector{}{}{})$.
	
	For the first case where $a^{\textit{BA}} \in \ActionSpaceBA(\uac{}{}{\textit{BA}}) \setminus \SmallestAction(\uac{}{}{\textit{BA}}, \poolcontentvector{}{}{})$, we have $\exists i \neq j, \uac{i}{}{\textit{BA}} > 0, \uac{j}{}{\textit{BA}} > 0$.
	The revenue of the \textit{BA} trader under this action is
	\begin{equation}\label{theorem_strict_single_proof6}
		U^{\textit{BA}}(\poolcontentvector{}{}{}, \tilde{a}^{\textit{BA}}) = -p^B \times \sum_{i = 1}^n \left(  G_{\textit{SAMM}}(\poolcontent{i}{}{A},\frac{p^A}{p^B}\poolcontent{i}{}{A}, \uac{i}{}{\textit{BA}}) \right) \,\, .
	\end{equation}

	From the $c$-smaller-better property, since $\forall 1 \leq i \leq n, \poolcontent{i}{}{A} \geq \poolcontent{\min}{}{A}$, we have
	\begin{equation*}
		\sum_{i = 1}^n \left(  G_{\textit{SAMM}}(\poolcontent{i}{}{A},\frac{p^A}{p^B}\poolcontent{i}{}{A}, \uac{i}{}{\textit{BA}}) \right) \geq   \sum_{i = 1}^n \left(  G_{\textit{SAMM}}(\poolcontent{\min}{}{A},\frac{p^A}{p^B}\poolcontent{\min}{}{A}, \uac{i}{}{\textit{BA}}) \right) \,\, .
	\end{equation*}
	Since $\exists i \neq j, \uac{i}{}{\textit{BA}} > 0, \uac{j}{}{\textit{BA}} > 0$, according to the $c$-non-splitting property, we have
	\begin{equation*}
		\sum_{i = 1}^n \left(  G_{\textit{SAMM}}(\poolcontent{\min}{}{A},\frac{p^A}{p^B}\poolcontent{\min}{}{A}, \uac{i}{}{\textit{BA}}) \right) > G_{\textit{SAMM}}(\poolcontent{\min}{}{A},\frac{p^A}{p^B}\poolcontent{\min}{}{A}, \uac{}{}{\textit{BA}}) \,\, .
	\end{equation*}
	Therefore, $\forall \tilde{a}^{\textit{BA}} \in \ActionSpaceBA(\uac{}{}{\textit{BA}}) \setminus \SmallestAction(\uac{}{}{\textit{BA}}, \poolcontentvector{}{}{})$, the revenue of the action $\tilde{a}^{\textit{BA}}$ (Equation~\ref{theorem_strict_single_proof6}) can be expanded as
	\begin{align*}
		U^{\textit{BA}}(\poolcontentvector{}{}{}, \tilde{a}^{\textit{BA}})&=		-p^B \times \sum_{i = 1}^n \left(  G_{\textit{SAMM}}(\poolcontent{i}{}{A},\frac{p^A}{p^B} \nonumber \poolcontent{i}{}{A}, \uac{i}{}{\textit{BA}}) \right) \\
		<& -p^B \times G_{\textit{SAMM}}(\poolcontent{\min}{}{A},\frac{p^A}{p^B}\poolcontent{\min}{}{A}, \uac{}{}{\textit{BA}}) \,\, .
	\end{align*}
	Since the right part of the above inequality is the revenue of the action $a_i^{\textit{BA}}$ (Equation~\ref{theorem_strict_single_proof2}), the revenue of action $\tilde{a}^{\textit{BA}}$ is strictly smaller than the revenue of action $a_i^{\textit{BA}}$:
	\begin{equation}\label{theorem_strict_single_proof9}
		U^{\textit{BA}}(\poolcontentvector{}{}{}, \tilde{a}^{\textit{BA}}) < U^{\textit{BA}}(\poolcontentvector{}{}{}, a_i^{\textit{BA}}) \,\, .
	\end{equation}
	
	Now we turn to the second case of an action that acquiring all tokens in a non-smallest shard, i.e., ${\tilde{a}^{\textit{BA}} \in \mathcal{A}^S(\uac{}{}{\textit{BA}})\setminus \SmallestAction(\uac{}{}{\textit{BA}}, \poolcontentvector{}{}{})}$.
	Here, we rewrite $\tilde{a}^{\textit{BA}}$ as $\tilde{a}_j^{\textit{BA}}$, the action acquiring all $\uac{}{}{\textit{BA}}$ in $\textit{shard}_j$, where $\poolcontent{j}{}{A} > \poolcontent{\min}{}{A}$.
	Then, the revenue of the \textit{BA} trader is 
	\begin{equation}\label{theorem_strict_single_proof10}
		U^{\textit{BA}}(\poolcontentvector{}{}{}, \tilde{a}_j^{\textit{BA}}) = -p^B \times G_{\textit{SAMM}}(\poolcontent{j}{}{A},\frac{p^A}{p^B}\poolcontent{j}{}{A}, \uac{}{}{\textit{BA}}) \,\, .
	\end{equation}

	From the $c$-smaller-better property, since $\poolcontent{j}{}{A} > \poolcontent{\min}{}{A}$, we have
	\begin{equation*}
		G_{\textit{SAMM}}(\poolcontent{j}{}{A},\frac{p^A}{p^B}\poolcontent{j}{}{A}, \uac{}{}{\textit{BA}}) > G_{\textit{SAMM}}(\poolcontent{\min}{}{A},\frac{p^A}{p^B}\poolcontent{\min}{}{A}, \uac{}{}{\textit{BA}}) \,\, .
	\end{equation*}
	Therefore, from Equation~\ref{theorem_strict_single_proof10}, we have
	\begin{equation*}
		U^{\textit{BA}}(\poolcontentvector{}{}{}, \tilde{a}_j^{\textit{BA}})<
		-p^B \times G_{\textit{SAMM}}(\poolcontent{\min}{}{A},\frac{p^A}{p^B}\poolcontent{\min}{}{A}, \uac{}{}{\textit{BA}}) \,\, .
	\end{equation*}
	Since the right part of the above inequality is the utility of the action $a_i^{\textit{BA}}$ (Equation~\ref{theorem_strict_single_proof2}), the revenue of the action $\tilde{a}_j^{\textit{BA}}$ is strictly lower than the revenue of the action $a_i^{\textit{BA}}$ when $\tilde{a}_j^{\textit{BA}} \in \mathcal{A}^S(\uac{}{}{\textit{BA}}) \setminus \SmallestAction(\uac{}{}{\textit{BA}}, \poolcontentvector{}{}{})$:
	\begin{equation}\label{theorem_strict_single_proof11}
		U^{\textit{BA}}(\poolcontentvector{}{}{}, \tilde{a}_j^{\textit{BA}}) < U^{\textit{BA}}(\poolcontentvector{}{}{}, a_i^{\textit{BA}}) \,\, .
	\end{equation}

	Combining the conditions for Equation~\ref{theorem_strict_single_proof9} and~\ref{theorem_strict_single_proof11}, then $\forall \tilde{a}^{\textit{BA}} \in \ActionSpaceBA(\uac{}{}{\textit{BA}}) \setminus \SmallestAction(\uac{}{}{\textit{BA}}, \poolcontentvector{}{}{})$, we have
	\begin{equation}\label{theorem_strict_single_proof12}
		U^{\textit{BA}}(\poolcontentvector{}{}{}, \tilde{a}^{\textit{BA}}) < U^{\textit{BA}}(\poolcontentvector{}{}{}, a_i^{\textit{BA}}) \,\, .
	\end{equation}

	Now we return to the utility of the \textit{BA} trader with strategy~$\pi^{BA}$.
	We tease out the deviating action:
	\begin{align}\label{theorem_strict_single_proof4}
		&U^{\textit{BA}}(\poolcontentvector{}{}{}, \uac{}{}{\textit{BA}}, \pi^{BA}) \nonumber \\
		=& \sum_{a^{\textit{BA}} \in \ActionSpaceBA(\uac{}{}{\textit{BA}})} \pi^{BA}(\poolcontentvector{}{}{}, \uac{}{}{\textit{BA}}, a^{\textit{BA}}) \times U^{\textit{BA}}(\poolcontentvector{}{}{}, a^{\textit{BA}}) \nonumber \\
		=& \sum_{a^{\textit{BA}} \in \ActionSpaceBA(\uac{}{}{\textit{BA}}) \setminus \{\tilde{a}^{\textit{BA}}\}}  \pi^{BA}(\poolcontentvector{}{}{}, \uac{}{}{\textit{BA}}, a^{\textit{BA}}) \times U^{\textit{BA}}(\poolcontentvector{}{}{}, a^{\textit{BA}}) \nonumber \\
		+&  \pi^{BA}(\poolcontentvector{}{}{}, \uac{}{}{\textit{BA}}, \tilde{a}^{\textit{BA}}) \times U^{\textit{BA}}(\poolcontentvector{}{}{}, \tilde{a}^{\textit{BA}}) \,\,.
	\end{align}

	From lemma~\ref{lemma_optimal_trader_action}, the revenue of any action ${{a}_t^{\textit{BA}} \in \ActionSpaceBA(\uac{}{}{\textit{BA}})}$ is no larger than the revenue of the action $a_i^{\textit{BA}} \in \SmallestAction(\uac{}{}{\textit{BA}})$.
	Combining Equation~\ref{theorem_strict_single_proof2}, we have
	\begin{multline}\label{theorem_strict_single_proof5}
		\sum_{a^{\textit{BA}} \in \ActionSpaceBA(\uac{}{}{\textit{BA}}) \setminus \{\tilde{a}^{\textit{BA}}\}}  \pi^{BA}(\poolcontentvector{}{}{}, \uac{}{}{\textit{BA}}, a^{\textit{BA}}) \times U^{\textit{BA}}(\poolcontentvector{}{}{}, a^{\textit{BA}}) \leq \\ 
		\sum_{a^{\textit{BA}} \in \ActionSpaceBA(\uac{}{}{\textit{BA}}) \setminus \{\tilde{a}^{\textit{BA}}\}}  \pi^{BA}(\poolcontentvector{}{}{}, \uac{}{}{\textit{BA}}, a^{\textit{BA}}) \times U^{\textit{BA}}(\poolcontentvector{}{}{}, a_i^{\textit{BA}}) 
	\end{multline}
	Combining Equations~\ref{theorem_strict_single_proof12} anc~\ref{theorem_strict_single_proof5}, we expand the utility of the \textit{BA} trader under strategy $\pi^{BA}$ in Equation~\ref{theorem_strict_single_proof4} as
	
	\begin{align*}\label{theorem_strict_single_proof8}
		&U^{\textit{BA}}(\poolcontentvector{}{}{}, \uac{}{}{\textit{BA}}, \pi^{BA}) \nonumber \\
		=& \sum_{a^{\textit{BA}} \in \ActionSpaceBA(\uac{}{}{\textit{BA}})- \{\tilde{a}^{\textit{BA}}\}}  \pi^{BA}(\poolcontentvector{}{}{}, \uac{}{}{\textit{BA}}, a^{\textit{BA}}) \times U^{\textit{BA}}(\poolcontentvector{}{}{}, a^{\textit{BA}}) \nonumber \\
		+&  \pi^{BA}(\poolcontentvector{}{}{}, \uac{}{}{\textit{BA}}, \tilde{a}^{\textit{BA}}) \times U^{\textit{BA}}(\poolcontentvector{}{}{}, \tilde{a}^{\textit{BA}}) \nonumber \\
		\stackrel{(\ref{theorem_strict_single_proof5})}{\leq}
		 & \sum_{a^{\textit{BA}} \in \ActionSpaceBA(\uac{}{}{\textit{BA}})- \{\tilde{a}^{\textit{BA}}\}}  \pi^{BA}(\poolcontentvector{}{}{}, \uac{}{}{\textit{BA}}, a^{\textit{BA}}) \times U^{\textit{BA}}(\poolcontentvector{}{}{}, a_i^{\textit{BA}}) \nonumber \\
		+&  \pi^{BA}(\poolcontentvector{}{}{}, \uac{}{}{\textit{BA}}, \tilde{a}^{\textit{BA}}) \times U^{\textit{BA}}(\poolcontentvector{}{}{}, \tilde{a}^{\textit{BA}}) \nonumber \\
		\stackrel{(\ref{theorem_strict_single_proof12})}{<}& \sum_{a^{\textit{BA}} \in \ActionSpaceBA(\uac{}{}{\textit{BA}})- \{\tilde{a}^{\textit{BA}}\}}  \pi^{BA}(\poolcontentvector{}{}{}, \uac{}{}{\textit{BA}}, a^{\textit{BA}}) \times U^{\textit{BA}}(\poolcontentvector{}{}{}, a_i^{\textit{BA}}) \nonumber \\
		+&  \pi^{BA}(\poolcontentvector{}{}{}, \uac{}{}{\textit{BA}}, \tilde{a}^{\textit{BA}}) \times U^{\textit{BA}}(\poolcontentvector{}{}{}, a_i^{\textit{BA}}) \nonumber \\
		=& \sum_{a^{\textit{BA}} \in \ActionSpaceBA(\uac{}{}{\textit{BA}})} \pi^{BA}(\poolcontentvector{}{}{}, \uac{}{}{\textit{BA}}, a^{\textit{BA}}) \times U^{\textit{BA}}(\poolcontentvector{}{}{}, a_i^{\textit{BA}}) \nonumber \\
		=& U^{\textit{BA}}(\poolcontentvector{}{}{}, a_i^{\textit{BA}}) \times \sum_{a^{\textit{BA}} \in \ActionSpaceBA(\uac{}{}{\textit{BA}})} \pi^{BA}(\poolcontentvector{}{}{}, \uac{}{}{\textit{BA}}, a^{\textit{BA}}) \nonumber \\
		=& U^{\textit{BA}}(\poolcontentvector{}{}{}, a_i^{\textit{BA}})  \nonumber \\
		\stackrel{(\ref{theorem_strict_single_proof15})}{=} & U^{\textit{BA}}(\poolcontentvector{}{}{}, \uac{}{}{\textit{BA}}, \pi^{BA}) \,\, .
	\end{align*}

	The above inequality indicates that the utility of the \textit{BA} trader under strategy $\pi^{BA}$ is strictly lower than the utility of the \textit{BA} trader under strategy $\tau^{BA}$.
\end{proof}

From the above theorem, all the best responses of the trader should only have a positive probability of taking an action that trades in exactly one of the smallest shards:
\begin{restatable}{cor}{lemmaoptimaltraderstrategyset}\label{lemma:optimal_trader_strategy_set}
	Considering any best response strategy $\tau^{BA}(\poolcontentvector{}{}{}, \uac{}{}{\textit{BA}}, a^{\textit{BA}})$ of the \textit{BA} trader, the strategy should only have a positive probability of taking an action that trades in one of the smallest shards:
	\begin{equation*}
		\forall a_i^{\textit{BA}}(\uac{}{}{\textit{BA}}) \in \ActionSpaceBA(\uac{}{}{\textit{BA}}) \setminus \SmallestAction(\uac{}{}{\textit{BA}}), \tau^{BA}(\poolcontentvector{}{}{}, \uac{}{}{\textit{BA}}, a_i^{\textit{BA}}) = 0 \,\,.
	\end{equation*}
	In other words, the sum of the probabilities of all actions that trade in one of the smallest shards should be 1:
	\begin{equation*}
		\sum_{a_i^{\textit{BA}}(\uac{}{}{\textit{BA}}) \in \SmallestAction(\uac{}{}{\textit{BA}})}
		\tau^{BA}(\poolcontentvector{}{}{}, \uac{}{}{\textit{BA}}, a_i^{\textit{BA}}(\uac{}{}{\textit{BA}})) = 1 \,\, .
	\end{equation*}
\end{restatable}

\subsection{Liquidity Provider Strategy and SPNE}\label{gt:lp_strategy}

Our analysis of the trader strategy shows that if the shards are balanced, traders behave as intended. 
Next, we consider the strategies of liquidity providers in equilibrium.
They should maintain the shard balance. 
Moreover, their incentives should rebalance the system state even in the face of attacks that break this balance.

\subsubsection{Scaffolding}

We want a liquidity provider to fill up smaller shards to keep shards balanced.
We call such an action the \emph{fillup action}, where if the liquidity provider adds tokens to a shard, then the shard is one of the smallest shards after this action.
We denote the fillup action by~$a_{lp}^{\textit{fill}}(\poolcontentvector{}{}{}, \lpac{}{}{A}, \lpac{}{}{B})$:
\begin{restate}{Definition~\ref{def:fillup_action}}
	\deffillupaction
\end{restate}

We also define the strategy that only takes the fillup action as the \emph{fillup strategy}:
\begin{restate}{Definition~\ref{def:fillup_strategy}}
	\deffillupstrategy
\end{restate}

Denote the minimal amount of deposited \textit{token~A} among all shards of status $\poolcontentvector{}{}{}= \left(\left(\poolcontent{1}{}{A},\poolcontent{1}{}{B}, \poolcontent{1}{}{S} \right), \cdots \left(\poolcontent{n}{}{A},\poolcontent{n}{}{B}, \poolcontent{n}{}{S}	 \right) \right)$ by
\begin{equation*}
	\rho^A(\poolcontentvector{}{}{}) = \min_{1 \leq i\leq n} \poolcontent{i}{}{A}\,.
\end{equation*}

We show that $a_{lp}^{\textit{fill}}$ is unique and has the maximal volume of the smallest shard in the next step.

\newcommand{\contentlemmauniquefillup}{For any action of liquidity provider ${a_{lp} = \left(\left(\lpac{1}{}{A}, \lpac{1}{}{B}\right), \cdots, \left(\lpac{n}{}{A}, \lpac{n}{}{B}\right)\right) \in \mathcal{A}_{lp}(\lpac{}{}{A} , \lpac{}{}{B} )}$, if 
	$\rho^A(\poolcontentvector{}{}{} + a_{lp}) \geq \rho^A(\poolcontentvector{}{}{} + a_{lp}^{\textit{fill}})$, then $a_{lp} = a_{lp}^{\textit{fill}}$.}

\begin{restatable}{lem}{lemmauniquefillup}\label{lemma:unique_fillup}
	\contentlemmauniquefillup
\end{restatable}

\begin{proof}[Proof Sketch]
	If another action results in a higher minimum reserve of \textit{token~A}, then this action must add a larger amount of tokens into each shard compared to the fill-up action, contrary to the assumption that actions have identical total input amounts.
\end{proof}

\begin{proof}
	Consider any $j$ s.t. $\lpac{j}{}{A^{\ast}} >0$, from the definition of $a_{lp}^{\textit{fill}}$ (Definition~\ref{def:fillup_action}), $\lpac{j}{}{A^{\ast}}  + \poolcontent{j}{}{A}$ is the minimal among all shards with state $\poolcontentvector{}{}{} + a_{lp}$, i.e.,
	\begin{equation*}
		\rho^A(\poolcontentvector{}{}{} + a_{lp}^{\textit{fill}}) = \fillupac{j}{}{A}  + \poolcontent{j}{}{A} \,\,.
	\end{equation*}

	Since then all shards in ${\poolcontentvector{}{}{} + a_{lp}}$ have no less than $\rho^A(\poolcontentvector{}{}{} + a_{lp})$ deposited \textit{token~A}, if $\rho^A(\poolcontentvector{}{}{} + a_{lp}) \geq \rho^A(\poolcontentvector{}{}{} + a_{lp}^{\textit{fill}})$, we have 
	\begin{equation*}
		\lpac{j}{}{A} + \poolcontent{j}{}{A} \geq \rho^A(\poolcontentvector{}{}{} + a_{lp}) \geq \rho^A(\poolcontentvector{}{}{} + a_{lp}^{\textit{fill}}) = \fillupac{j}{}{A}  + \poolcontent{j}{}{A} \,\,.
	\end{equation*}

	Therefore, we have
	\begin{equation}\label{lemma_fillup_proof4}
		\lpac{j}{}{A} \geq \fillupac{j}{}{A} \,\,.
	\end{equation}
	Considering the sum of $\lpac{j}{}{A}$ and $\fillupac{j}{}{A}$, we have 
	\begin{equation}\label{lemma_fillup_proof1}
		\sum_{\lpac{j}{}{A} >0} \lpac{i}{}{A} \geq \sum_{\fillupac{j}{}{A} >0} \fillupac{j}{}{A} = \lpac{}{}{A} \,\,.
	\end{equation}
	For $a_{lp}$, the sum of input tokens in all shards is $\lpac{}{}{A}$:
	\begin{equation}\label{lemma_fillup_proof2}
		\sum_{i=1}^n \lpac{i}{}{A} = \lpac{}{}{A} \,\,.
	\end{equation}
	And $\forall 1\leq i \leq n$, $\lpac{i}{}{A}$ is non-negative,
	\begin{equation}\label{lemma_fillup_proof3}
		\lpac{i}{}{A} \geq 0 \,\,.
	\end{equation}
	Combining Expressions~\ref{lemma_fillup_proof4}~\ref{lemma_fillup_proof1},~\ref{lemma_fillup_proof2} and~\ref{lemma_fillup_proof3},
	all inequation holds with equality, which indicates that $a_{lp} = a_{lp}^{\textit{fill}}$:

	\begin{equation*}
		\forall 1 \leq i \leq n, \lpac{i}{}{A} = \fillupac{i}{}{A} \,\,. 
	\end{equation*}
\end{proof}

We have shown that traders are incentivized to trade in smaller shards in Section~\ref{gt:trader_strategy}, which incentivizes liquidity providers to add their tokens to smaller shards.
Additionally, if the liquidity provider makes small shards larger, she would get more trading fees, which further incentivizes them to add liquidity to small shards:
\newcommand{\contentlemmauniquefillupmaximal}{For any two shards $i$ and $j$, if $\poolcontent{i}{}{A} < \poolcontent{j}{}{A}$, for any output amount $\uac{}{}{\textit{BA}}$ of \textit{token~A}, the trading fee of $\textit{shard}_i$ is strictly smaller than the trading fee of $\textit{shard}_j$:
\begin{equation*}
	tf_{\textit{SAMM}}(\poolcontent{i}{}{A},\frac{p^A}{p^B}\poolcontent{i}{}{A}, \uac{}{}{\textit{BA}}) < tf_{\textit{SAMM}}(\poolcontent{j}{}{A},\frac{p^A}{p^B}\poolcontent{j}{}{A}, \uac{}{}{\textit{BA}}) \,\, .
\end{equation*}}

\begin{restatable}{lem}{lemmatradingfee}\label{lemma:trading_fee}
	\contentlemmauniquefillupmaximal
\end{restatable}

\begin{proof}[Proof Sketch]
	Due to the $c$-smaller-better property, the gross amount in a larger shard is larger than that in a smaller shard.
	However, the net amount of a larger shard is smaller than that of a smaller shard (Equation~\ref{net_amount_pre}).
	Therefore, the trading fee of a larger shard is larger than that of a smaller shard since the gross amount is the sum of the net amount and the trading fee.
\end{proof}

\begin{proof}
	From $c$-smaller-better property, since $\poolcontent{i}{}{A} < \poolcontent{j}{}{A}$, the gross amount of $\textit{shard}_i$ is strictly smaller than the gross amount of $\textit{shard}_j$,
	\begin{equation*}
		G_{\textit{SAMM}}(\poolcontent{i}{}{A},\frac{p^A}{p^B}\poolcontent{i}{}{A}, \uac{}{}{\textit{BA}}) < G_{\textit{SAMM}}(\poolcontent{j}{}{A},\frac{p^A}{p^B}\poolcontent{j}{}{A}, \uac{}{}{\textit{BA}}) \,\,.
	\end{equation*}
	Since the gross amount is the sum of the net amount and the trading fee, by expanding the above inequality, we have
	\begin{multline}\label{lemma_trading_fee_proof1}
		tf_{\textit{SAMM}}(\poolcontent{i}{}{A},\frac{p^A}{p^B}\poolcontent{i}{}{A}, \uac{}{}{\textit{BA}}) + \textit{net}^B(\poolcontent{}{}{A},\frac{p^A}{p^B}\poolcontent{i}{}{A}, \pooloutput{}{}{A}) < \\  tf_{\textit{SAMM}}(\poolcontent{j}{}{A},\frac{p^A}{p^B}\poolcontent{j}{}{A}, \uac{}{}{\textit{BA}}) + \textit{net}^B(\poolcontent{}{}{A},\frac{p^A}{p^B}\poolcontent{j}{}{A}, \pooloutput{}{}{A})\,\, .
	\end{multline}
	Considering the net amount of these two shards, since ${\poolcontent{i}{}{A} < \poolcontent{j}{}{A}}$, $\textit{shard}_i$ has higher slippage than $\textit{shard}_j$:
	\begin{equation}\label{lemma_trading_fee_proof2}
		\textit{net}^B(\poolcontent{}{}{A},\frac{p^A}{p^B}\poolcontent{i}{}{A}, \pooloutput{}{}{A}) > \textit{net}^B(\poolcontent{}{}{A},\frac{p^A}{p^B}\poolcontent{j}{}{A}, \pooloutput{}{}{A}) \,\, .
	\end{equation}
	Combining Equation~\ref{lemma_trading_fee_proof1} and~\ref{lemma_trading_fee_proof2}, we conclude that the trading fee of $\textit{shard}_i$ is strictly smaller than the trading fee of $\textit{shard}_j$:
	\begin{equation*}
		tf_{\textit{SAMM}}(\poolcontent{i}{}{A},\frac{p^A}{p^B}\poolcontent{i}{}{A}, \uac{}{}{\textit{BA}}) < tf_{\textit{SAMM}}(\poolcontent{j}{}{A},\frac{p^A}{p^B}\poolcontent{j}{}{A}, \uac{}{}{\textit{BA}}) \,\, .
	\end{equation*}
\end{proof}

\subsubsection{Perfect parallelism under balanced shards}
When all shards have identical sizes, the fillup action is to add tokens to all shards evenly, which is the best response of the liquidity provider:

\begin{restate}{Theorem~\ref{theorem_perfect_parallelism_and_lp}}
	\contenttheoremperfectparallelismandlp
\end{restate}

\begin{proof}[Proof Sketch]
	Traders prefer trading in smaller shards to reduce their costs.
	At the same time, fees are higher in larger shards.
	So the liquidity provider should increase her share in the smallest shards. 
	This dual objective is optimally achieved by uniformly distributing tokens across all shards.
\end{proof}

\begin{proof}
	Without loss of generality, we only consider \textit{BA} traders since actions of \textit{AB} traders are symmetric.

	Since the traders have given the best response, we only need to prove that $\tau_{lp}$ is also the best response.

	Consider any action of liquidity provider ${a_{lp} = \left(\left(\lpac{1}{}{A}, \lpac{1}{}{B}\right), \cdots, \left(\lpac{n}{}{A}, \lpac{n}{}{B}\right)\right) \in \mathcal{A}_{lp}(\lpac{}{}{A} , \lpac{}{}{B} )}$, the shard state after this action is $\poolcontentvector{}{}{} + a_{lp}$.
	From Corollary~\ref{lemma:optimal_trader_strategy_set}, any best response of the trader only has a positive probability of taking an action that trades in exactly one of the smallest shards.
	Therefore, from Equation~\ref{Ulp_action_strategy_simplified}, when the trader use any best response $\tau^{BA}$, the liquidity provider's revenue with action $a_{lp}$ is

	\begin{multline}\label{theorem_perfect_parallelism_and_lp_proof1}
		U_{lp}(\poolcontentvector{}{}{},\lpac{}{}{A}, \lpac{}{}{B}, a_{lp},  \tau^{BA}, \tau^{AB}) = \\
	 E_{ \uac{}{}{\textit{BA}} \sim D^{BA}}
	\left[ 
		\sum_{\substack{a_i^{\textit{BA}}(\uac{}{}{\textit{BA}})\\\in \SmallestAction(\uac{}{}{\textit{BA}}, \poolcontentvector{}{}{}+a_{lp})}}
	\left( 
		\substack{\tau^{BA}( \poolcontentvector{}{}{} + a_{lp}, \uac{}{}{\textit{BA}}, a_i^{\textit{BA}}(\uac{}{}{\textit{BA}}))
	\times \\
		U_{lp}(\poolcontentvector{}{}{}, a_{lp},a_i^{\textit{BA}}(\uac{}{}{\textit{BA}}))}
	\right)
	\right] 
	\end{multline}
	Since $\forall 1 \leq i, j \leq n$, $\poolcontent{i}{}{A} = \poolcontent{j}{}{A}$. 
	Therefore, the minimal deposited amount of \textit{token~A} among all shards in $\poolcontentvector{}{}{} + a_{lp}$ is:
	\begin{equation*}
		\rho^A(\poolcontentvector{}{}{} + a_{lp}) = \poolcontent{1}{}{A} + \lpac{\min}{}{A} \,\, .
	\end{equation*}

	From Equation~\ref{TFBA}, the revenue of a liquidity provider due to action $a_{lp}$ and the \textit{BA} trader action ${a_i^{\textit{BA}}(\uac{}{}{\textit{BA}}) \in \SmallestAction(\uac{}{}{\textit{BA}}, \poolcontentvector{}{}{} + a_{lp})}$ is acquiring all $\uac{}{}{\textit{BA}}$ in one of the smallest shard after the liquidity provider's action, say $\textit{shard}_i$, 
	is
	\begin{align*}
		&U_{lp}(\poolcontentvector{}{}{}, a_{lp},a_i^{\textit{BA}}(\uac{}{}{\textit{BA}})) \nonumber \\
		=&p^B \times
	 \textit{tf}( \poolcontent{i}{}{A} + \lpac{i}{}{A},  \frac{p^A}{p^B}\left(\poolcontent{i}{}{A} + \lpac{i}{}{A}\right), \uac{}{}{\textit{BA}})  
	 \times \frac{\lpac{i}{}{A}}{ \lpac{i}{}{A} +  \poolcontent{i}{}{A}}  \nonumber \\
	 =& p^B \times \textit{tf}( \poolcontent{1}{}{A} + \lpac{\min}{}{A},  \frac{p^A}{p^B}\left(\poolcontent{1}{}{A} + \lpac{\min}{}{A}\right), \uac{}{}{\textit{BA}})	\times \frac{\lpac{\min}{}{A}}{ \lpac{\min}{}{A} +  \poolcontent{1}{}{A}} 
	 \,\, .
	\end{align*}

	When $a_{lp} = \hat{a}_{lp}$, the liquidity provider adds tokens to all shards evenly, i.e., $\lpac{i}{}{A} = \frac{1}{n}\lpac{}{}{A}$.
	Then each shard is identical, the trader's choice of $\textit{shard}_i$ has the same revenue for the liquidity provider as $\textit{shard}_1$:
	\begin{equation*}\label{theorem_perfect_parallelism_and_lp_proof51}
		U_{lp}(\poolcontentvector{}{}{}, \hat{a}_{lp},a_i^{\textit{BA}}(\uac{}{}{\textit{BA}}))
		=U_{lp}(\poolcontentvector{}{}{}, \hat{a}_{lp},a_1^{\textit{BA}}(\uac{}{}{\textit{BA}})) \,\, .
	\end{equation*}

	Therefore, from Equation~\ref{theorem_perfect_parallelism_and_lp_proof1}, the revenue of the liquidity provider under action $\hat{a}_{lp}$ is
	\begin{multline}\label{theorem_perfect_parallelism_and_lp_proof50}
		U_{lp}(\poolcontentvector{}{}{},\lpac{}{}{A}, \lpac{}{}{B}, \hat{a}_{lp},  \tau^{BA}, \tau^{AB}) = \\
	 E_{ \uac{}{}{\textit{BA}} \sim D^{BA}}
	\left[ 
		U_{lp}(\poolcontentvector{}{}{},\hat{a}_{lp},a_i^{\textit{BA}}(\uac{}{}{\textit{BA}}))
	\right]  \,\, .
	\end{multline}

	Since $\sum_{i=1}^n \lpac{i}{}{A} = \lpac{}{}{A}$, when $\forall 1\leq i\leq n, \lpac{i}{}{A} = \frac{1}{n}\lpac{}{}{A}$, we have for any action $a_{lp} \in \mathcal{A}_{lp}(\lpac{}{}{A} , \lpac{}{}{B} ),
	\lpac{\min}{}{A} \leq \frac{1}{n}\lpac{}{}{A}$.
	Then, 
	\begin{equation}\label{theorem_perfect_parallelism_and_lp_proof5}
		\frac{\lpac{\min}{}{A}}{ \lpac{\min}{}{A} +  \poolcontent{1}{}{A}} \leq \frac{\frac{1}{n}\lpac{}{}{A}}{ \frac{1}{n}\lpac{}{}{A} +  \poolcontent{1}{}{A}} \,\, .
	\end{equation}

	From Lemma~\ref{lemma:trading_fee}, the trading fee of the smallest shard is no larger than the trading fee of any other shard:
	\begin{multline}\label{theorem_perfect_parallelism_and_lp_proof6}
		tf_{\textit{SAMM}}(\poolcontent{1}{}{A} + \lpac{\min}{}{A},  \frac{p^A}{p^B}\left(\poolcontent{1}{}{A} + \lpac{\min}{}{A}\right), \uac{}{}{\textit{BA}}) \leq \\ tf_{\textit{SAMM}}(\poolcontent{i}{}{A} + \frac{1}{n}\lpac{}{}{A},  \frac{p^A}{p^B}\left(\poolcontent{i}{}{A} + \frac{1}{n}\lpac{}{}{A}\right), \uac{}{}{\textit{BA}}) \,\, .
	\end{multline}
	Combining Equation~\ref{theorem_perfect_parallelism_and_lp_proof5} and~\ref{theorem_perfect_parallelism_and_lp_proof6}, we have
	\begin{multline*}
		p^B \times \textit{tf}( \poolcontent{1}{}{A} + \lpac{\min}{}{A},  \frac{p^A}{p^B}\left(\poolcontent{1}{}{A} + \lpac{\min}{}{A}\right), \uac{}{}{\textit{BA}})	\times \frac{\lpac{\min}{}{A}}{ \lpac{\min}{}{A} +  \poolcontent{1}{}{A}} \leq \\
		p^B \times \textit{tf}( \poolcontent{i}{}{A} + \frac{1}{n}\lpac{}{}{A},  \frac{p^A}{p^B}\left(\poolcontent{i}{}{A} + \frac{1}{n}\lpac{}{}{A}\right), \uac{}{}{\textit{BA}})	\times \frac{\frac{1}{n}\lpac{}{}{A}}{ \frac{1}{n}\lpac{}{}{A} +  \poolcontent{1}{}{A}} \,\, ,
	\end{multline*}
	which indicates that the revenue of the liquidity provider under action $\hat{a}_{lp}$ is not smaller than any other action when traders take action $a_i^{\textit{BA}}(\uac{}{}{\textit{BA}})$, i.e.,
	\begin{equation}\label{theorem_perfect_parallelism_and_lp_proof7}
		U_{lp}(\poolcontentvector{}{}{}, a_{lp},a_i^{\textit{BA}}(\uac{}{}{\textit{BA}})) \leq U_{lp}(\poolcontentvector{}{}{}, \hat{a}_{lp},a_i^{\textit{BA}}(\uac{}{}{\textit{BA}})) \,\, .
	\end{equation}

	Combining Equation~\ref{theorem_perfect_parallelism_and_lp_proof1} and~\ref{theorem_perfect_parallelism_and_lp_proof7}, the revenue of the liquidity provider under action $\hat{a}_{lp}$ is not smaller than the revenue of the liquidity provider under any other action when traders use any best response $\tau^{BA}$.
	\begin{align}\label{theorem_perfect_parallelism_and_lp_proof3}
		& U_{lp}(\poolcontentvector{}{}{}, a_{lp},  \tau^{BA}, \tau^{AB},\lpac{}{}{A}, \lpac{}{}{B})  \nonumber \\
	\stackrel{(\ref{theorem_perfect_parallelism_and_lp_proof7})}{\leq}& E_{ \uac{}{}{\textit{BA}} \sim D^{BA}} \left[
	\sum_{\substack{a_i^{\textit{BA}}(\uac{}{}{\textit{BA}})\\\in \SmallestAction(\uac{}{}{\textit{BA}}, \poolcontentvector{}{}{} + a_{lp})}}
		\left(
		\substack{
			\pi_{lp}( \poolcontentvector{}{}{}, \lpac{}{}{A}, \lpac{}{}{B}, a_{lp}) \\ \times
			U_{lp}(\poolcontentvector{}{}{}, \hat{a}_{lp},a_i^{\textit{BA}}(\uac{}{}{\textit{BA}}))
		}
		\right)
	\right]\nonumber \\
	\stackrel{(\ref{theorem_perfect_parallelism_and_lp_proof7})}{\leq}& E_{ \uac{}{}{\textit{BA}} \sim D^{BA}} \left[
		U_{lp}(\poolcontentvector{}{}{},\hat{a}_{lp},a_i^{\textit{BA}}(\uac{}{}{\textit{BA}}))
	\right]\nonumber \\
	\stackrel{(\ref{theorem_perfect_parallelism_and_lp_proof50})}{=}& U_{lp}(\poolcontentvector{}{}{}, \hat{a}_{lp},  \tau^{BA}, \tau^{AB},\lpac{}{}{A}, \lpac{}{}{B}) \,\, .
	\end{align}

	Consider the utility of the liquidity provider under action $\hat{a}_{lp}$ of evenly depositing tokens in all shards, system state~$\poolcontentvector{}{}{}$,~$\lpac{}{}{A}$,~ $\lpac{}{}{B}$, and the \textit{BA} trader strategy~$\tau^{BA}$.
	From the definition of the utility function of the liquidity provider (Equation~\ref{Ulp_strategy_strategy}), the utility of the liquidity provider's strategy $\tau_{lp}$ is equal to the revenue of the liquidity provider under action $\hat{a}_{lp}$:

	\begin{align}\label{theorem_perfect_parallelism_and_lp_proof54}
		&U_{lp}(\poolcontentvector{}{}{},\lpac{}{}{A}, \lpac{}{}{B}, \tau_{lp}, \tau^{BA}, \tau^{AB}) \nonumber \\
		=&\sum_{a_{lp} \in \mathcal{A}_{lp}(\lpac{}{}{A} , \lpac{}{}{B} )}
		\left(
			\substack
			{
				\tau_{lp}( \poolcontentvector{}{}{}, \lpac{}{}{A}, \lpac{}{}{B}, a_{lp}) \\
				\times U_{lp}(\poolcontentvector{}{}{},\lpac{}{}{A}, \lpac{}{}{B}, a_{lp},  \tau^{BA}, \tau^{AB})
			}
		\right)\nonumber \\
		=& U_{lp}(\poolcontentvector{}{}{},\lpac{}{}{A}, \lpac{}{}{B}, \hat{a}_{lp},  \tau^{BA}, \tau^{AB})   \,\, .
	\end{align}

	Then the liquidity provider strategy $\tau_{lp}$ has no smaller utility than $\pi^{BA}$:
	\begin{align*}
		&U_{lp}(\poolcontentvector{}{}{},\lpac{}{}{A}, \lpac{}{}{B},  \pi_{lp}, \tau^{BA}, \tau^{AB}) \nonumber \\
		\stackrel{(\ref{Ulp_strategy_strategy})}{=}& \sum_{a_{lp} \in \mathcal{A}_{lp}(\lpac{}{}{A} , \lpac{}{}{B} )} \left(
		\substack{\pi_{lp}( \poolcontentvector{}{}{}, \lpac{}{}{A}, \lpac{}{}{B}, a_{lp}) \\ \times
		U_{lp}(\poolcontentvector{}{}{},\lpac{}{}{A}, \lpac{}{}{B}, a_{lp},  \tau^{BA}, \tau^{AB})
		}
		\right) \nonumber \\
		\stackrel{(\ref{theorem_perfect_parallelism_and_lp_proof3})}{\leq} & \sum_{a_{lp} \in \mathcal{A}_{lp}(\lpac{}{}{A} , \lpac{}{}{B} )} \left(
			\substack{\pi_{lp}( \poolcontentvector{}{}{}, \lpac{}{}{A}, \lpac{}{}{B}, a_{lp}) \\ \times
			U_{lp}(\poolcontentvector{}{}{},\lpac{}{}{A}, \lpac{}{}{B}, \hat{a}_{lp},  \tau^{BA}, \tau^{AB})
			}
			\right) \nonumber \\
			\stackrel{\;\;\;\;\;\;\;}{=}& U_{lp}(\poolcontentvector{}{}{},\lpac{}{}{A}, \lpac{}{}{B}, \hat{a}_{lp},  \tau^{BA}, \tau^{AB}) \nonumber \\
		\stackrel{(\ref{theorem_perfect_parallelism_and_lp_proof54})}{=}& U_{lp}(\poolcontentvector{}{}{}, \tau_{lp},  \tau^{BA}, \tau^{AB},\lpac{}{}{A}, \lpac{}{}{B}) \nonumber \,\, .
	\end{align*}
	That is, the liquidity provider strategy $\tau_{lp}$ is the best response to the trader strategy $\tau^{BA}$.

	Therefore, the liquidity provider strategy $\tau_{lp}$ and any best response of the \textit{BA} trader $\tau^{BA}$ are an SPNE.
\end{proof}

The above theorem indicates that the liquidity provider keeps the same amount of deposited tokens in the shard after her action.
Therefore, the system always works in a state where all shards have the same amount of deposited tokens.

Since randomly choosing one of the smallest shards to trade is a dominant strategy for the trader, the system always works in perfect parallelism:

\begin{restatable}{cor}{corollaryalwaysstable}
	Denote by $\hat{a}_{lp} = \left(\left(\frac{1}{n}\lpac{}{}{A},  \frac{1}{n}\lpac{}{}{B}\right),\cdots\right)$ the action of evenly depositing tokens in all shards.
	In $\Gamma_n(tf_{\textit{SAMM}})$, if for all $i$ and $j$ that the liquidity amounts are the same, $\poolcontent{i}{}{A} = \poolcontent{j}{}{A}$ and $\poolcontent{i}{}{B} = \poolcontent{j}{}{B}$, the liquidity provider strategy which only takes action $\hat{a}_{lp}$,
	\begin{equation*}
		\tau_{lp}( \poolcontentvector{}{}{}, \lpac{}{}{A}, \lpac{}{}{B}, a_{lp}) = \begin{cases}
			1, & \text{if } a_{lp} = \hat{a}_{lp}  \\
			0, & \text{otherwise}
		  \end{cases}
	\end{equation*}
	and the \textit{BA} trader strategy of randomly selecting one of the smallest shards to trade,
	\begin{align*}
		\tau^{BA}(\poolcontentvector{}{}{}, \uac{}{}{\textit{BA}}, a^{\textit{BA}}) = \begin{cases}
			\frac{1}{n_{\min}(\poolcontentvector{}{}{})}, & \text{if } a^{\textit{BA}} \in \SmallestAction(\uac{}{}{\textit{BA}}, \poolcontentvector{}{}{})  \\
			0, & \text{otherwise}
		  \end{cases}
	\end{align*}
	constitute an SPNE.
\end{restatable}

\subsubsection{Convergence to Balanced Shards}
We now show that even if the system reaches an unbalanced state, maybe due to an attacker, it converges to a balanced state since the liquidity provider uses the fillup strategy.
We can conclude that the fillup strategy is the only best response in all SPNE:

\begin{restate}{Theorem~\ref{theorem:always_fillup}}
	\contenttheoremalwaysfillup	
\end{restate}

\begin{proof}[Proof Sketch]
	Given any action $a_{lp}$ that is not the fillup action, we can construct a new action $a_{lp}'$ that is strictly better than~$a_{lp}$.
	By ensuring the smallest shard in $\poolcontentvector{}{}{} + a_{lp}'$ is larger than that in $\poolcontentvector{}{}{} + a_{lp}$, we increase trading fees garnered from each transaction. 
	Moreover, this smallest shard is also the smallest before the action, maximizing the liquidity provider's share. 
	Consequently, the liquidity provider earns higher revenue under $a_{lp}'$ than under $a_{lp}$. 
	Thus, any strategy incorporating an action other than the fillup action is not optimal. Figure~\ref{fig:action} illustrates an example of the~$a_{lp}'$ construction.
\end{proof}

\begin{proof}
	We prove this by contradiction.
	(1)Initially, for any action that is not the fillup action, we construct another action resulting in a larger smallest shard.
	(2)Second, we prove that the constructed action has higher utility than the original action.
	(3)Third, for any strategy that does not always take the fillup action, we construct a new strategy based on the constructed action.
	(4)Finally, we prove that the new strategy has a higher utility than the original strategy, which means that the original strategy is not the best response.
	
	(1) Construction of the new action:
	Consider any action of a liquidity provider ${a_{lp} = \left(\left(\lpac{1}{}{A}, \lpac{1}{}{B}\right), \cdots, \left(\lpac{n}{}{A}, \lpac{n}{}{B}\right)\right) \in \mathcal{A}_{lp}(\lpac{}{}{A} , \lpac{}{}{B} )}$. We construct another action $a_{lp}'$.
	First, we want the smallest shard in $\poolcontentvector{}{}{} + a_{lp}'$ to be larger than the smallest shard in $\poolcontentvector{}{}{} + a_{lp}$, which lead to a higher trading fee in a single trade according to Lemma~\ref{lemma:trading_fee}.
	Second, we want the smallest shard in $\poolcontentvector{}{}{} + a_{lp}'$ to be unique and also the smallest in $\poolcontentvector{}{}{}$ to make the liquidity provider have more shares in the smallest shard than in $\poolcontentvector{}{}{} + a_{lp}$.

	Denote by $i^{\ast}$ the smallest shard in $\poolcontentvector{}{}{}$ with the smaller index, i.e.
	\begin{align}\label{theorem_always_fillup_proof20}
		\forall j, \poolcontent{i^{\ast}}{}{A} \leq \poolcontent{j}{}{A}  \,\,;\nonumber \\
		\forall j, \poolcontent{i^{\ast}}{}{A} = \poolcontent{j}{}{A} \Rightarrow i^{\ast} \leq j \,\,.
	\end{align}
	Consider the fillup action $a_{lp}^{\textit{fill}}(\poolcontentvector{}{}{}, \lpac{}{}{A}, \lpac{}{}{B}) = \left(\left(\fillupac{1}{}{A}, \fillupac{1}{}{B}\right), \cdots, \left(\fillupac{n}{}{A}, \fillupac{n}{}{B}\right)\right)$.
	If $a_{lp} = a_{lp}^{\textit{fill}}(\poolcontentvector{}{}{}, \lpac{}{}{A}, \lpac{}{}{B})$, we take ${a_{lp}' = a_{lp}}$ and we are done.
	If ${a_{lp} \neq a_{lp}^{\textit{fill}}(\poolcontentvector{}{}{}, \lpac{}{}{A}, \lpac{}{}{B})}$, from Lemma~\ref{lemma:unique_fillup}, we have ${\rho^A(\poolcontentvector{}{}{}, \lpac{}{}{A}, \lpac{}{}{B} + a_{lp}) < \rho^A(\poolcontentvector{}{}{}, \lpac{}{}{A}, \lpac{}{}{B} + a_{lp}^{\textit{fill}}(\poolcontentvector{}{}{}, \lpac{}{}{A}, \lpac{}{}{B}))}$.
	Denote this difference by ${z = \rho^A(\poolcontentvector{}{}{}, \lpac{}{}{A}, \lpac{}{}{B} + a_{lp}^{\textit{fill}}(\poolcontentvector{}{}{}, \lpac{}{}{A}, \lpac{}{}{B})) - \rho^A(\poolcontentvector{}{}{}, \lpac{}{}{A}, \lpac{}{}{B} + a_{lp})}$.
	We construct the action $a_{lp}'= \left(\left(\lpac{1}{}{A'}, \lpac{1}{}{B'}\right), \cdots, \left(\lpac{n}{}{A'}, \lpac{n}{}{B'}\right)\right)$ as follows. Figure~\ref{fig:action} shows an example of the construction.
	\begin{align*}
		\textit{ for } i = i^{\ast}: \lpac{i}{}{A'} = \fillupac{i}{}{A} - \frac{1}{2}z,  \nonumber \\
		\text{ for } i \neq i^{\ast}: \lpac{i}{}{A'} = \fillupac{i}{}{A} + \frac{1}{2(n-1)}z  \,\, .
	\end{align*}
	Then, $\textit{shard}_{i^{\ast}}$ in $\poolcontentvector{}{}{} + a_{lp}'$ is the only shard with the minimal amount of deposited \textit{token~A}, make it the best choice for the subsequent trader.
	\begin{equation}\label{theorem_always_fillup_proof0}
		\SmallestAction(\uac{}{}{\textit{BA}}, \poolcontentvector{}{}{} + a_{lp}') = \{a_{i^{\ast}}^{\textit{BA}}(\uac{}{}{\textit{BA}})\} \,\,.
	\end{equation}
	From the Definition~\ref{def:fillup_action}, $\rho^A(\poolcontentvector{}{}{} + a_{lp}^{\textit{fill}}) = \lpac{i^{\ast}}{}{A^{\ast}} + \poolcontent{i^{\ast}}{}{A}$.
	Then, the smallest amount of deposited \textit{token~A} in $\poolcontentvector{}{}{} + a_{lp}'$ is larger than that in $\poolcontentvector{}{}{} + a_{lp}$:
	\begin{align}\label{theorem_always_fillup_proof1}
		\lpac{i^{\ast}}{}{A'} + \poolcontent{i^{\ast}}{}{A} 
		&= \fillupac{i^{\ast}}{}{A} - \frac{1}{2}z + \poolcontent{i^{\ast}}{}{A} \nonumber \\
		&= (\fillupac{i^{\ast}}{}{A} + \poolcontent{i^{\ast}}{}{A}) - \frac{1}{2}z \nonumber \\
		&= \rho^A(\poolcontentvector{}{}{} + a_{lp}^{\textit{fill}}) 
		 - \frac{1}{2}z \nonumber \\
		 &> \rho^A(\poolcontentvector{}{}{} + a_{lp}^{\textit{fill}}) 
		 - z \nonumber \\
		&= \rho^A(\poolcontentvector{}{}{} + a_{lp}) \,\,.
	\end{align}
	(2) The revenues of actions $a_{lp}$ and $a_{lp}'$:
	To compare the revenue due to both actions, we consider the strategies and utility of the subsequent trader.
	Any best response of the trader $\tau^{BA}( \poolcontentvector{}{}{}, \uac{}{}{\textit{BA}}, a^{\textit{BA}})$ uses the smallest shard (Lemma~\ref{lemma:optimal_trader_strategy_set}):
	\begin{equation*}
		\sum_{a_i^{\textit{BA}}(\uac{}{}{\textit{BA}}) \in \SmallestAction(\uac{}{}{\textit{BA}}, \poolcontentvector{}{}{})} \tau^{BA}( \poolcontentvector{}{}{}, \uac{}{}{\textit{BA}}, a_i^{\textit{BA}}(\uac{}{}{\textit{BA}})) = 1 \,\,.
	\end{equation*}
	For the shard state $\poolcontentvector{}{}{} + a_{lp}'$, combining Equation~\ref{theorem_always_fillup_proof0}, the probability of taking action $a_{i^{\ast}}^{\textit{BA}}(\uac{}{}{\textit{BA}})$ is 1:
	\begin{align*}
		&\tau^{BA}( \poolcontentvector{}{}{}, \uac{}{}{\textit{BA}}, a_{i^{\ast}}^{\textit{BA}}(\uac{}{}{\textit{BA}})) \nonumber \\
		=& \sum_{a_i^{\textit{BA}}(\uac{}{}{\textit{BA}}) \in \SmallestAction(\uac{}{}{\textit{BA}}, \poolcontentvector{}{}{})} \tau^{BA}( \poolcontentvector{}{}{}, \uac{}{}{\textit{BA}}, a_i^{\textit{BA}}(\uac{}{}{\textit{BA}})) \nonumber \\
		=& 1 \,\,.
	\end{align*}

	Therefore, from Equation~\ref{Ulp_action_strategy_simplified}, the utility of the liquidity provider following action $a_{lp}'$ and trader's best response strategy $\tau^{BA}$ is
	\begin{multline}\label{theorem_always_fillup_proof2}
		U_{lp}(\poolcontentvector{}{}{}, a_{lp}',  \tau^{BA}, \tau^{AB},\lpac{}{}{A}, \lpac{}{}{B})  = \\
		E_{ \uac{}{}{\textit{BA}} \sim D^{BA}} \left[
			U_{lp}(\poolcontentvector{}{}{} , a_{lp}',a_{i^{\ast}}^{\textit{BA}}(\uac{}{}{\textit{BA}}))
		\right] .
	\end{multline}

	Similarly, the utility following action $a_{lp}$ is
	\begin{multline}\label{theorem_always_fillup_proof3}
		U_{lp}(\poolcontentvector{}{}{}, a_{lp},  \tau^{BA}, \tau^{AB},\lpac{}{}{A}, \lpac{}{}{B})  = \\
		E_{ \uac{}{}{\textit{BA}} \sim D^{BA}} \left[
		\sum_{\substack{a_i^{\textit{BA}}(\uac{}{}{\textit{BA}})\in \SmallestAction(\uac{}{}{\textit{BA}}, \poolcontentvector{}{}{})}}
			\left(
				\substack{U_{lp}(\poolcontentvector{}{}{}, a_{lp},a_i^{\textit{BA}}(\uac{}{}{\textit{BA}})) \\ \times \tau^{BA}(\poolcontentvector{}{}{} + a_{lp}, \uac{}{}{\textit{BA}}, a_i^{\textit{BA}})	}			
			\right)
		\right] .
	\end{multline}
	Thus, the revenue following action $a_{lp}'$ is (using Equation~\ref{TFBA}):
	\begin{multline}\label{theorem_always_fillup_proof4}
		U_{lp}(\poolcontentvector{}{}{}, a_{lp}',a_{i^{\ast}}^{\textit{BA}}(\uac{}{}{\textit{BA}})) = \\
		 p^B \times \textit{tf}( \poolcontent{i^{\ast}}{}{A} + \lpac{i^{\ast}}{}{A'}, \frac{p^A}{p^B} (\poolcontent{i^{\ast}}{}{A} + \lpac{i^{\ast}}{}{A'}), \uac{}{}{\textit{BA}}) \times \frac{\lpac{i^{\ast}}{}{A'}}{ \lpac{i^{\ast}}{}{A'} +  \poolcontent{i^{\ast}}{}{A}} \,\,.
	\end{multline}
	Similarly, we can expand the expression of $U_{lp}(\poolcontentvector{}{}{}, a_{lp},a_i^{\textit{BA}}(\uac{}{}{\textit{BA}}))$ for $a_i^{\textit{BA}} \in \SmallestAction(\uac{}{}{\textit{BA}}, \poolcontentvector{}{}{} + a_{lp})$ as
	\begin{multline}\label{theorem_always_fillup_proof5}
		U_{lp}(\poolcontentvector{}{}{}, a_{lp},a_i^{\textit{BA}}(\uac{}{}{\textit{BA}})) = \\
		 p^B \times \textit{tf}( \poolcontent{i}{}{A} + \lpac{i}{}{A}, \frac{p^A}{p^B} (\poolcontent{i}{}{A} + \lpac{i}{}{A}), \uac{}{}{\textit{BA}}) \times \frac{\lpac{i}{}{A}}{ \lpac{i}{}{A} +  \poolcontent{i}{}{A}} \,\,.
	\end{multline}
	Since $a_i^{\textit{BA}} \in \SmallestAction(\uac{}{}{\textit{BA}}, \poolcontentvector{}{}{} + a_{lp})$, $\textit{shard}_i$ has the smallest amount of deposited \textit{token~A} in $\poolcontentvector{}{}{} + a_{lp}$:
	\begin{equation}\label{theorem_always_fillup_proof6}
		\poolcontent{i}{}{A} + \lpac{i}{}{A} = \rho^A(\poolcontentvector{}{}{} + a_{lp}) \,\,.
	\end{equation}
	Interpreting Equation~\ref{theorem_always_fillup_proof1}, the smallest shard in $\poolcontentvector{}{}{} + a_{lp}'$ is larger than the smallest shard in $\poolcontentvector{}{}{} + a_{lp}$:
	\begin{equation}\label{theorem_always_fillup_proof7}
		\lpac{i^{\ast}}{}{A'} + \poolcontent{i^{\ast}}{}{A} > \poolcontent{i}{}{A} + \lpac{i}{}{A} \,\,.
	\end{equation}

	From lemma~\ref{lemma:trading_fee}, the trading fee of trading in a larger shard is larger,
	\begin{multline}\label{theorem_always_fillup_proof8}
		\textit{tf}( \poolcontent{i^{\ast}}{}{A} + \lpac{i^{\ast}}{}{A'}, \frac{p^A}{p^B} (\poolcontent{i^{\ast}}{}{A} + \lpac{i^{\ast}}{}{A'}), \uac{}{}{\textit{BA}}) > \\  \textit{tf}( \poolcontent{i}{}{A} + \lpac{i}{}{A}, \frac{p^A}{p^B} (\poolcontent{i}{}{A} + \lpac{i}{}{A}), \uac{}{}{\textit{BA}}) \,\,.
	\end{multline}
	From the definition of $i^{\ast}$ in Equation~\ref{theorem_always_fillup_proof20}, we have
	\begin{equation}\label{theorem_always_fillup_proof9}
		\poolcontent{i^{\ast}}{}{A} \leq \poolcontent{i}{}{A} \,\,.
	\end{equation}

	Combine Equations~\ref{theorem_always_fillup_proof7} and~\ref{theorem_always_fillup_proof9}, we have
	\begin{equation}\label{theorem_always_fillup_proof11}
		\frac{\lpac{i^{\ast}}{}{A'}}{ \lpac{i^{\ast}}{}{A'} +  \poolcontent{i^{\ast}}{}{A}} = 1 - \frac{ \poolcontent{i^{\ast}}{}{A}}{\lpac{i^{\ast}}{}{A'} +  \poolcontent{i^{\ast}}{}{A}} > 1 - \frac{ \poolcontent{i}{}{A}}{\lpac{i}{}{A} +  \poolcontent{i}{}{A}} = \frac{\lpac{i}{}{A}}{ \lpac{i}{}{A} +  \poolcontent{i}{}{A}} \,\,.
	\end{equation}

	Combining Equations~\ref{theorem_always_fillup_proof4},~\ref{theorem_always_fillup_proof5},~\ref{theorem_always_fillup_proof8} and~\ref{theorem_always_fillup_proof11}, the revenue of the liquidity provider with action $a_{lp}'$ is higher than with action $a_{lp} \neq a_{lp}^{\textit{fill}}(\poolcontentvector{}{}{}, \lpac{}{}{A}, \lpac{}{}{B})$:
	\begin{equation}\label{theorem_always_fillup_proof12}
		U_{lp}(\poolcontentvector{}{}{}, a_{lp}',a_{i^{\ast}}^{\textit{BA}}(\uac{}{}{\textit{BA}})) > U_{lp}(\poolcontentvector{}{}{}, a_{lp},a_i^{\textit{BA}}(\uac{}{}{\textit{BA}})) \,\,.
	\end{equation}
	Combining Equations~\ref{theorem_always_fillup_proof2} and~\ref{theorem_always_fillup_proof3}, the utility of the liquidity provider under action $a_{lp}'$ is higher than that under action $a_{lp} \neq a_{lp}^{\textit{fill}}(\poolcontentvector{}{}{}, \lpac{}{}{A}, \lpac{}{}{B})$:
	\begin{multline*}
		U_{lp}(\poolcontentvector{}{}{}, a_{lp}',  \tau^{BA}, \tau^{AB},\lpac{}{}{A}, \lpac{}{}{B}) > \\ U_{lp}(\poolcontentvector{}{}{}, a_{lp},  \tau^{BA}, \tau^{AB},\lpac{}{}{A}, \lpac{}{}{B}) \,\,.
	\end{multline*}

	When $a_{lp} = a_{lp}^{\textit{fill}}(\poolcontentvector{}{}{}, \lpac{}{}{A}, \lpac{}{}{B})$, we have $a_{lp}' = a_{lp}$.
	Therefore, the following inequality holds for all $a_{lp} \in \mathcal{A}_{lp}(\lpac{}{}{A} , \lpac{}{}{B} )$:
	\begin{equation}\label{theorem_always_fillup_proof13}
		U_{lp}(\poolcontentvector{}{}{}, a_{lp}',a_{i^{\ast}}^{\textit{BA}}(\uac{}{}{\textit{BA}})) \geq U_{lp}(\poolcontentvector{}{}{}, a_{lp},a_i^{\textit{BA}}(\uac{}{}{\textit{BA}})) \,\,.
	\end{equation}

	(3) Construction of a new strategy:
	We showed that for any action $a_{lp} \neq a_{lp}^{\textit{fill}}(\poolcontentvector{}{}{}, \lpac{}{}{A}, \lpac{}{}{B})$, the constructed action $a_{lp}'$ has a higher revenue.
	Now, for any liquidity provider strategy $\pi_{lp}$, we can construct a new strategy $\pi_{lp}'$ where for every action $a_{lp}$, the probability of constructed action $a_{lp}'$ in the new strategy is the same as the original action $a_{lp}$ in the original strategy:
	\begin{equation*}
		\pi_{lp}'( \poolcontentvector{}{}{}, \lpac{}{}{A}, \lpac{}{}{B}, a_{lp}') = \pi_{lp}( \poolcontentvector{}{}{}, \lpac{}{}{A}, \lpac{}{}{B}, a_{lp}) \,\,.
	\end{equation*}

	If the original strategy is different from the fillup strategy $\pi_{lp} \neq \tau^{\textit{fill}}_{lp}$,
	then 
	$\exists \tilde{a}_{lp} \in \mathcal{A}_{lp}(\lpac{}{}{A} , \lpac{}{}{B} )$, s.t. $\tilde{a}_{lp} \neq a_{lp}^{\textit{fill}}(\poolcontentvector{}{}{}, \lpac{}{}{A}, \lpac{}{}{B})$ and $\pi_{lp}( \poolcontentvector{}{}{}, \lpac{}{}{A}, \lpac{}{}{B}, a_{lp}) > 0$.

	(4) Comparison of utilities:
	From the definition (Equation~\ref{Ulp_strategy_strategy}), the utility of the liquidity provider under $\pi_{lp}$ is
	\begin{equation}\label{theorem_always_fillup_proof30}
		U_{lp}(\poolcontentvector{}{}{},\pi_{lp},  \tau^{BA}, \tau^{AB},\lpac{}{}{A}, \lpac{}{}{B}) = 
		\sum_{a_{lp} \in \mathcal{A}_{lp}(\lpac{}{}{A} , \lpac{}{}{B} )}
		\left(
			\substack
			{
				\pi_{lp}( \poolcontentvector{}{}{}, \lpac{}{}{A}, \lpac{}{}{B}, a_{lp}) \\ \times U_{lp}(\poolcontentvector{}{}{}, a_{lp},  \tau^{BA}, \tau^{AB},\lpac{}{}{A}, \lpac{}{}{B})
			}
		\right)
 	\,\,.
	\end{equation}

	Similarly, the utility of the liquidity provider under $\pi_{lp}'$ is
	\begin{align}\label{theorem_always_fillup_proof31}
		&U_{lp}(\poolcontentvector{}{}{}, \pi_{lp}',  \tau^{BA}, \tau^{AB},\lpac{}{}{A}, \lpac{}{}{B}) \nonumber\\ 
		=&\sum_{a_{lp}' \in \mathcal{A}_{lp}(\lpac{}{}{A} , \lpac{}{}{B} )} 
		\left(
		\substack{\pi_{lp}'( \poolcontentvector{}{}{}, \lpac{}{}{A}, \lpac{}{}{B}, a_{lp}') \\ \times U_{lp}(\poolcontentvector{}{}{}, a_{lp}',  \tau^{BA}, \tau^{AB},\lpac{}{}{A}, \lpac{}{}{B})}
		\right)\nonumber \\
		=& \sum_{a_{lp}' \in \mathcal{A}_{lp}(\lpac{}{}{A} , \lpac{}{}{B} ) \setminus \{\tilde{a}_{lp}'\}} 
		\left(
		\substack{\pi_{lp}'( \poolcontentvector{}{}{}, \lpac{}{}{A}, \lpac{}{}{B}, a_{lp}') \\ \times U_{lp}(\poolcontentvector{}{}{}, a_{lp}',  \tau^{BA}, \tau^{AB},\lpac{}{}{A}, \lpac{}{}{B})}
		\right) \nonumber \\
		& + \pi_{lp}'( \poolcontentvector{}{}{}, \lpac{}{}{A}, \lpac{}{}{B}, \tilde{a}_{lp}') \times U_{lp}(\poolcontentvector{}{}{}, \tilde{a}_{lp}',  \tau^{BA}, \tau^{AB},\lpac{}{}{A}, \lpac{}{}{B}) \,\,.
	\end{align}
	Since $\pi_{lp}'( \poolcontentvector{}{}{}, \lpac{}{}{A}, \lpac{}{}{B}, \tilde{a}_{lp}') = \pi_{lp}( \poolcontentvector{}{}{}, \lpac{}{}{A}, \lpac{}{}{B}, \tilde{a}_{lp}) > 0$, from Equation~\ref{theorem_always_fillup_proof12}, we have
	\begin{multline}\label{theorem_always_fillup_proof14}
		\pi_{lp}'( \poolcontentvector{}{}{}, \lpac{}{}{A}, \lpac{}{}{B}, \tilde{a}_{lp}') \times U_{lp}(\poolcontentvector{}{}{}, \tilde{a}_{lp}',  \tau^{BA}, \tau^{AB},\lpac{}{}{A}, \lpac{}{}{B}) > \\
		\pi_{lp}( \poolcontentvector{}{}{}, \lpac{}{}{A}, \lpac{}{}{B}, \tilde{a}_{lp}) \times U_{lp}(\poolcontentvector{}{}{}, \tilde{a}_{lp},  \tau^{BA}, \tau^{AB},\lpac{}{}{A}, \lpac{}{}{B}) \,\,.
	\end{multline}
	Since $\pi_{lp}'( \poolcontentvector{}{}{}, \lpac{}{}{A}, \lpac{}{}{B}, a_{lp}') = \pi_{lp}( \poolcontentvector{}{}{}, \lpac{}{}{A}, \lpac{}{}{B}, {a}_{lp}) > 0$, from Equation~\ref{theorem_always_fillup_proof13}, we have 
	\begin{multline}\label{theorem_always_fillup_proof15}
		\pi_{lp}'( \poolcontentvector{}{}{}, \lpac{}{}{A}, \lpac{}{}{B}, a_{lp}') \times U_{lp}(\poolcontentvector{}{}{}, a_{lp}',  \tau^{BA}, \tau^{AB},\lpac{}{}{A}, \lpac{}{}{B}) \geq \\
		\pi_{lp}( \poolcontentvector{}{}{}, \lpac{}{}{A}, \lpac{}{}{B}, a_{lp}) \times U_{lp}(\poolcontentvector{}{}{}, a_{lp},  \tau^{BA}, \tau^{AB},\lpac{}{}{A}, \lpac{}{}{B}) \,\,.
	\end{multline}
	Combining Equations~\ref{theorem_always_fillup_proof30}
	~\ref{theorem_always_fillup_proof31} ,~\ref{theorem_always_fillup_proof14} and~\ref{theorem_always_fillup_proof15}, the utility of the liquidity provider under $\pi_{lp}'$ is higher than that under $\pi_{lp}$:
	\begin{equation*}
		U_{lp}(\poolcontentvector{}{}{}, \pi_{lp}',  \tau^{BA}, \tau^{AB},\lpac{}{}{A}, \lpac{}{}{B}) >  U_{lp}(\poolcontentvector{}{}{}, \pi_{lp},  \tau^{BA}, \tau^{AB},\lpac{}{}{A}, \lpac{}{}{B}) \,\,.
	\end{equation*}
	Therefore, any liquidity provider strategy $\pi_{lp} \neq \tau^{\textit{fill}}_{lp}$ is not the best response when the trader follows any best response.
	Therefore, in all SPNE, the liquidity provider's best response is the fillup strategy $\tau^{\textit{fill}}_{lp}$.
\end{proof}

\subsubsection{Specific SPNE under deviation}

The above theorem does not prove the existence of an SPNE.
We settle this by finding a specific SPNE in $\Gamma_n(\textit{tf}_{SAMM})$.
In this SPNE, if there are multiple smallest shards, the trader uses the smallest shard in the last step.
If there is more than one smallest shard in the last step, the trader selects the one with the smallest index.
Denote by $i_{\min}(\poolcontentvector{}{}{})$ the index of the shard with the smallest amount of deposited \textit{token~A} in $\poolcontentvector{}{}{}$.
When $i^{\ast} = i_{\min}(\poolcontentvector{}{}{})$, we have
\begin{align}\label{eq:coincidence}
	\forall j, \poolcontent{i^{\ast}}{}{A} \leq \poolcontent{j}{}{A}  \nonumber \,\,;\\
	\forall j, \poolcontent{i^{\ast}}{}{A} = \poolcontent{j}{}{A} \Rightarrow i^{\ast} \leq j \,\,.
\end{align}
If the shard with index $i_{\min}(\poolcontentvector{}{}{})$ is the only shard, the \textit{BA} trader would only trade in that shard.
Then the liquidity provider only taking the fill-up action is the best response strategy:

\begin{restatable}{thm}{theoremfillup}\label{theorem_fill_up}
	In $\Gamma_n(\textit{tf}_{SAMM})$, assume that the shard state in step~$k$ is $\poolcontentvector{}{k}{}$, $i^{\ast} = i_{\min}(\poolcontentvector{}{k}{})$ is the index of the shard with the smallest amount of deposited \textit{token~A} in $\poolcontentvector{}{k}{}$.
	Then the trader strategy to trade in the smallest shard in the last step if it is the smallest one, or randomly select one of the smallest shards, namely,
	\begin{equation}\label{eq:theorem_fill_up_tau_t}
		\tau^{BA}(\poolcontentvector{}{}{}, \uac{}{}{\textit{BA}}, a^{\textit{BA}}) = 
		\begin{cases}
			1, & \text{if } \poolcontent{i^{\ast}}{}{} = \rho^A(\poolcontentvector{}{k}{},\lpac{}{}{A}, \lpac{}{}{B}) \\
			& \text{and }a^{\textit{BA}} = a_{i^{\ast}}^{\textit{BA}}(\uac{}{}{\textit{BA}})
			\\
			\frac{1}{n_{\min}(\poolcontentvector{}{}{})}, & \text{if } \poolcontent{i^{\ast}}{}{} = \rho^A(\poolcontentvector{}{k}{},\lpac{}{}{A}, \lpac{}{}{B}) \\
			& \text{and }a^{\textit{BA}} \in \SmallestAction(\uac{}{}{\textit{BA}}, \poolcontentvector{}{}{}) \\
			0, & \text{otherwise}
		\end{cases}				
	\end{equation}
	and the liquidity provider's fillup strategy:
	\begin{equation*}
		\tau^{\textit{fill}}_{lp}( \poolcontentvector{}{}{}, \lpac{}{}{A}, \lpac{}{}{B}, a_{lp}) = \begin{cases}
			1, & \text{if } a_{lp} = a_{lp}^{\textit{fill}}(\poolcontentvector{}{}{}, \lpac{}{}{A}, \lpac{}{}{B}) \\
			0, & \text{otherwise}
		  \end{cases},
	\end{equation*}
	are an SPNE in step $k$.
\end{restatable}

\begin{proof}[Proof Sketch]
	Since the trader always trades in one of the smallest shards, $\tau^{BA}$ is a dominant strategy for her.
	Therefore, we only need to prove that $\tau_{lp}$ is the best response to $\tau^{BA}$.
	If the liquidity provider does not take the fill-up action, then the trader trades in the smallest shard with a smaller amount of deposited \textit{token~A} than that under the fill-up action.
	Since larger shards have a higher trading fee under a fixed trade, the utility of the liquidity provider is higher when she takes the fill-up action.
\end{proof}
\begin{proof}
	Considering the revenue of the liquidity provider under the action $a_{lp}^{\textit{fill}}(\poolcontentvector{}{}{}, \lpac{}{}{A}, \lpac{}{}{B})$.
	The amount of deposited \textit{token~A} of $\textit{shard}_{i^{\ast}}$ in $\poolcontentvector{}{}{} + a_{lp}^{\textit{fill}}(\poolcontentvector{}{}{}, \lpac{}{}{A}, \lpac{}{}{B})$  is $\poolcontent{i}{}{A} + \fillupac{i^{\ast}}{}{A}$.
	
	We first prove that $\fillupac{i^{\ast}}{}{A} > 0$ by contradiction.
	If $\fillupac{i^{\ast}}{}{A} = 0$, then $\poolcontent{i^{\ast}}{}{A} + \fillupac{i^{\ast}}{}{A} = \poolcontent{i^{\ast}}{}{A}$. Then for any input amount $\fillupac{j}{}{A} > 0$, we have $ \poolcontent{j}{}{A} + \fillupac{j}{}{A} > \poolcontent{j}{}{A}$.
	From the definition of $i^{\ast}= 
	i_{\min}(\poolcontentvector{}{}{})$ (Equation~\ref{eq:coincidence}), we have $\poolcontent{i^{\ast}}{}{A} \leq \poolcontent{j}{}{A}$.
	Therefore, the amount of deposited \textit{token~A} of $\textit{shard}_{i^{\ast}}$ in $\poolcontentvector{}{}{} + a_{lp}^{\textit{fill}}(\poolcontentvector{}{}{}, \lpac{}{}{A}, \lpac{}{}{B})$ is strictly smaller than $\textit{shard}_j$
	\begin{equation*}
		 \poolcontent{j}{}{A} + \fillupac{j}{}{A}  > \poolcontent{j}{}{A} \geq \poolcontent{i^{\ast}}{}{A} = \poolcontent{j}{}{A} + \fillupac{i^{\ast}}{}{A} \,\,.
	\end{equation*}
	This contradicts the definition of $a_{lp}^{\textit{fill}}(\poolcontentvector{}{}{}, \lpac{}{}{A}, \lpac{}{}{B})$ (Definition~\ref{def:fillup_action}) since $\fillupac{j}{}{A} >0$. 
	Therefore, $\fillupac{i^{\ast}}{}{A} > 0$.
	Then also from that definition, the amount of deposited \textit{token~A} of $\textit{shard}_{i^{\ast}}$ after the fillup action is smallest among all shards:
	\begin{equation}\label{theorem_fill_up_proof0}
		\poolcontent{i^{\ast}}{}{A} + \fillupac{i^{\ast}}{}{A} = \rho^A(\poolcontentvector{}{}{} + a_{lp}^{\textit{fill}}(\poolcontentvector{}{}{}, \lpac{}{}{A}, \lpac{}{}{B})) \,\,.
	\end{equation}
	Therefore, from the definition of $\tau^{BA}$ (Equation~\ref{eq:theorem_fill_up_tau_t}), we have $\tau^{BA}(\poolcontentvector{}{}{} + a_{lp}^{\textit{fill}}(\poolcontentvector{}{}{}, \lpac{}{}{A}, \lpac{}{}{B}), \uac{}{}{\textit{BA}}, a_{i^{\ast}}^{\textit{BA}}) = 1$.

	Then, we turn to the revenue of the liquidity provider with the fillup action and prove that it is no smaller than any other action.
	From Equation~\ref{Ulp_action_strategy_simplified}, the utility of the liquidity provider under the fill-up action and the trader's best response $\tau^{BA}$ is:
	\begin{multline}\label{theorem_fill_up_proof1}
		U_{lp}(\poolcontentvector{}{}{}, a_{lp}^{\textit{fill}}(\poolcontentvector{}{}{}, \lpac{}{}{A}, \lpac{}{}{B}),  \tau^{BA}, \tau^{AB},\lpac{}{}{A}, \lpac{}{}{B}) = \\
		E_{ \uac{}{}{\textit{BA}} \sim D^{BA}} \left[
			U_{lp}(\poolcontentvector{}{}{}, a_{lp}^{\textit{fill}}(\poolcontentvector{}{}{}, \lpac{}{}{A}, \lpac{}{}{B}),a_{i^{\ast}}^{\textit{BA}}(\uac{}{}{\textit{BA}}))
		\right] \,\,.
	\end{multline}
	Similarly, the utility of the liquidity provider under any action $a_{lp} \in \mathcal{A}_{lp}(\lpac{}{}{A}, \lpac{}{}{B} )$ and the trader's best response $\tau^{BA}$ is
	\begin{equation}\label{theorem_fill_up_proof2}
		U_{lp}(\poolcontentvector{}{}{}, a_{lp},  \tau^{BA}, \tau^{AB},\lpac{}{}{A}, \lpac{}{}{B}) = 
		E_{ \uac{}{}{\textit{BA}} \sim D^{BA}} \left[
		\sum_{\substack{a_i^{\textit{BA}}(\uac{}{}{\textit{BA}})\\\in \SmallestAction(\uac{}{}{\textit{BA}}, \poolcontentvector{}{}{})}}
			\left(
				\substack{U_{lp}(\poolcontentvector{}{}{}, a_{lp},a_i^{\textit{BA}}(\uac{}{}{\textit{BA}}))  \times\\ \tau^{BA}(\poolcontentvector{}{}{} + a_{lp}, \uac{}{}{\textit{BA}}, a_i^{\textit{BA}})	}			
			\right)
		\right] \,\,.
	\end{equation}
	From the definition of $U_{lp}( \poolcontentvector{}{}{}, a_{lp}, a^{\textit{BA}})$ (Equation~\ref{TFBA}), we have
	\begin{multline}\label{theorem_fill_up_proof3}
		U_{lp}(\poolcontentvector{}{}{}, a_{lp}^{\textit{fill}}(\poolcontentvector{}{}{}, \lpac{}{}{A}, \lpac{}{}{B}),a_{i^{\ast}}^{\textit{BA}}(\uac{}{}{\textit{BA}})) = \\
		 p^B \times \textit{tf}( \poolcontent{i^{\ast}}{}{A} + \fillupac{i^{\ast}}{}{A}, \frac{p^A}{p^B} (\poolcontent{i^{\ast}}{}{A} + \fillupac{i^{\ast}}{}{A}), \uac{}{}{\textit{BA}}) \times \frac{\fillupac{i^{\ast}}{}{A}}{ \fillupac{i^{\ast}}{}{A} +  \poolcontent{i^{\ast}}{}{A}} \,\,,
	\end{multline}
	and
	\begin{equation}\label{theorem_fill_up_proof4}
		U_{lp}(\poolcontentvector{}{}{}, a_{lp},a_i^{\textit{BA}}(\uac{}{}{\textit{BA}})) = 
		 p^B \times \textit{tf}( \poolcontent{i}{}{A} + \lpac{i}{}{A}, \frac{p^A}{p^B} (\poolcontent{i}{}{A} + \lpac{i}{}{A}), \uac{}{}{\textit{BA}}) \times \frac{\lpac{i}{}{A}}{ \lpac{i}{}{A} +  \poolcontent{i}{}{A}} \,\,.
	\end{equation}
	Since $a_i^{\textit{BA}} \in \SmallestAction(\uac{}{}{\textit{BA}}, \poolcontentvector{}{}{})$, $\textit{shard}_i$ has the smallest amount of deposited \textit{token~A} in $\poolcontentvector{}{}{} + a_{lp}$:
	\begin{equation}\label{theorem_fill_up_proof5}
		\poolcontent{i}{}{A} + \lpac{i}{}{A} = \rho^A(\poolcontentvector{}{}{} + a_{lp}) \,\,.
	\end{equation}
	From Lemma~\ref{lemma:unique_fillup}, the minimal amount of deposited \textit{token~A} in $\poolcontentvector{}{}{} + a_{lp}^{\textit{fill}}(\poolcontentvector{}{}{}, \lpac{}{}{A}, \lpac{}{}{B})$ is no less than that in  $\poolcontentvector{}{}{} + a_{lp}$ for any $a_{lp} \in \mathcal{A}_{lp}(\lpac{}{}{A} , \lpac{}{}{B} )$:
	\begin{equation}\label{theorem_fill_up_proof6}
		\rho^A(\poolcontentvector{}{}{} + a_{lp}^{\textit{fill}}(\poolcontentvector{}{}{}, \lpac{}{}{A}, \lpac{}{}{B})) \geq \rho^A(\poolcontentvector{}{}{} + a_{lp}) \,\,.
	\end{equation}
	Combining Equations~\ref{theorem_fill_up_proof0},~\ref{theorem_fill_up_proof5} and~\ref{theorem_fill_up_proof6}, we have 
	\begin{equation}\label{theorem_fill_up_proof7}
		\poolcontent{i^{\ast}}{}{A} + \fillupac{i^{\ast}}{}{A} \geq \poolcontent{i}{}{A} + \lpac{i}{}{A} \,\,.
	\end{equation}
	From Lemma~\ref{lemma:trading_fee}, the trading fee of trading in a larger shard is larger:
	\begin{equation}\label{theorem_fill_up_proof8}
		\textit{tf}( \poolcontent{i^{\ast}}{}{A} + \fillupac{i^{\ast}}{}{A}, \frac{p^A}{p^B} (\poolcontent{i^{\ast}}{}{A} + \fillupac{i^{\ast}}{}{A}), \uac{}{}{\textit{BA}}) > \textit{tf}( \poolcontent{i}{}{A} + \lpac{i}{}{A}, \frac{p^A}{p^B} (\poolcontent{i}{}{A} + \lpac{i}{}{A}), \uac{}{}{\textit{BA}}) \,\,.
	\end{equation}
	From the definition of $i^{\ast}= 
	i_{\min}(\poolcontentvector{}{}{})$ (Equation~\ref{eq:coincidence}), we have $\poolcontent{i^{\ast}}{}{A} \leq \poolcontent{i}{}{A}$.
	Combining with Equation~\ref{theorem_fill_up_proof7}, we have
	\begin{equation}\label{theorem_fill_up_proof9}
	\frac{ \poolcontent{i^{\ast}}{}{A}}{\fillupac{i^{\ast}}{}{A} +  \poolcontent{i^{\ast}}{}{A}} \leq \frac{ \poolcontent{i}{}{A}}{\lpac{i}{}{A} +  \poolcontent{i}{}{A}} \,\,.
	\end{equation}

	Therefore, the liquidity provider has more share in $\textit{shard}_i$ in $\poolcontentvector{}{}{} + a_{lp}^{\textit{fill}}(\poolcontentvector{}{}{}, \lpac{}{}{A}, \lpac{}{}{B})$ than any smallest shard in $\poolcontentvector{}{}{} + a_{lp}$:
	\begin{equation}\label{theorem_fill_up_proof10}
		\frac{\fillupac{i^{\ast}}{}{A}}{ \fillupac{i^{\ast}}{}{A} +  \poolcontent{i^{\ast}}{}{A}}  = 1 - \frac{ \poolcontent{i^{\ast}}{}{A}}{\fillupac{i^{\ast}}{}{A} +  \poolcontent{i^{\ast}}{}{A}} \geq 1 - \frac{ \poolcontent{i}{}{A}}{\lpac{i}{}{A} +  \poolcontent{i}{}{A}} = \frac{\lpac{i}{}{A}}{ \lpac{i}{}{A} +  \poolcontent{i}{}{A}} \,\,.
	\end{equation}
	Combining Equations~\ref{theorem_fill_up_proof8} and~\ref{theorem_fill_up_proof10}, we have 
	\begin{multline}\label{theorem_fill_up_proof11}
		p^B \times \textit{tf}( \poolcontent{i^{\ast}}{}{A} + \fillupac{i^{\ast}}{}{A}, \frac{p^A}{p^B} (\poolcontent{i^{\ast}}{}{A} + \fillupac{i^{\ast}}{}{A}), \uac{}{}{\textit{BA}}) \times \frac{\fillupac{i^{\ast}}{}{A}}{ \fillupac{i^{\ast}}{}{A} +  \poolcontent{i^{\ast}}{}{A}} \geq \\
		p^B \times \textit{tf}( \poolcontent{i}{}{A} + \lpac{i}{}{A}, \frac{p^A}{p^B} (\poolcontent{i}{}{A} + \lpac{i}{}{A}), \uac{}{}{\textit{BA}}) \times \frac{\lpac{i}{}{A}}{ \lpac{i}{}{A} +  \poolcontent{i}{}{A}} \,\,.
	\end{multline}
	Then, the revenue of the liquidity provider with the fill-up action is no less than any other action given the trader's best response $\tau^{BA}$:
	\begin{align*}
		&U_{lp}(\poolcontentvector{}{}{}, a_{lp}^{\textit{fill}}(\poolcontentvector{}{}{}, \lpac{}{}{A}, \lpac{}{}{B}),a_{i^{\ast}}^{\textit{BA}}(\uac{}{}{\textit{BA}})) \nonumber \\
		\stackrel{(\ref{theorem_fill_up_proof3})}{=}& p^B \times \textit{tf}( \poolcontent{i^{\ast}}{}{A} + \fillupac{i^{\ast}}{}{A}, \frac{p^A}{p^B} (\poolcontent{i^{\ast}}{}{A} + \fillupac{i^{\ast}}{}{A}), \uac{}{}{\textit{BA}}) \times \frac{\fillupac{i^{\ast}}{}{A}}{ \fillupac{i^{\ast}}{}{A} +  \poolcontent{i^{\ast}}{}{A}} \nonumber \\
		\stackrel{(\ref{theorem_fill_up_proof11})}{\geq}& p^B \times \textit{tf}( \poolcontent{i}{}{A} + \lpac{i}{}{A}, \frac{p^A}{p^B} (\poolcontent{i}{}{A} + \lpac{i}{}{A}), \uac{}{}{\textit{BA}}) \times \frac{\lpac{i}{}{A}}{ \lpac{i}{}{A} +  \poolcontent{i}{}{A}} \nonumber \\
		\stackrel{(\ref{theorem_fill_up_proof4})}{=} &U_{lp}(\poolcontentvector{}{}{}, a_{lp},a_i^{\textit{BA}}(\uac{}{}{\textit{BA}}))  \,\,.
	\end{align*}
	Therefore, from Equations~\ref{theorem_fill_up_proof1} and~\ref{theorem_fill_up_proof2}, the revenue of the liquidity provider with the fill-up action is no less than any other action given the trader's best response $\tau^{BA}$:
	\begin{equation}\label{theorem_fill_up_proof12}
		U_{lp}(\poolcontentvector{}{}{}, a_{lp}^{\textit{fill}}(\poolcontentvector{}{}{}, \lpac{}{}{A}, \lpac{}{}{B}),  \tau^{BA}, \tau^{AB},\lpac{}{}{A}, \lpac{}{}{B}) \geq  U_{lp}(\poolcontentvector{}{}{}, a_{lp},  \tau^{BA}, \tau^{AB},\lpac{}{}{A}, \lpac{}{}{B}) \,\,.
	\end{equation}
	After analyzing actions, we consider the utility of the liquidity provider with strategy $\tau_{lp}$ and any mixed strategy $\pi_{lp}$.
	From the definition of the utility of the liquidity provider under $\pi_{lp}$ (Equation~\ref{Ulp_strategy_strategy}), we have
	\begin{align}\label{theorem_fill_up_proof13}
		&U_{lp}(\poolcontentvector{}{}{}, \tau_{lp},  \tau^{BA}, \tau^{AB},\lpac{}{}{A}, \lpac{}{}{B})  \nonumber\\
		=&
		\sum_{a_{lp} \in \mathcal{A}_{lp}(\lpac{}{}{A} , \lpac{}{}{B} )} 
		\left(
		\substack
		{
			\tau_{lp}( \poolcontentvector{}{}{}, \lpac{}{}{A}, \lpac{}{}{B}, a_{lp}) \\ \times U_{lp}(\poolcontentvector{}{}{}, a_{lp},  \tau^{BA}, \tau^{AB},\lpac{}{}{A}, \lpac{}{}{B}) 
		}
		\right)
		\nonumber \\
		=& U_{lp}(\poolcontentvector{}{}{}, a_{lp}^{\textit{fill}}(\poolcontentvector{}{}{}, \lpac{}{}{A}, \lpac{}{}{B}),a_{i^{\ast}}^{\textit{BA}}(\uac{}{}{\textit{BA}})),
	\end{align}
	and 
	\begin{align}\label{theorem_fill_up_proof14}
		&U_{lp}(\poolcontentvector{}{}{}, \pi_{lp},  \tau^{BA}, \tau^{AB},\lpac{}{}{A}, \lpac{}{}{B}) \nonumber\\
		=&\sum_{a_{lp} \in \mathcal{A}_{lp}(\lpac{}{}{A} , \lpac{}{}{B} )} 
		\left(
			\substack
			{
				\pi_{lp}( \poolcontentvector{}{}{}, \lpac{}{}{A}, \lpac{}{}{B}, a_{lp}) \\ \times U_{lp}(\poolcontentvector{}{}{}, a_{lp},  \tau^{BA}, \tau^{AB},\lpac{}{}{A}, \lpac{}{}{B})
			}
		\right)	
		\nonumber \\
		\stackrel{(\ref{theorem_fill_up_proof12})}{\leq} & \sum_{a_{lp} \in \mathcal{A}_{lp}(\lpac{}{}{A} , \lpac{}{}{B} )} 
		\left(
			\substack{\pi_{lp}( \poolcontentvector{}{}{}, \lpac{}{}{A}, \lpac{}{}{B}, a_{lp}) \\ \times U_{lp}(\poolcontentvector{}{}{}, a_{lp}^{\textit{fill}}(\poolcontentvector{}{}{}, \lpac{}{}{A}, \lpac{}{}{B}),  \tau^{BA}, \tau^{AB},\lpac{}{}{A}, \lpac{}{}{B})}
		\right)
		\nonumber \\
		= & U_{lp}(\poolcontentvector{}{}{}, a_{lp}^{\textit{fill}}(\poolcontentvector{}{}{}, \lpac{}{}{A}, \lpac{}{}{B}),a_{i^{\ast}}^{\textit{BA}}(\uac{}{}{\textit{BA}})) \nonumber \\
		&	\times \sum_{a_{lp} \in \mathcal{A}_{lp}(\lpac{}{}{A} , \lpac{}{}{B} )} \tau_{lp}( \poolcontentvector{}{}{}, \lpac{}{}{A}, \lpac{}{}{B}, a_{lp}) \nonumber\\
		=& U_{lp}(\poolcontentvector{}{}{}, \tau_{lp},  \tau^{BA}, \tau^{AB},\lpac{}{}{A}, \lpac{}{}{B}) \,\,.
	\end{align}
	Therefore, the liquidity provider strategy $\tau_{lp}$ is the best response to the trader strategy $\tau^{BA}$:
	\begin{equation}\label{theorem_fill_up_proof15}
		U_{lp}(\poolcontentvector{}{}{}, \tau_{lp},  \tau^{BA}, \tau^{AB},\lpac{}{}{A}, \lpac{}{}{B}) \geq  U_{lp}(\poolcontentvector{}{}{}, \pi_{lp},  \tau^{BA}, \tau^{AB},\lpac{}{}{A}, \lpac{}{}{B}) \,\,.
	\end{equation}
	Since $\tau^{BA}$ is also a best response to the trader strategy $\tau_{lp}$, the liquidity provider strategy $\tau_{lp}$ and the trader strategy $\tau^{BA}$ are an SPNE.
\end{proof}

In summary, we showed that the system achieves a stable state where the perfect parallelism strategy is the dominant strategy for traders. 
Moreover, the system demonstrates robustness against deviations, particularly due to attacks.

\section{Solana Experiment Details}\label{app:solana}

We evaluate the performance of SAMM on Solana under a setup (\S\ref{app:solana:setup}) analogous to our Sui experiment~(\S\ref{sec:evaluation}).
The results (\S\ref{app:solana:result}) show that SAMM's throughput on Solana scales linearly with the number of shards, but up to 16 shards.

\subsection{Setup}\label{app:solana:setup}
Each smart contract in Solana has an associated account, allowing transactions using different accounts to be executed in parallel.
We implement SAMM\footnote{\url{https://github.com/MountainGold/spl-samm/tree/main/token-swap}} in Rust and deploy it on a Solana testnet with one validator.
The average latency of transactions is around $0.3$ seconds with low demand. 
We issue transactions using Solana's JavaScript RPC interface. 

Solana's performance is artificially limited~\cite{solanagas2024}:
Solana's mainnet and testnet limit the gas used by a single account within a block, and the total gas used in a block to four times the single-account gas limit.
We conducted experiments using these standard settings, marked \emph{Solana-Default}. 
In addition, to evaluate the architecture's limits, we conducted experiments where we compiled Solana with the gas limit for the whole block removed, marked \emph{Solana-NoBlockLimit}.
We note that removing the gas limit for a single account results in unstable performance, with no useful results.

As with Sui, we tested the throughput of simple token transfers on Solana as a baseline, reaching a throughput above $2500\textit{tps}$ with latency under one second (\S\ref{app:solana}).
This is higher than the maximum observed in our AMM experiments. 

\begin{figure}[t]
	\centering
	\includegraphics[width=0.8\textwidth]{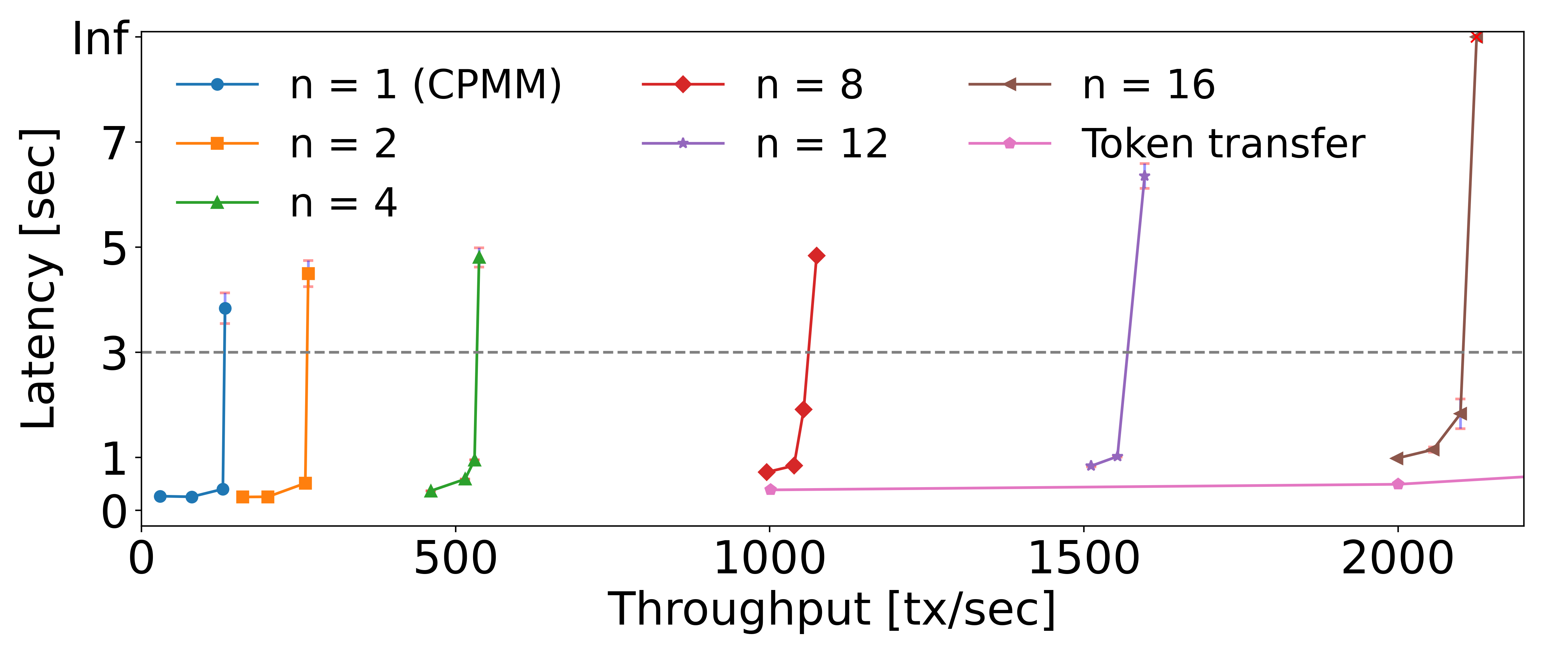}
	\caption{Solana trade latency as a function of demand with~$n$ SAMM shards.}
	\label{fig:multiple_latency_solana}
\end{figure} 

\subsection{Experimental Results}\label{app:solana:result}

Figure~\ref{fig:multiple_latency_solana} details the results of the Solana experiments.
We only test each frequency twice since the results are stable.
For each throughput value (X~axis) we calculate the average latency (Y~axis).
Error bars show all measured values.

For a single CPMM contract ($n=1$), similar to the OmniSwap experiment, the average latency increases gradually with the transaction frequency up to~$129\textit{tps}$ before crossing the 3-second line.
With higher frequencies, transactions frequently fail. 
Solana-Default and Solana-NoBlockLimit are not distinguishable since they have the same gas limit for a single contract.
The latency for a single SAMM contract, whether on Solana-Default or Solana-NoBlockLimit, is comparable to that of a single CPMM contract.

With $n \leq 4$ contracts, the latency in Solana-Default and Solana-NoBlockLimit is indistinguishable. 
However, for $n \geq 5$, Solana-Default does not improve due to the gas limit. 
Therefore, we only present the results for Solana-NoBlockLimit.
With $n=16$, beyond $2099\textit{tps}$, the throughput starts declining when increasing the expected frequency.
We therefore consider the latency to be "Inf" above $2099$, marking it with~$\times$.

We also tested the throughput of simple token transfers on Solana as a baseline, reaching a throughput above $2500\textit{tps}$ (outside the figure range) with latency under one second, higher than the maximum observed in our AMM experiments.

To analyze the system's capacity, we again set a latency cap of $3$ seconds to determine the throughput.
Figure~\ref{fig:multiple_tps_solana} shows that Solana-NoBlockLimit achieves a linear throughput increase up to~16 shards, while Solana-Default scales linearly up to~4 shards, reflecting the unencumbered architecture's capacity and gas limit, respectively.
Beyond 4 shards, Solana-Default's throughput remains flat even up to 32 shards (beyond the range), whereas Solana-NoBlockLimit's throughput declines after 16 shards, with 17 shards failing to reach the maximum \textit{tps} of 16 shards, marked with $\times$ in the figure.

\section{Evaluation of Increasing Parallel Component}\label{app:evaluation_heavier}
\begin{figure}[t]
	\centering
	\includegraphics[width=0.8\textwidth]{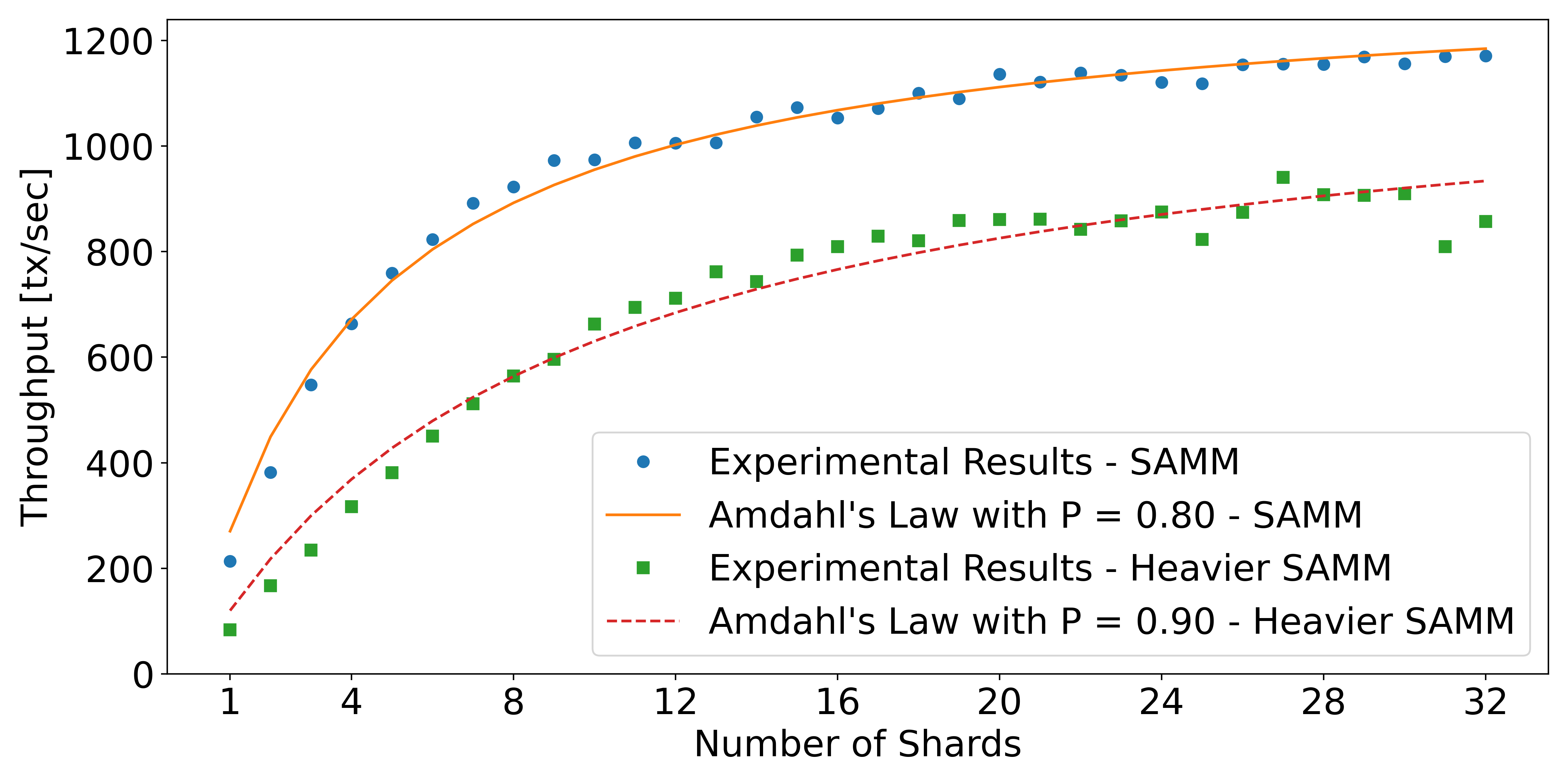}
	\caption{Maximal throughput as a function of the number of shards in SAMM / heavier SAMM.}
	\label{fig:heavier_samm}
\end{figure}

We posit that enhancing the performance of SAMM necessitates mitigating the serial bottlenecks within the platform. 
While modifications to the core architecture of Sui are beyond the scope of this study, we enhance the parallelizable aspects of SAMM by introducing superfluous operations into each trade. 
This methodology follows the experimental setup described in Section~\ref{sec:evaluation}.

Figure~\ref{fig:heavier_samm} presents the maximum throughput results for standard SAMM transactions as previously discussed in Section 8, alongside results from transactions with added operations with three repetitions. 
Although the inclusion of additional operations reduces overall performance due to increased overhead, it significantly enhances the parallel component, with $P=0.9$ ($R^2=0.968$). 
Remarkably, the throughput with $32$ shards is ten times that of a single shard, a substantial improvement over the ratios observed in Section~\ref{sec:evaluation}.

Consequently, addressing the serial bottlenecks within blockchain platforms is vital for future improvements in SAMM performance.

\section{Simulation Setup}\label{app:simulation_setup}
We simulate SAMM with workload from public CPMM trading data, namely the Uniswap~v2 Ethereum AMM for the pair of tokens USDC and ETH (hereinafter the \emph{workload}), from block~\num{12000000} to~\num{19500000} (from 2021-03-08 to 2024-03-23, about 3 years). 
We conservatively choose the throughput of SAMM according to the performance evaluation of Sui since its throughput improves sublinearly with the number of shards.
We simulate Uniswap v2 and SAMM using~1 to~32 shards.
We test three different values of~$c$, namely~$0.003$, $0.005$, and~$0.01$.
We set $r_{\max} = 5 \times r_{\min}$
We choose $r_{\max}, r_{\min}$ and $\beta_1$ (See Section~\ref{sec:SAMM_parameter}) to minimize the maximal trading fee ratio.
This is to limit the maximal cost for traders.
Hence, we set $r_{\max} = 5 \times r_{\min}$ to minimize them at the same time.
Together with the restrictions of satisfying $c$-smaller-better and $c$-non-splitting (Corollary~\ref{corollary:parameter}), the optimization problem is

\begin{equation*}
	\begin{aligned}
		& \underset{r_{\max}, r_{\min}, \beta_1}{\min}
		& & r_{\max} \\
		& \text{subject to}
		& & r_{\max} = 5 \times r_{\min} \,\,, \\
		& & & c \leq 1 - (-\beta_1)^{-\frac{1}{3}} \,\,, \\
		& & & c \leq \frac{r_{\max} - r_{\min}}{-\beta_1} \,\,.
	\end{aligned}
\end{equation*}
Hence we get $r_{\max}, r_{\min}$ and $\beta_1$ given $c$.

In the simulation, we do not add synthetic arbitrage transactions.
There are some trades that have a larger than $c$ ratio of the shard size, especially with large shard numbers, when the shard sizes are small.
We still do not consider the gas fee in the simulation, which would only make the results better since it discourages trade splitting (See Section~\ref{GT:game_model}).

We run each simulation instance as follows.
We uniformly at random select a point in Uniswap v2's history (before~\num{16500000}, to ensure there are enough trades) as the starting point.
For SAMM with~$n$ shards, we evenly distribute the liquidity of the Uniswap at that point among the~$n$ shards to establish the initial state of SAMM. 
We then simulate the real trades from that time point onward. 
Each trade in the real data is simulated as a trading demand with the required amount of tokens matching the output amount and tokens of the actual trade. 
To minimize costs, the trader selects the shard offering the lowest price and may split the trade into several smaller trades if this reduces her costs.
We assume both the Uniswap and SAMM operate at maximal throughput measured in our performance evaluation. 
We count trade splits, the number of transactions in different SAMM shards, liquidity provider revenue, and trader costs over 1 second, following a 1-second warm-up period.
We repeat each simulation 100 times and calculate averages.
Note that a $1$-second simulation corresponds to several hours' trading in the real world.

\section{Transaction Distribution in Simulation}\label{app:simulation_distribution}

\begin{figure}[t]
	\centering
	\includegraphics[width=0.7\textwidth]{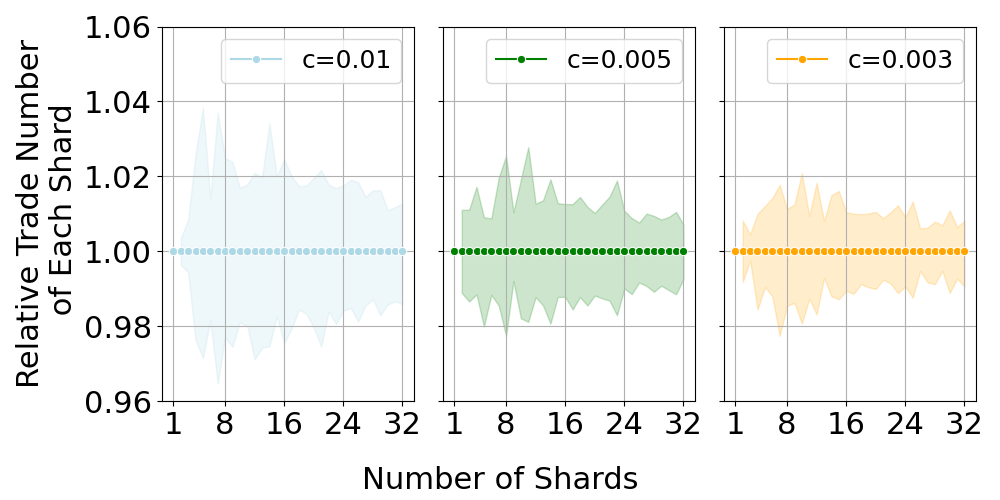}
	\caption{Distribution of transactions in SAMM}
	\label{fig:balanced}
\end{figure}

Our simulation (\S\ref{sec:simulation:setup}) confirms that the distribution of trades in SAMM is balanced.
Figure~\ref{fig:balanced} illustrates the number of trades executed in each shard compared to the average.
Error bars show the range of relative trade numbers in each shard compared to the average.
The difference from the average is always under~$5\%$.

\section{Revenue of Participants}\label{app:revenue}
The increased slippage due to sharded liquidity incurs higher net cost for traders.
This additional cost could be offset by reducing trading fees (Equation~\ref{cost_function}).
In a single trade, the trader, liquidity provider, and arbitrageur effectively form a small zero-sum game: As arbitrageurs gain more revenue from larger slippage, the revenue of traders and liquidity providers decreases.

However, the system's throughput increases due to parallelism, making the full game not zero-sum.
By reducing trading fees in a single trade, it is possible for the trader's cost to be lower than in a scenario with a single AMM where slippage is smaller.
Nonetheless, the revenue of liquidity providers, which is the total trading fees across all trades, may still increase due to the higher throughput.

To validate this, we analyze the revenue of liquidity providers and the costs incurred by traders in our previous simulation (\S\ref{sec:simulation:setup}) and compare these results with the actual revenue and costs observed in historical data.

Here, we formally define the cost of traders and the revenue of liquidity providers in a single trade.
Assume that the trader wants to get $\pooloutput{}{}{A}$ amount of \textit{token~A} and a shard has token reserves $\poolcontent{}{}{A}$ and $\poolcontent{}{}{B}$.
The trader's cost is the token price of \textit{token~B}, $p^B$, times the gross amount of \textit{token~B} they need to pay (Equation~\ref{cost_function}):
\begin{equation*}
	C_{trader} = p^B \times G(\poolcontent{}{}{A},\poolcontent{}{}{B}, \pooloutput{}{}{A}) \,\,.
\end{equation*}

The liquidity provider's revenue is the trading fee collected from the trade in terms of \textit{token~B} times the price of \textit{token~B}:
\begin{equation*}
	R_{lp} = p^B \times \textit{tf}(\poolcontent{}{}{A}, \poolcontent{}{}{B}, \pooloutput{}{}{A}) \,\,.
\end{equation*}

The cost and revenue for a trade to get $\pooloutput{}{}{B}$ amount of \textit{token~B} is symmetric.
We measure the average cost for traders and the total revenue for liquidity providers across all trades in the simulation.

We get historical prices of ETH to USDC from the Yahoo! Finance webpage~\cite{yahoo2024} and use it as the external market price (recall that the simulation is for this token pair).
We calculate the revenue for liquidity providers and costs for traders in USDC, converting ETH amounts at the real-time price.

Figure~\ref{fig:tradercost} shows the average ratio of traders' costs in SAMM compared to the costs in the external market, considering different numbers of shards and varying values of $c$. 
We observe that with larger $c$, the cost for traders increases, as expected.
With a fixed $c$, the cost for traders decreases as the number of shards increases, aligning with SAMM's $c$-smaller-better property, where more shards result in less liquidity in each shard. 
Additionally, we compare the cost for traders in Uniswap, depicted by the black line. 
In all cases, the cost of SAMM is either smaller ($c=0.003$ with at least 2 shards) or slightly larger than that of Uniswap, differing by less than $1\%$ with $c=0.01$ and $0.3\%$ with $c=0.005$ and at least 2 shards.
We aggregate trading fees across all shards to evaluate the revenue of liquidity providers. 
Figure~\ref{fig:lprevenue} shows SAMM consistently generates higher revenue than Uniswap, with $c=0.01$ and 7 shards yielding over 15 times the revenue. 
Initially, revenue in SAMM rises with more shards due to increased throughput but with even more shards it declines as trading fees per trade decrease in smaller shards (See \S\ref{app:gtAnalysis}, Lemma~\ref{lemma:trading_fee}).

SAMM significantly outperforms Uniswap in terms of liquidity provider revenue while maintaining comparable costs for traders. 
This increase without major cost hikes demonstrates this is not a zero-sum game; SAMM's higher throughput allows for more trades and more total trading fees.

.

\section{Volume Capacity of SAMM}\label{app:volume_capacity}

The \emph{volume capacity} of an AMM refers to its ability to facilitate trades without causing the price to deviate beyond an acceptable threshold from the pre-trade spot price. 
To evaluate the volume capacity of SAMM, we use the same real-world trading dataset as in our simulation (\S\ref{sec:simulation:setup}). 
For each historical trade, we consider the pool's token reserves immediately before the trade and the trade's desired output amount. 
We then calculate the effective execution price for this trade in two scenarios: first, on Uniswap v2 (using the price from the trade), and second, on SAMM. 
For the SAMM analysis, we test a range of configurations, varying the number of shards from 1 to 32 (splitting reserved tokens across shards equally) and the parameter $c$ with values of $0.003, 0.005$, and $0.01$. 
SAMM allows the trader's swap to be optimally split across shards to minimize the overall trade cost. 
Note that, unlike Section~\ref{sec:simulation}, this is a static analysis. 
Each trade is evaluated independently against its historical pre-trade state, allowing for a direct comparison of each AMM's price impact without the confounding effects of a sequential simulation.

To quantify the volume capacity, we calculate the ratio of the additional trader cost to the cost derived by reported price (\S\ref{sec:preliminaries:trade}) before the trade.
We call this ratio the \emph{trade cost increase}.
We then measure the proportion of trades whose trade cost increase exceeds predefined thresholds of $0.5\%$, $1\%$, and $1.5\%$.

We observe that, for a given $c$, the proportion of trades exceeding these thresholds is identical for any number of shards. 
This is expected:
First, equally splitting a trade across all shards in a multi-shard pool results in the exact same trade cost as executing the entire trade in a single-shard pool. 
This is because the trading fee ratio and net price remains constant if the trade volume scales linearly in the shard volume. 
Therefore, for the same trade, the optimal cost of multi-shard trade is always no more than the cost of the single-shard trade.
Second, although the $c$-smaller-better and $c$-non-splitting properties often make it more cost-effective to trade within a single shard of a multi-shard pool, a different situation applies to the trades that fail our cost thresholds. 
These trades are typically large relative to the pool's liquidity, a condition under which splitting the trade across multiple shards becomes the a better strategy. 
However, even with the cost reduction from this optimal splitting, our analysis shows the improvement is insufficient to move a failing trade below the threshold. 
Consequently, we do not observe any trade that passes the cost threshold in a multi-shard pool but would have failed in a single-shard pool.
Given this invariance, we present the results in Table~\ref{tab:trade_cost_percentage} without specifying the number of shards.

\begin{table*}[t] 
\centering
\begin{tabular}{lrrr}
\toprule
\textbf{Tested AMM} & 
\begin{tabular}{@{}c@{}}\small\bfseries \% of Trades \\ Exceeding \small\bfseries 0.5\% \\ Trade Cost\end{tabular} & 
\begin{tabular}{@{}c@{}}\small\bfseries \% of Trades \\ Exceeding  \small\bfseries 1\%  \\ Trade Cost\end{tabular} & 
\begin{tabular}{@{}c@{}}\small\bfseries \% of Trades \\ Exceeding  \small\bfseries 1.5\% \\ Trade Cost\end{tabular} \\
\midrule
Uniswap v2     & 0.44\%  & 0.041\%  & 0.019\%  \\
\addlinespace
SAMM ($c=0.003$) & 0.079\% & 0.015\%  & 0.0067\% \\
SAMM ($c=0.005$) & 100\%   & 0.017\%  & 0.0071\% \\
SAMM ($c=0.01$)  & 100\%   & 100\%    & 0.0082\% \\
\bottomrule
\end{tabular}
\caption{Percentage of trades exceeding specific trade cost thresholds for SAMM vs. Uniswap v2.}
\label{tab:trade_cost_percentage} 
\end{table*}

The results reveal several key insights. 
First, SAMM's performance is highly sensitive to the parameter $c$. 
For high fee parameters ($c=0.005$ and $c=0.01$), SAMM shows a $100\%$ failure rate at low thresholds. 
For instance, with $c=0.005$, the trading fee ratios for small trades are already larger than $0.5\%$, causing them to fail the threshold due to fee alone.
Additionally large trades fail due to high slippage. 
A similar outcome occurs for $c=0.01$ at the $0.5\%$ and $1\%$ thresholds. 
As expected, a larger $c$ leads to higher fees and thus a higher failure rate at strict cost thresholds.

However, when the fee parameter is chosen appropriately for a given trade cost threshold (e.g., $c=0.003$ for all tested thresholds, $c=0.005$ for thresholds $1\%$ and $1.5\%$; $c=0.01$ for threshold $1.5\%$), SAMM demonstrates a significantly higher volume capacity than Uniswap v2. 
At the $0.5\%$ threshold, Uniswap v2 has a failure rate of $0.44\%$, whereas SAMM with $c=0.003$ has a failure rate of only $0.079\%$, over five times lower.

This result indicates that SAMM can accommodate a much larger proportion of real-world trades. 
The reason for this superior performance lies in SAMM's balance between fees and slippage. 
When a trade is small, its cost is dominated by the relatively high fee ratio, while slippage is negligible. 
When a trade is large, the fee ratio becomes smaller, and the cost is driven by increased slippage. 
This concentration reduces the extreme price movements that would otherwise cause more trades to fail the cost threshold.

\section{Analysis of Sandwich Attacks}\label{app:sandwich}

In a sandwich attack, a victim trader intends to swap some \textit{token~A} to buy \textit{token~B}.
The attacker first buys some \textit{token~B} in front of the victim, thereby driving up the price of \textit{token~B}.
Then, the victim's transaction is executed: the victim trader buys \textit{token~B} at a higher price, further increasing the price of \textit{token~B}.
Finally, the attacker sells all the \textit{token~B} bought in the first step at the elevated price, earning more \textit{token~A} than she initially paid.

We aim to confirm that reduced liquidity due to sharding does not incentivize sandwich attacks, both with and without a set maximum price.
For simplicity, we neglect trading fees in this analysis, as they are small relative to the trade volume. 
In most AMMs, fees are no larger than $0.3\%$ of the trade volume~\cite{zhang2018formal,adams2020uniswap,adams2021uniswap}. 
We also disregard gas fees, as attackers are required to issue two transactions in a sandwich attack, meaning they incur the same gas costs in both the sharded and non-sharded settings.

Unlike settings in the game-theoretic analysis (\S\ref{sec:game}), we assume that the trader wants to contribute a fixed amount of \textit{token~A} to the pool, denoted by $\poolinput{}{}{A}$. 
This assumption is crucial because, if the trader were instead aiming to get a fixed amount of \textit{token~B} without specifying a maximal price, the revenue from a sandwich attack could become arbitrarily large, independent of the liquidity amount.

We assume that the victim trader sets a maximal price for the trade, which corresponds to a minimal amount of \textit{token~B} received, denoted by $\pooloutput{\min}{}{B}$.
If $\pooloutput{\min}{}{B} =0$, it means that the victim trader does not set a maximal price.

We first study the case of a single CPMM pool with $\poolcontent{}{}{A}$ \textit{token~A} and $\poolcontent{}{}{B}$ \textit{token~B}.
We assume that initially there is no arbitrage opportunity, that is (\S\ref{sec:preliminaries:trade}),
\begin{equation}\label{eq:app:arbitrage}
	\frac{\poolcontent{}{}{A}}{\poolcontent{}{}{B}} = \frac{p^B}{p^A} \,\,.
\end{equation}

According to the CPMM protocol (Equation~\ref{leagaltrade2}), the expected amount of output \textit{token~B} for the victim trader, which is the output without the attack, is
\begin{equation}\label{eq:app:output}
	\pooloutput{0}{}{B} = \poolcontent{}{}{B} - \frac{ \poolcontent{}{}{A}\times  \poolcontent{}{}{B}}{ \poolcontent{}{}{A}+ \poolinput{}{}{A}}  = \frac{ \poolcontent{}{}{B} \times \poolinput{}{}{A}}{ \poolcontent{}{}{A}+ \poolinput{}{}{A}}
	\,\,.
\end{equation}

We simply require that the expected output is no less than the minimal output or the trade cannot be executed in the CPMM:
\begin{equation*}
	\pooloutput{0}{}{B} \geq \pooloutput{\min}{}{B}\,\,.
\end{equation*}

We denote by $s$ the \emph{slippage tolerance}, which is the ratio of the difference between the minimal output and the expected output to the expected output:
\begin{equation}\label{eq:app:tolerance}
	s = \frac{ \pooloutput{0}{}{B} - \pooloutput{\min}{}{B}}{\pooloutput{0}{}{B}}  \,\,.
\end{equation}

The maximal revenue of a sandwich attacker depends on the amount that the victim trader is willing to pay to move the price.
Heimbach and Wattenhofer~\cite[Proof of Theorem 2]{heimbach2022eliminating} show that the maximal revenue of the attacker in a sandwich attack is (measured in \textit{token~A}):

\begin{align}\label{eq:app:revenue}
	\textit{Rev} =& \frac{ \poolinput{}{}{A} \times s \times (\poolinput{}{}{A} + \poolcontent{}{}{A})}{ \poolinput{}{}{A} \times s + \poolcontent{}{}{A}} \,\,.
\end{align}

We observe that, if the victim trader sets a minimal output ($\pooloutput{\min}{}{B} > 0$), the maximal revenue of the attacker increases with the pool size~$\poolcontent{}{}{A}$; Otherwise, the maximal revenue of the attacker is independent of the pool size:

\begin{corollary}
	For a single CPMM pool without arbitrage opportunities, fixing the minimal output of the victim trader, the maximal revenue of sandwich attack increases with the pool size if a positive minimal output is set; Otherwise, the maximal revenue is independent of the pool size:
	\begin{equation*}
		\pooloutput{\min}{}{B}>0 \Rightarrow
		\frac{\partial \textit{Rev}}{\partial \poolcontent{}{}{A}} > 0 \,\,; 
		\pooloutput{\min}{}{B} = 0 \Rightarrow \frac{\partial \textit{Rev}}{\partial \poolcontent{}{}{A}} = 0 \,\,.
	\end{equation*}
\end{corollary}

We obtain this result by calculating the derivative of Equation~\ref{eq:app:revenue} with respect to $\poolcontent{}{}{A}$, as follows.
\begin{proof}
	Starting with Equation~\ref{eq:app:revenue}, the maximal revenue of the attacker is
	\begin{align*}
		\textit{Rev} \stackrel{\;\;\;\;\;\;\;}{=}& \frac{ \poolinput{}{}{A} \times s \times (\poolinput{}{}{A} + \poolcontent{}{}{A})}{ \poolinput{}{}{A} \times s + \poolcontent{}{}{A}}  \\
		\stackrel{(\ref{eq:app:tolerance})}{=}& \frac{ \poolinput{}{}{A} \times \frac{ \pooloutput{0}{}{B} - \pooloutput{\min}{}{B}}{\pooloutput{0}{}{B}} \times (\poolinput{}{}{A} + \poolcontent{}{}{A})}{ \poolinput{}{}{A} \times \frac{ \pooloutput{0}{}{B} - \pooloutput{\min}{}{B}}{\pooloutput{0}{}{B}} + \poolcontent{}{}{A}} \\
		\stackrel{\;\;\;\;\;\;\;}{=}& \frac{ \poolinput{}{}{A} \times (\pooloutput{0}{}{B} - \pooloutput{\min}{}{B}) \times (\poolinput{}{}{A} + \poolcontent{}{}{A})}{ \poolinput{}{}{A} \times (\pooloutput{0}{}{B} - \pooloutput{\min}{}{B}) + \poolcontent{}{}{A} \times \pooloutput{0}{}{B}} 
		\\
		\stackrel{(\ref{eq:app:output})}{=}& \frac{ \poolinput{}{}{A} \times (\frac{ \poolcontent{}{}{B} \times \poolinput{}{}{A}}{ \poolcontent{}{}{A}+ \poolinput{}{}{A}} - \pooloutput{\min}{}{B}) \times (\poolinput{}{}{A} + \poolcontent{}{}{A})}{ \poolinput{}{}{A} \times (\frac{ \poolcontent{}{}{B} \times \poolinput{}{}{A}}{ \poolcontent{}{}{A}+ \poolinput{}{}{A}} - \pooloutput{\min}{}{B}) + \poolcontent{}{}{A} \times \frac{ \poolcontent{}{}{B} \times \poolinput{}{}{A}}{ \poolcontent{}{}{A}+ \poolinput{}{}{A}}} 
		\\
		\stackrel{(\ref{eq:app:arbitrage})}{=}& 
		\frac{ \poolinput{}{}{A} \times (\frac{ \frac{p^A}{p^B}\poolcontent{}{}{A} \times \poolinput{}{}{A}}{ \poolcontent{}{}{A}+ \poolinput{}{}{A}} - \pooloutput{\min}{}{B}) \times (\poolinput{}{}{A} + \poolcontent{}{}{A})}{ \poolinput{}{}{A} \times (\frac{ \frac{p^A}{p^B}\poolcontent{}{}{A} \times \poolinput{}{}{A}}{ \poolcontent{}{}{A}+ \poolinput{}{}{A}} - \pooloutput{\min}{}{B}) + \poolcontent{}{}{A} \times \frac{ \frac{p^A}{p^B}\poolcontent{}{}{A} \times \poolinput{}{}{A}}{ \poolcontent{}{}{A}+ \poolinput{}{}{A}}}
		\\
		\stackrel{\;\;\;\;\;\;\;}{=}&\frac{ \poolinput{}{}{A} \times (\frac{ \frac{p^A}{p^B}\poolcontent{}{}{A} \times \poolinput{}{}{A}}{ \poolcontent{}{}{A}+ \poolinput{}{}{A}} - \pooloutput{\min}{}{B}) \times (\poolinput{}{}{A} + \poolcontent{}{}{A})}{ \poolinput{}{}{A}\times\left(  (\frac{ \frac{p^A}{p^B}\poolcontent{}{}{A} \times \poolinput{}{}{A}}{ \poolcontent{}{}{A}+ \poolinput{}{}{A}} - \pooloutput{\min}{}{B}) + \poolcontent{}{}{A} \times \frac{ \frac{p^A}{p^B}\poolcontent{}{}{A}}{ \poolcontent{}{}{A}+ \poolinput{}{}{A}}\right)}
		\\
		\stackrel{\;\;\;\;\;\;\;}{=}&\frac{ (\frac{p^A}{p^B}\poolinput{}{}{A}-\pooloutput{\min}{}{B})\poolcontent{}{}{A} - \poolinput{}{}{A}\times\pooloutput{\min}{}{B}}{\frac{p^A}{p^B}\poolcontent{}{}{A} - \pooloutput{\min}{}{B}}
		\\
		\stackrel{\;\;\;\;\;\;\;}{=}&\poolinput{}{}{A} - \frac{p^B}{p^A}\pooloutput{\min}{}{B} - \frac{(\frac{p^B}{p^A}\pooloutput{\min}{}{B})^2}{\poolcontent{}{}{A} - \frac{p^B}{p^A}\pooloutput{\min}{}{B}} \,\,.
	\end{align*}
	Then, the derivative of $\textit{Rev}$ with respect to $\poolcontent{}{}{A}$ is
	\begin{equation*}
		\frac{\partial \textit{Rev}}{\partial \poolcontent{}{}{A}} = \frac{(\frac{p^B}{p^A}\pooloutput{\min}{}{B})^2}{(\poolcontent{}{}{A} - \frac{p^B}{p^A}\pooloutput{\min}{}{B})^2}.
	\end{equation*}
	Therefore, when $\pooloutput{\min}{}{B} > 0$, $\frac{\partial \textit{Rev}}{\partial \poolcontent{}{}{A}} > 0$, and when $\pooloutput{\min}{}{B} = 0$,~${\frac{\partial \textit{Rev}}{\partial \poolcontent{}{}{A}} = 0}$.
\end{proof}
The above Corollary indicates that the maximal revenue of the attacker is non-increasing as the liquidity decreases.
Thus, neglecting trading fees, less liquidity in each SAMM shard does not worsen sandwich attacks.

\fi 

\end{document}